\documentclass[twocolumn,notitlepage,nofootinbib,floatfix,superscriptaddress]{revtex4-2}
\usepackage{rayyli}
\usepackage{pgfplots}

\usepackage{xurl} 

\definecolor{MellowGreenRGB}{RGB}{173, 204, 153} 

\pgfplotsset{compat=1.17}
\usetikzlibrary{patterns}

\newcommand{\mc}{\mathcal}
\newcommand{\row}{\operatorname{row}}
\newcommand{\col}{\operatorname{col}}
\newcommand{\code}[1]{$[\![ #1 ]\!]$}
\DeclareMathOperator{\CSS}{{\mathsf{CSS}}}
\DeclareMathOperator{\stab}{{\mathsf{Stab}}}

\begin{document}
\title{Check-weight-constrained quantum codes: Bounds and examples}

\author{Lily Wang}
\thanks{These authors contributed equally to this work.}
\affiliation{CSE Division, University of Michigan, Ann Arbor, MI, USA}
\author{Andy Zeyi Liu}
\thanks{These authors contributed equally to this work.}
\affiliation{Yale Quantum Institute \& Department of Applied Physics, Yale University, New Haven, CT, USA}
\author{Ray Li}
\affiliation{Math \& CS Department, Santa Clara University, Santa Clara, CA, USA}
\author{Aleksander Kubica}
\affiliation{Yale Quantum Institute \& Department of Applied Physics, Yale University, New Haven, CT, USA}
\author{Shouzhen Gu}
\affiliation{Yale Quantum Institute \& Department of Applied Physics, Yale University, New Haven, CT, USA}

\begin{abstract}

Quantum low-density parity-check (qLDPC) codes can be implemented by measuring only low-weight checks, making them compatible with noisy quantum hardware and central to the quest to build noise-resilient quantum computers.
A fundamental open question is how constraints on check weight limit the achievable parameters of qLDPC codes.
Here, we study stabilizer and subsystem codes with constrained check weight, combining analytical arguments with numerical optimization to establish strong upper bounds on their parameters.
We show that stabilizer codes with checks of weight at most three cannot have nontrivial distance.
We also prove tight tradeoffs between rate and distance for broad families of CSS stabilizer and subsystem codes with checks of weight at most four and two, respectively.
Notably, our bounds are applicable to general qLDPC codes, as they rely only on check-weight constraints without assuming geometric locality or special graph connectivity.
In the finite-size regime, we derive numerical upper bounds using linear programming techniques and identify explicit code constructions that approach these limits, delineating the landscape of practically relevant qLDPC codes with tens or hundreds of physical qubits.

\end{abstract}

\maketitle

\section{Introduction}

One of the fundamental differences between classical and quantum error correction arises from the fact that measurement generally disturbs the state of a quantum system~\cite{NielsenChuang10,Raussendorf2012}.
This distinction is most clearly manifested in how information about errors is extracted.
In classical codes, the error syndrome can be obtained by directly reading out the values of individual bits.
In contrast, in quantum error correction one must infer the error syndrome through carefully designed measurements of multi-qubit operators (also referred to as checks), rather than of individual qubits, in order to avoid revealing or disturbing the encoded logical information.
Such measurements are challenging to implement on noisy quantum hardware and are typically limited to checks of small weight, that is, supported on only a small number of qubits.
Consequently, quantum low-density parity-check (qLDPC) codes~\cite{BE21LDPC}, which are quantum error-correcting (QEC) codes~\cite{Shor1995decoherence,Steane1996} with checks of constant weight, play a central role in the quest to build noise-resilient quantum computers, enabling the construction of logical qubits with vanishing error rates from imperfect physical components.

A QEC code is described by its code parameters---the number of physical qubits $n$, the number of encoded logical qubits $k$ (also referred to as the dimension), and the distance $d$---as well as its check weight $w$ (defined as weight of the largest check).
The surface code~\cite{Kitaev03anyons,bravyi1998quantumcodeslatticeboundary}, one of the most important and thoroughly studied QEC codes, has the advantage of geometrically local checks, as well as check weight four, but encodes only $k=1$ logical qubit with distance $d=\sqrt n$.
In fact, all geometrically local codes in $D$ spatial dimensions have parameters constrained by the Bravyi-Poulin-Terhal-Haah (BPTH) bounds~\cite{BTbound,BPTbound,bravyi2011subsystem,HaahBound}.
If there is no constraint on the check weight, then codes having asymptotically good parameters $k,d=\Theta(n)$ can easily be constructed~\cite{CS96CSScode}.
More recently, asymptotically good qLDPC codes have also been discovered~\cite{PKAsymptoticGoodLDPC,leverrier2022quantumtannercodes,Dinur_goodqLDPCCodes,lin2022goodquantumldpccodes,hsieh2025explicitlosslessvertexexpanders}, following a line of works developing sophisticated techniques~\cite{TillichZemorHGP,hastings2020fiber,evra2020decodable,Panteleev_2022,Breuckmann2020}.
Weight reduction methods can then be used to produce a stabilizer code of check weight 5~\cite{hastings2023quantumweightreduction} or a subsystem code of check weight 3~\cite{baspin2024wire} at the cost of a constant multiplicative reduction in the code parameters of a qLDPC code.
Therefore, it seems like there is a tradeoff between the code parameters and the check weight of a quantum code.
One may also ask whether it is possible to weight-reduce further without significantly sacrificing code parameters.

\begin{figure}[ht!]
    \centering
    \raisebox{60mm}{(a)}\includegraphics[width=0.9\columnwidth,trim={5cm 0 4cm 0},clip]{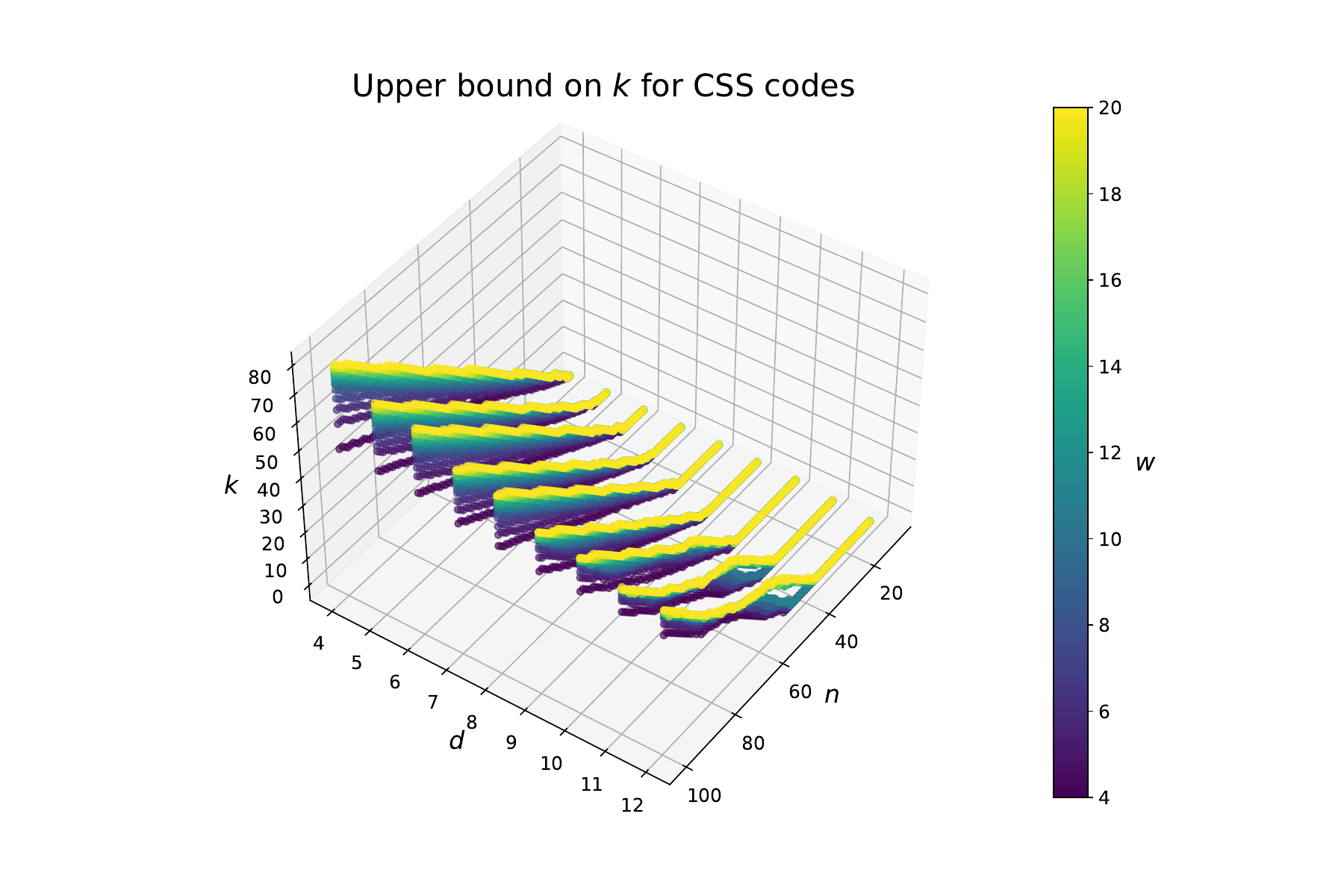}
    \raisebox{51mm}{(b)}\includegraphics[width=0.9\columnwidth]{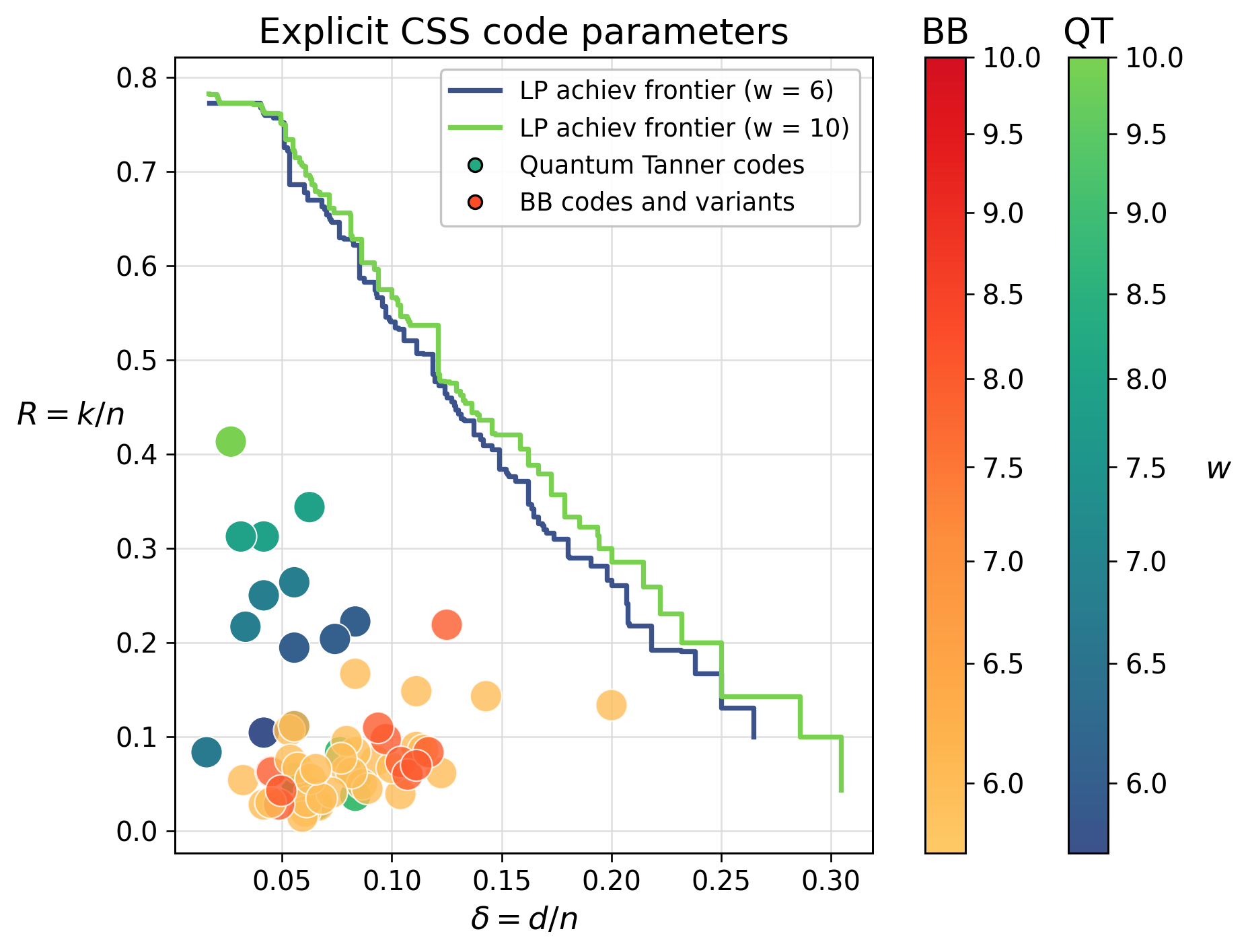}
    \caption{
    (a) Visualization of the upper bounds on code parameters of CSS codes in the $(n,k,d)$ space as a function of check weight $w$ in different colors, where $n$ and $k$ are the number of physical and logical qubits, respectively, and $d$ is the distance.
    (b) Parameters of our quantum Tanner codes as well as other explicit qLDPC codes in the literature with check weight $w\le 10$. The solid curves are the LP achievability frontiers, i.e., the maximum rates $R_{LP}(\delta)=\max\{R_i:\delta_i\geq\delta\}$ over all feasible points of at least the given relative distance (considering bounds obtained for $n\le 300$) and fixed $w$. 
    }
    \label{fig:CSSall3D}
\end{figure}

In recent years, qLDPC codes have transformed from a topic of theoretical curiosity to the one of practical relevance due to breakthrough experimental progress in demonstrating quantum error correction and fault tolerance~\cite{RyanAnderson2021,Bluvstein2023,google25QECbelow, Putterman2025, SalesRodriguez2025}.
In particular, the rise of experimental platforms with nonlocal interactions, such as movable neutral atoms~\cite{Saffman10Rydberg, Browaeys20Rydberg} or trapped ions~\cite{Cirac95trappedion, Leibfried03trappedion}, allows us to think beyond the planar quantum architecture motivated by superconducting circuits~\cite{Devoret13superconducting, Blais21cQED} and frequently considered in earlier studies.
An urgent theory question is then to bridge the gap between asymptotically good codes that typically exhibit poor constants and practically relevant codes in the regime of tens or hundreds of qubits.
What are the theoretical limits on the performance of QEC codes, and what are concrete examples of practically relevant codes that can approach these limits?

In this work, we study QEC codes with constrained check weight, approaching this timely question both analytically and numerically.
We prove bounds on the parameters of codes with small check weight.
In particular, we show that stabilizer codes with check weight 3 cannot have distance greater than 2 and a large subfamily of CSS codes with check weight 4 have parameters satisfying $kd^2=O(n)$.
Additionally, we prove that CSS subsystem codes with check weight 2 have parameters bounded by $d\le \sqrt n$ and $kd\le n$.
These bounds are tight as they can be saturated up to constant factors by the surface code and the subsystem codes in Ref.~\cite{bravyi2011subsystem}.
However, we emphasize that our bounds only use the constraint on the check weight and do not assume that the codes are geometrically local or have any special connectivity properties. This distinguishes our results from prior works~\cite{Baspin2022connectivity,Baspin25improved,Li25locality,Li25subsystem} and makes them applicable to generic qLDPC codes.
Note that there is no nontrivial bound on the asymptotic parameters of stabilizer codes with check weight 5 or subsystem codes with check weight 3 due to weight reduction~\cite{hastings2023quantumweightreduction,baspin2024wire}.

In the finite-size regime, we use linear programming (LP) techniques to numerically bound the possible parameters of codes with small check weight; see Fig.~\ref{fig:CSSall3D}(a).
We also numerically search for QEC codes with up to a few hundred qubits, identifying and optimizing explicit examples of codes with favorable parameters, and compare them with our upper bounds; see Fig.~\ref{fig:CSSall3D}(b).
We observe that higher check weight allows for better parameters, consistent with previous works~\cite{symons2025BBcovering}.
The parity-check matrices for our explicit finite code instances are available on GitHub~\cite{github}.

The rest of the paper is organized as follows.
In Sec.~\ref{sec:background}, we review background on quantum stabilizer and subsystem codes.
We then prove asymptotic bounds on stabilizer codes with check weight 3 or 4 and subsystem codes with check weight 2 in Sec.~\ref{sec:asymptoticbounds}.
The next two sections focus on codes of finite size, with our numerical upper bounds on code parameters presented in Sec.~\ref{sec:finitesizebounds} and explicit codes presented in Sec.~\ref{sec:smallcodes}.
We conclude and discuss future directions in Sec.~\ref{sec:discussion}.

\section{Background}
\label{sec:background}

We first review the definitions of classical linear codes, quantum stabilizer codes, and subsystem codes as well as present some basic facts about them.

\paragraph{Classical codes.}
A classical linear code on $n$ bits is a subspace $C\subseteq \mathbb F_2^n$. The code $C$ may be expressed as the kernel of a parity-check matrix $H\in \mathbb F_2^{m\times n}$, in which case the $m$ rows of $H$ are referred to as the checks of $C$. The dimension of the code is $k=\dim C$ and its distance $d$ is the smallest weight of a nonzero codeword, i.e., $d = \min_{x\in C\setminus \{0\}} |x|$, where $|\cdot|$ denotes the Hamming weight of a binary vector. When normalized by the code length, $R=k/n$ and $\delta=d/n$ are the rate and relative distance, respectively. The parameters of $C$ are written as $[n,k,d]$.

The dual code of $C$, denoted $C^\perp$, is defined as the subspace of vectors orthogonal to all codewords of $C$, i.e., $C^\perp = \{u\in \mathbb F_2^n: u\cdot v=0 \; \forall v\in C\}$. If $H$ is a parity-check matrix for $C$, then $C^\perp=\row(H)$.

\paragraph{Stabilizer codes.}
A stabilizer code~\cite{Gottesman96stabilizercodes} $\mc C$ is typically defined as the $+1$-eigenspace of a group of stabilizers $\mc S$, where $\mc S$ is an abelian subgroup of the Pauli group on $n$ qubits that does not contain $-I$. It is often convenient to work with a generating set of stabilizers $\{S_1, \dots S_r\}$, and we refer to the elements $S_i$ as checks. The checks of a code do not necessarily need to be independent. The weight of a Pauli operator $P$, denoted $|P|$, is the size of its support. Because our results focus on the check weights of stabilizer codes, we will adopt the convention that a stabilizer code comes equipped with a generating set of stabilizers; that is, distinct generating sets are considered different codes, even if the codespaces are identical.

The parameters of a stabilizer code are given by \code{n,k,d}, where $n$ is the number of physical qubits; $k=n-|\mc S|$ is the dimension or number of logical qubits, where $|\cdot|$ denotes the rank of a group; and $d$ is the distance, which is the smallest weight of a nontrivial logical operator.
We say that a code has check weight $w$ and qubit degree $q$ if each of its checks has weight at most $w$ and each qubit is in at most $q$ checks, respectively.
If a family of stabilizer codes has $w,q=O(1)$ as $n\to \infty$, we say that it is a qLDPC code. 
We are interested in obtaining bounds on code parameters when $w$, $q$ are fixed to be small constants.

It will be convenient to define a notion of a subset of qubits in a code that are in a fixed state.
\begin{definition}[disentangled subsystem]
    Let $\mc C$ be an $n$-qubit stabilizer code and $A$ be a subset of $m<n$ qubits such that there exist $m$ independent stabilizers supported on $A$. Then we say that $A$ is disentangled from the rest of the qubits.
\end{definition}
Note that the qubits in $A$ may be in an entangled state, but they are not entangled with the rest of the qubits.
If $\mc C$ is an \code{n,k,d} stabilizer code with a subset $A$ of $m$ disentangled qubits, then the restriction of any stabilizer to $A$ (and thus, the restriction to its complement) is a stabilizer. Therefore, removing $A$ by restricting all stabilizers to the complement of $A$ gives an \code{n-m,k,d} stabilizer code.

A special class of stabilizer codes are the Calderbank–Shor–Steane (CSS) codes~\cite{CS96CSScode,Steane96CSScode}, where all checks are of purely $X$ type or of purely $Z$ type. Alternatively, we may view CSS codes as a pair of classical codes $C_X$, $C_Z$ satisfying $C_X^\perp\subseteq C_Z$. Here, $C_X$ ($C_Z$) is the space of logical $Z$ ($X$) operators and $C_X^\perp$ ($C_Z^\perp$) is the space of $X$-type ($Z$-type) stabilizers, considering operators as a binary vector with the same support. Using the same identification, we may define the parity-check matrices $H_X$, $H_Z$ of the CSS code as the ones for the respective classical codes $C_Z$, $C_X$.
If $C_X$ and $C_Z$ have respective dimensions $k_X$ and $k_Z$, the CSS code encodes $k=k_X+k_Z-n$ logical qubits. The distance can be expressed as the minimum-weight nontrivial logical $X$ or logical $Z$ operator, given by $d = \min (d_X, d_Z)$, where $d_X = \min_{x \in C_Z \setminus C_X^\perp} |x|$ and $d_Z = \min_{x \in C_X \setminus C_Z^\perp} |x|$.

\paragraph{Subsystem codes.}
Subsystem codes~\cite{Poulin05subsystem} are a generalization of stabilizer codes where information stored in a subset of the logical qubits, called gauge qubits, is not used. This allows us to infer the values of stabilizers by measuring checks of smaller weights.
More formally, a subsystem code is defined by a gauge group $\calG$, which is a (not necessarily abelian) subgroup of the $n$-qubit Pauli group. 
We assume that $\calG$ comes equipped with a generating set of checks. 
A subsystem code is said to be CSS if all checks are of purely $X$ type or of purely $Z$ type.
The check weight $w$ and qubit degree $q$ of a subsystem code are defined with respect to the checks in the gauge group $\calG$.

The stabilizer group $\calS$ of a subsystem code is the center of $\calG$ (up to phase).
We view the stabilizer group $\calS$ as defining a code with $k+g$ pairs of logical operators $\bar X_1,\bar Z_1,\dots,\bar X_{k+g},\bar Z_{k+g}$, where the gauge group is $\calG = \ab{\calS, \bar X_{k+1},\bar Z_{k+1},\dots,\bar X_{k+g},\bar Z_{k+g}}$.
We define $k$ to be the dimension of the code, and call $g$ the number of gauge qubits.
Dimension counting gives $2g = |\calG|-|\calS|$ and $|\calG|+|\calS| + 2k = 2n$.
Stabilizer codes are a subset of subsystem codes: any stabilizer code can be viewed as a subsystem code where $\mc G=\mc S$.

For subsystem codes, we make a distinction between \emph{bare} logical operators, which act trivially on the gauge qubits, and \emph{dressed} logical operators, which may not. 
Formally, the bare logical operators are elements of the centralizer $C(\calG) = \ab{\calS, \bar X_1,\bar Z_1,\dots,\bar X_{k},\bar Z_{k}}$, and the dressed logical operators are elements of $C(\calS) = \ab{\calS, \bar X_1,\bar Z_1,\dots,\bar X_{k+g},\bar Z_{k+g}}$. 
Note that for stabilizer codes, there is no distinction. 
We say two logical operators are \emph{gauge-equivalent} if their product is a gauge operator (up to phase), so that dressed logical operators are the logical operators gauge-equivalent to a bare logical operator.
Logical operators are \emph{nontrivial} if they are outside the gauge group $\calG$.
The distance $d$ is defined to be the minimum weight of a nontrivial dressed logical operator. 
For CSS subsystem codes, $d_X$ and $d_Z$ are the minimum weight of nontrivial dressed logical operators of $X$ type and $Z$ type, respectively.

\section{Asymptotic bounds on code parameters}
\label{sec:asymptoticbounds}

In this section, we prove asymptotic bounds on the parameters of stabilizer and subsystem codes when the check weights are restricted. For stabilizer codes, we prove in Sec.~\ref{subsec:stabilizerwt3} that a code of check weight 3 supporting logical qubits has distance at most 2.\footnote{A similar result was proven in Ref.~\cite{AharanovEldar}, where it was shown using more involved techniques that a general commuting 3-local qubit Hamiltonian can be diagonalized by a constant-depth quantum circuit. Hence, the code defined by the Hamiltonian's ground space can only have constant (albeit possibly large) distance. In contrast, our approach gives an elementary proof showing a concrete distance upper bound of 2 for the restricted class of stabilizer codes. A similar result was also obtained in unpublished work of Krishna and Tillich \cite{Krishnatillich}.} This is regardless of the qubit degree of the code. When the code is allowed to have checks of weight 4, we show in Sec.~\ref{subsec:CSSwt4} that under certain assumptions the resulting code is a generalized surface code, which has parameters satisfying $kd^2=O(n)$.\footnote{In unpublished work, Krishna and Tillich \cite{Krishnatillich} also showed bounds on restricted classes of codes with check weight 4.} From quantum weight reduction~\cite{hastings2023quantumweightreduction} and the existence of asymptotically good qLDPC codes (codes where $k,d=\Theta(n)$)~\cite{PKAsymptoticGoodLDPC,leverrier2022quantumtannercodes,Dinur_goodqLDPCCodes}, it follows that there exist good CSS codes of check weight 5 and qubit degree 5.

Now considering subsystem codes, one where the checks have weight 1 cannot support logical qubits and have distance greater than 1 since every qubit is either disentangled, a gauge qubit, or the support of a logical operator. We show in Sec.~\ref{subsec:subsystemwt2} that a CSS subsystem code with check weight 2 has parameters satisfying $d\le \sqrt n$ and $kd\le n$, regardless of the qubit degree or number of checks that need to be multiplied together to form a stabilizer. This bound can be saturated up to constant factors by the codes of Ref.~\cite{bravyi2011subsystem}, which are additionally local in two dimensions and have qubit degree 4.
Asymptotically good subsystem codes of check weight 3 can be achieved using the wire codes construction~\cite{baspin2024wire}, which is a weight reduction technique mapping a stabilizer code to a subsystem code with check weight 3. When the input is a good qLDPC code, the resulting subsystem code will be good as well.

\subsection{Stabilizer codes of check weight 3}
\label{subsec:stabilizerwt3}
For any stabilizer code with check weight 3, we show that either its distance is at most two or it does not encode any logical qubits. 
We begin by proving the result for the more structured case of CSS codes, and then reduce the analysis of stabilizer codes to that of CSS codes. For both steps, we suppose codes exist which exceed these parameters and use a minimal example to derive several restrictive properties that such a code must have, eventually leading to a contradiction.
Intuitively, having only check weight 3 while requiring nontrivial distance does not allow the code to be sufficiently expanding, meaning that each additional check touches at most one new qubit, precluding any logical qubits.

\begin{definition}
    For a stabilizer code $\mc C$ with checks $\{S_i\}$, define the total check weight of $\mc C$ to be $\wt(\mc C) = \sum_i |S_i|$.
\end{definition}

Note that the total check weight of a code is well defined because our definition of a code includes the stabilizer generators in addition to the codespace.

\begin{definition}[minimal code]
    For a given class of codes $\textsf{C}$, let $n(\textsf{C})$ be the minimum number of qubits over all $\mc C \in \textsf{C}$.
    We define a minimal code $\mc C$ of the class $\textsf{C}$ to be one that minimizes the total check weight $\wt(\mc C)$ among all codes in $\textsf{C}$ with $n(\textsf{C})$ qubits.
\end{definition}

Let $\stab(3)$ and $\CSS(3)$ be the class of stabilizer and CSS codes, respectively, with check weight 3, nonzero dimension, and distance $d>2$. We will show by contradiction that $\stab(3)$ and $\CSS(3)$ are empty.
Note that if either set is nonempty, any minimal code of $\stab(3)$ or $\CSS(3)$ does not contain any subset of disentangled qubits since otherwise, there would be a smaller code with the same dimension and distance. Also, the checks of any minimal code of $\stab(3)$ or $\CSS(3)$ are independent.

Our argument is based on analyzing the structure of how checks overlap with each other.
We call two checks \emph{intersecting} if their supports share at least one qubit.
We further make the distinction between checks that intersect by acting as the same or different Pauli operators on their common support, which we call \emph{coinciding} or \emph{mismatching} checks, respectively.
For example $X_1X_2$ intersects both $X_1X_3$ and $Z_1Z_2$, coinciding with $X_1X_3$ and mismatching with $Z_1Z_2$.
The check $X_1X_2X_3$ both coincides and mismatches with $X_1Z_2Z_3$.

We first establish a few simplifications that hold for a minimal code of $\CSS(3)$.

\begin{lemma} \label{lem:C3}
    In any minimal code of $\CSS(3)$, the following properties hold.
    \begin{enumerate}
        \item\label{prop:1X1Z} Each qubit is contained in at least one $X$ check and one $Z$ check.
        \item\label{prop:intersect1} Any two $X$ checks share at most one qubit, and any two $Z$ checks share at most one qubit.
        \item\label{prop:2X2Z} Each qubit is contained in at most two $X$ checks and at most two $Z$ checks.
        \item\label{prop:intersect2} The number of unordered pairs of intersecting checks of the same type ($X$ or $Z$) is at least the number of weight-3 checks.
    \end{enumerate}
\end{lemma}
\begin{proof}
    \begin{enumerate}
        \item If not, say qubit 1 is not in any $X$ checks, then $Z_1$ is a logical operator of the code, which must be a stabilizer due to the distance. But this implies that qubit 1 is disentangled, a contradiction of minimality. A similar contradiction arises if a qubit is not in any $Z$ checks.
        \item If two checks of the same type share at least two qubits, we could replace one of them by the product of the checks to obtain a smaller total check weight, contradicting minimality.
        \item Suppose, for a contradiction, we have three $Z$ checks with supports containing qubit 1. Each of them must have support on at least one other qubit or else qubit 1 would be disentangled. By Property~\ref{prop:intersect1}, these other qubits are distinct. But then any $X$ check with support on qubit 1, which must exist by Property~\ref{prop:1X1Z}, must have weight at least 4 by commutation with the $Z$ checks. Similarly, it is not possible for a qubit to be contained in three $X$ checks.
        \item By double counting, it suffices to associate every weight-3 check with two \emph{ordered} pairs of intersecting checks of same type ($X$ or $Z$) such that each ordered pair is associated with at most one weight-3 check.
    Let $S_1=Z_i Z_j Z_k$ be a $Z$ check, as the case of $X$ checks is handled similarly. This check must coincide with at least one other $Z$ check. Otherwise, $X_iX_j$ and $X_iX_k$ are logical operators, which must be stabilizers because of the distance of the code. But then there would be three independent stabilizers supported on the three qubits $i$, $j$, and $k$, which disentangles them.
    If there are at least two other $Z$ checks $S_2$, $S_3$ coinciding with $S_1$, we associate $S_1$ with the pairs $(S_1, S_2)$ and $(S_1, S_3)$.
    
    Otherwise, suppose there is only one other $Z$ check $S_2$ coinciding with $S_1$, say on qubit $k$.
    There are no other $Z$ checks supported on qubits $i$ and $j$ by Property~\ref{prop:intersect1}, so $S_3=X_iX_j$ is a logical operator of the code, which must be a stabilizer due to the code's distance. Furthermore, $S_3$ is a check, as it cannot be a product of weight-3 checks by minimality and no other weight-2 $X$ checks can have support on qubit $i$ or $j$ due to commutation with $S_1$ and $S_2$. Now $S_3$ must coincide with another $X$ check $S_4$ or else $Z_iZ_j$ would be a stabilizer and qubits $i$ and $j$ would be disentangled. Then we associate $S_1$ with the pairs $(S_1, S_2)$ and $(S_3, S_4)$.
    Note that $(S_3, S_4)$ cannot be associated with another weight-3 check because the first component $S_3$ has weight two and no other weight-3 $X$ check has support on qubit $i$ or $j$.
    \end{enumerate}
\end{proof}

\begin{theorem}\label{thm:css_wt3}
    An \code{n,k,d} CSS code with check weight 3 must have either distance $d \leq 2$ or dimension $k = 0$.
\end{theorem}
\begin{proof}
    Assuming such a code exists, consider any minimal code of $\CSS(3)$. Let $r_2$ and $r_3$ be the number of weight-2 and weight-3 checks, respectively, and let $r=r_2+r_3$. 

    Because each qubit is in at least one $X$ check and one $Z$ check (Lemma~\ref{lem:C3}.\ref{prop:1X1Z}), checks of the same type share at most one qubit (Lemma~\ref{lem:C3}.\ref{prop:intersect1}), and each qubit is not in more than two $X$ or two $Z$ checks (Lemma~\ref{lem:C3}.\ref{prop:2X2Z}), we can count the total number of qubits as
    \begin{equation}
        2n = 3r_3+2r_2-t,
    \end{equation}
    where $t$ is the number of unordered pairs of intersecting checks of the same type.
    By Lemma~\ref{lem:C3}.\ref{prop:intersect2}, $t\ge r_3$. Therefore,
    \begin{equation}
        2n \le 3r_3+2r_2-r_3 = 2r,
    \end{equation}
    which shows that $k=0$ as all checks are independent. This is a contradiction, so $\CSS(3)=\emptyset$.
\end{proof}

We now handle the case of general stabilizer codes, showing that $\stab(3)=\emptyset$ by contradiction. Our strategy is to show that a minimal code in $\stab(3)$ is isomorphic to a CSS code via single-qubit Clifford deformations. Because code parameters are preserved under the isomorphism and Theorem~\ref{thm:css_wt3} established that $\CSS(3)=\emptyset$, this will imply that $\stab(3)=\emptyset$ as well. The reduction to a CSS code is in two parts. We first establish that each qubit in a code in $\stab(3)$ only has two out of the three Pauli operators acting on it. Thus, it is possible to apply local Clifford rotations so that checks comprise only $X$ or $Z$ operators. Then, we show that there is a globally consistent choice of rotations so that every check is transformed to a pure $X$ or pure $Z$ Pauli operator.

\begin{lemma}\label{lem:S3}
    In any minimal code of $\stab(3)$, the following properties hold.
    \begin{enumerate}
        \item\label{prop:coincideintersect} Two checks cannot coincide on more than one qubit, and they cannot both coincide and mismatch.
        \item \label{prop:stabwt3} If three checks act on a given qubit as $X$, $Y$, and $Z$, respectively, then they each have weight exactly three.
        \item \label{prop:2of3Paulis} For any qubit, the checks act on it as only two of the three Pauli operators.
    \end{enumerate}
\end{lemma}
\begin{proof}
    \begin{enumerate}
        \item If two checks $S_1$ and $S_2$ coincide on more than one qubit, we could replace $S_1$ by $S_1S_2$, which has lower weight, contradicting minimality of the code.
    If $S_1$ and $S_2$ both coincide and mismatch they must, by commutativity, be of the form $S_1=P_iP_jP_k$ and $S_2=P_iQ_jQ_k$ where $i$, $j$, $k$ are distinct and $P_j\ne Q_j$, $P_k\ne Q_k$. But then we could replace $S_1$ with $S_1S_2$, which has weight two, contradicting minimality of the code.
    \item Let $S_1$, $S_2$, and $S_3$ be three checks acting on a qubit, say qubit 1, as $X$, $Y$, and $Z$, respectively.
    If two of the checks have weight two, they must have the same support by commutativity, but this is not possible because those two qubits would be disentangled.
    If exactly one of the checks, say $S_1$, has weight two, then $S_1=X_1P_i$, $S_2=Y_1P'_iQ_j$, and $S_3=Z_1P''_iR_k$, where $P_i$, $P'_i$, and $P''_i$ are pairwise distinct by Property~\ref{prop:coincideintersect}. But then we could replace $S_2$ with $S_1S_2S_3 = Q_j R_k$, which has weight two, contradicting minimality of the code.
    \item We proceed by contradiction. Our strategy is to analyze low-weight operators that commute with known stabilizers, which implies the existence of another stabilizer that anticommutes with the operator by the distance and minimality of the code. By considering different cases, we show that there always exist enough independent stabilizers supported on a small number of qubits to disentangle those qubits. To reduce the casework, we permute qubits or relabel Pauli operators using single-qubit Clifford rotations without loss of generality.

    Suppose qubit 1 has checks acting on it as the three different Pauli operators, say $S_1$, $S_2$, and $S_3$. By Property~\ref{prop:stabwt3} and without loss of generality, $S_1=X_1X_2X_3$, $S_2=Z_1Z_2X_4$, and either $S_3=Y_1 Y_2 X_5$ or $S_3=Y_1 Z_3 Z_4$. If $S_3=Y_1Y_2X_5$, then $X_1 X_2$ commutes with all three checks. It cannot be a stabilizer or else qubit 3 is disentangled, so there must be another check $S_4$ which anticommutes with it, say containing $Z_1$ without loss of generality. The check $S_4$ must commute with $S_1$ and $S_3$, but its support cannot overlap more with that of $S_2$ by Property~\ref{prop:coincideintersect}, so $S_4=Z_1 Z_3 Z_5$, up to Clifford rotation. But then $\{S_1, S_3, S_4\}$ is equivalent to $\{S_1, S_2, Y_1Z_3Z_4\}$ after relabeling qubits and applying Clifford rotations. This allows us to only consider the second case where $S_3=Y_1 Z_3 Z_4$.
    
    In the second case, there are checks $S_1=X_1X_2X_3$, $S_2=Z_1Z_2X_4$, and $S_3=Y_1 Z_3 Z_4$. Consider the operator $L=X_3X_4$, which commutes with all three checks.
    This operator cannot be a stabilizer because otherwise, the four independent stabilizers $S_1, S_2, S_3, L$ would disentangle the first four qubits.
    It also cannot be a nontrivial logical operator because of the distance of the code. Therefore, there must be a check $S_4$ that anticommutes with $L$. Without loss of generality, $S_4$ acts on qubit 3 as $Y_3$ or $Z_3$.
    
    In the case that $S_4$ contains $Y_3$, it mismatches with the $S_1$ and $S_3$, so its support includes either qubits 2 and 4 or qubit 1. The former is not possible because the first four qubits would support four checks and be a disentangled subsystem. 
    In the latter case, $S_4$ contains $Z_1$ to commute with $S_1$ and $S_3$, as its support cannot contain qubit 2 or 4 by Property~\ref{prop:coincideintersect}. The support of $S_4$ contains another qubit or else the first four qubits would be disentangled, so $S_4=Z_1Y_3X_5$ without loss of generality.

    Now $X_5$ commutes with all checks so far and it cannot be a stabilizer or else qubit 5 is disentangled. Thus, there is a check $S_5$ that anticommutes with it. Since $S_5$ commutes with $S_4=Z_1Y_3X_5$, its support includes qubit 5 and either qubit 1 or qubit 3. Note that the first three checks contain at least two different Pauli operators acting on qubit 1 and on qubit 3, and their supports do not contain qubit 5, so $S_5$ must have support on another qubit among the first four. Therefore, the support of $S_5$ is contained in the first five qubits, so those qubits are disentangled.
    
    In the case that $S_4$ contains $Z_3$, it coincides with $S_3$.
    By Property~\ref{prop:coincideintersect}, $S_3$ acts trivially on both qubits 1 and 4. Moreover, to commute with both $S_1$ and $S_2$, it must contain $Z_2$. To avoid disentangling the first four qubits, $S_4$ acts on an additional qubit, so $S_4=Z_2Z_3X_5$ without loss of generality. The analysis is then similar to the previous case: the operator $X_5$ cannot be a logical operator, so there must be an anticommuting check $S_5$ whose support is contained in the first five qubits, disentangling the subsystem.
    \end{enumerate}
\end{proof}

Now, we are ready to prove the reduction of stabilizer codes in $\stab(3)$ to CSS codes.

\begin{lemma}\label{lem:stabilizer2CSSwt3}
    Any minimal code of $\stab(3)$ is isomorphic under single-qubit Clifford operations to a code in $\CSS(3)$.
\end{lemma}
\begin{proof}
    By Lemma~\ref{lem:S3}.\ref{prop:2of3Paulis}, each qubit only has two types of Pauli operators acting on it over all checks. By Clifford rotations, we can write all checks as consisting of (possibly mixed) $X$ and $Z$ operators. Define a graph $G = (V,E)$ where each vertex is a check and two vertices are adjacent if they mismatch.

    We claim $G$ has no odd cycles. For a contradiction, consider any shortest odd cycle $S_0, S_1, \ldots S_m=S_0$. It follows that this cycle cannot have any chords (edges between nonadjacent vertices), or else there would be a shorter odd cycle. All checks in the cycle have weight exactly 3, since any check mismatches with two other checks on two qubits each, and a check of weight 2 would imply its neighboring checks coincide on two qubits, contradicting Lemma~\ref{lem:S3}.\ref{prop:coincideintersect}. 
    
    Define $A$ as the set of qubits in the support of any check in the cycle. Let $i$ be any qubit in $A$. We will show $i$ is in at least three checks of the cycle by contradiction. Suppose qubit $i$ is in exactly two checks from the cycle which coincide on $i$. Fix one of the checks as $S_1 = P_i P_j P_k$. The two neighboring checks of $S_1$ must mismatch with $S_1$ on qubits $j$ and $k$ since there are no more checks in the cycle supported on qubit $i$. But then those two checks coincide with each other on two qubits, contradicting Lemma~\ref{lem:S3}.\ref{prop:coincideintersect}. The case with qubit $i$ being in only one check from the cycle is similarly not possible.

    Lastly, suppose qubit $i$ is in the support of exactly two checks $S_1$, $S_2$ from the cycle which mismatch on qubit $i$ (and another qubit $j$). Since the cycle has no chords, $S_1$ and $S_2$ are consecutive in the cycle. Without loss of generality, let $S_1$ act on qubit $j$ as $X_j$ and $S_2$ act on qubit $j$ as $Z_j$. From the cycle, each check mismatches with another check, $S_1$ with $S_0$ and $S_2$ with $S_3$. Because there are no more checks acting on qubit $i$, $S_0$ and $S_3$ must both act on qubit $j$ to mismatch with $S_1$ and $S_2$, respectively, with $S_0$ acting on qubit $j$ as $Z_j$ and $S_3$ acting on qubit $j$ as $X_j$. Therefore, $S_0$ and $S_3$ must be distinct and mismatch with each other. Because the cycle has no chords, $S_0, S_1, S_2, S_3$ is the entire cycle, which is even, a contradiction.

    Any qubit in $A$ is therefore in at least three checks, but since each check has weight 3, we have at least as many checks supported on $A$ as number of qubits. Thus, $A$ is a disentangled subsystem, which contradicts the minimality of the code.

    Therefore, $G$ has no odd cycles, and thus has a bipartition $V_1 \sqcup V_2$. We apply Hadamard rotations so that the checks of $V_1$ are of pure $X$ type, which is possible since the checks of $V_1$ do not mismatch with each other.
    We show that all checks in $V_2$ are of pure $Z$ type, demonstrating that the resulting code is CSS. Indeed, if $S_1\in V_2$ acts as $X$ on qubit $i$, no check acts on qubit $i$ as $Z$ because $V_1$ only contains $X$ checks and the other checks of $V_2$ do not mismatch with $S_1$. But then $X_i$ would be a logical operator, disentangling qubit $i$ if it is a stabilizer or contradicting the distance if it is nontrivial.
\end{proof}

Combining Lemma~\ref{lem:stabilizer2CSSwt3} with Theorem~\ref{thm:css_wt3} gives the desired result. 

\begin{theorem}
    \label{thm:stabwt3}
    A stabilizer code with check weight 3 must have either distance $d \leq 2$ or dimension $k = 0$.
\end{theorem}

\begin{proof}
    If a stabilizer code with check weight 3, distance $d> 2$, and $k>0$ exists, the minimal one is isomorphic to a CSS code by Lemma~\ref{lem:stabilizer2CSSwt3}, but $\CSS(3)=\emptyset$ by Theorem~\ref{thm:css_wt3}, so this is not possible.
\end{proof}

\subsection{CSS codes of check weight 4}
\label{subsec:CSSwt4}
For stabilizer codes with check weight 4, we consider the subclass of CSS codes $\CSS^*(4)$ such that each qubit is in exactly two $X$ checks and exactly two $Z$ checks.
Note that the codes in $\CSS^*(4)$ can have linearly dependent checks; otherwise, they would not have any logical qubits by simple dimension counting.
We show that any code in $\CSS^*(4)$ is a generalized surface code (also called homological code on a surface)~\cite{Kitaev03anyons,Bombin07homological} and thus has parameters satisfying $kd^2=O(n)$.

\begin{definition}[generalized surface code]
    A generalized surface code $\mc C$ is a CSS code defined by a cellulation of a closed manifold $M$ with vertices $V$, edges $E$, and faces $F$. The qubits of $\mc C$ are placed on $E$. The $X$ checks $S_v$ are associated with vertices $v\in V$, where the support of $S_v$ is the set of qubits on the edges incident to $v$. The $Z$ checks $S_f$ are associated with faces $f\in F$, where the support of $S_f$ is the set of qubits on the edges forming the boundary of $f$.
\end{definition}

\begin{theorem}
    \label{thm:wt4CSSdeg2}
    Let $\mc C\in \CSS^*(4)$ be an \code{n,k,d} CSS code with $k>0$. Then $kd^2\le An$, where $A$ is a universal constant.
\end{theorem}

To show that any code $\CSS^*(4)$ is a generalized surface code, we first simplify the analysis by excluding a certain structure of overlapping checks.

\begin{lemma}
    \label{lem:4overlap}
    Let $\mc C\in \CSS^*(4)$ be an \code{n,k,d>2} code that has an $X$ check and a $Z$ check that overlap on four qubits. Then there exist a disentangled set of two qubits. Consequently, the rest of the qubits forms a code $\mc C'\in \CSS^*(4)$ with parameters \code{n-2,k,d}.
\end{lemma}
\begin{proof}
    Without loss of generality, let the overlapping checks be $S_X=X_1X_2X_3X_4$ and $S_Z=Z_1Z_2Z_3Z_4$. There must be another $X$ check $S_X'$ with support on qubit 1. By commutation with $S_Z$, it must overlap with $S_Z$ on another qubit, say qubit 2. Because there are no other $X$ checks touching qubits 1 or 2, $Z_1Z_2$ is a logical operator, which must be a stabilizer because $d>2$.

    We claim that $X_1X_2$ is also a stabilizer. Consider the $Z$ checks with support on qubits $1$, $2$, $3$, or $4$. By commutation with $S_X$ and the fact that each qubit is in two $Z$ checks, there can be at most two checks other than $S_Z$. If there is only one, then it is again $S_Z'=Z_1Z_2Z_3Z_4$. If $S_Z'$, $S_Z''$ are two $Z$ checks, each with weight-2 support on $\{1,2,3,4\}$, then in order to generate the stabilizer $Z_1Z_2$, one of their supports must be $\{1,2\}$ when restricted to $\{1,2,3,4\}$.
    In either case, $X_1X_2$ is a logical operator since all $Z$ checks touching qubit 1 also touch qubit 2, which is a stabilizer since $d>2$.

    Because $X_1X_2$ and $Z_1Z_2$ are stabilizers, qubits 1 and 2 are disentangled from the rest of the code. Removing them gives a new code $\mc C'$ with parameters \code{n-2,k,d}. Importantly, $\mc C'$ has check weight 4, and each remaining qubit is still in exactly two $X$ checks and two $Z$ checks.
\end{proof}

\begin{lemma}
    \label{lem:wt4homologicalcode}
    Let $\mc C\in \CSS^*(4)$ be an \code{n,k,d>2} code such that there does not exist any disentangled set of two qubits. Then $\mc C$ (or another code in $\CSS^*(4)$ giving the same codespace) is a generalized surface code.
\end{lemma}

\begin{proof}
    Let $G = (V, E)$ be the multigraph\footnote{There may be multiple edges connecting the same two vertices.} with the $X$ checks as vertices and qubits as edges, where the support of an $X$ check is the incident edges. The graph is well-defined because every qubit is in exactly two $X$ checks. The support of every $Z$ check $S$ is a subset of edges $E_S\subseteq E$. Because $S$ commutes with every $X$ check, each vertex in $V$ is incident to an even number of edges in $E_S$. Since $|S|\le 4$ and by Lemma~\ref{lem:4overlap}, this number must be 0 or 2, which means $E_S$ is a union of disjoint cycles.
    
    We show that $E_S$ must in fact be a single cycle. Suppose instead that $E_S = E_S'\sqcup E_S''$ is the disjoint union of two length-two cycles, say corresponding to qubits $1$, $2$ and qubits $3$, $4$, respectively. Then there exist $X$ checks $S_X^1$, $S_X^2$ with support containing qubits $1$, $2$ and $X$ checks $S_X^3$, $S_X^4$ containing qubits $3$, $4$. Because $S_X^1$, $S_X^2$ are the only $X$ checks touching qubits $1$ and $2$, the operator $Z_1Z_2$ is a logical operator, which must be a stabilizer if $d>2$. Similarly, $Z_3Z_4$ is a stabilizer. Therefore, we can replace the check $S=Z_1Z_2Z_3Z_4$ with the two checks $Z_1Z_2$ and $Z_3Z_4$ without changing the codespace of $\mc C$, and each qubit will still be in exactly two $Z$ checks. Therefore, we assume for the remainder of the proof that the edges $E_S$ corresponding to a $Z$ check is a single cycle.
    
    Now, to each $Z$ check, we associate an $m$-gon, $m=2, 3, 4$, to the corresponding cycle, where its edges are identified with edges in $E$. By gluing the different polygons together along their edges in the appropriate direction when they are identified with the same edge in $E$, we obtain a closed manifold $M$. This uses the fact that every qubit is in exactly two $Z$ checks so that every edge of every polygon is glued with an edge of another polygon. The graph $G$ has a natural embedding on $M$, as its edges are the boundaries of polygons which were glued together, and the faces of the embedding $F$ are the polygons~\cite{Stillwell93}. Thus, $\mc C$ is a generalized surface code on $M$.
\end{proof}

See Fig.~\ref{fig:manifold_ex} for an illustration of the proof of Lemma~\ref{lem:wt4homologicalcode}.

\begin{figure}[htpb]
	\centering
    \hspace{-3em}
    \begin{minipage}{0.45\columnwidth}
        \raisebox{23mm}{(a)}\includegraphics[width=0.62\textwidth]{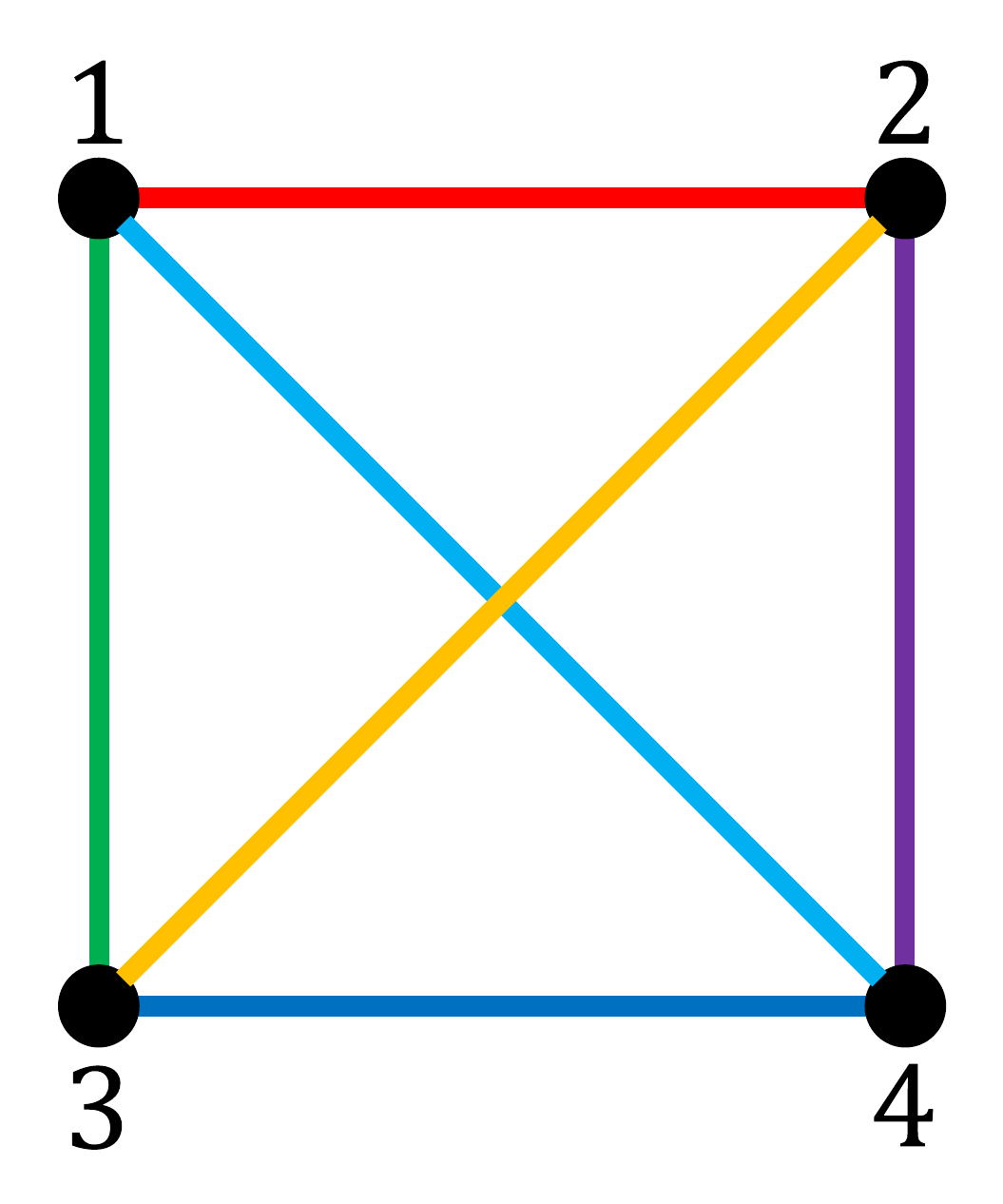}\\
        \raisebox{24mm}{(c)}\includegraphics[width=0.62\textwidth]{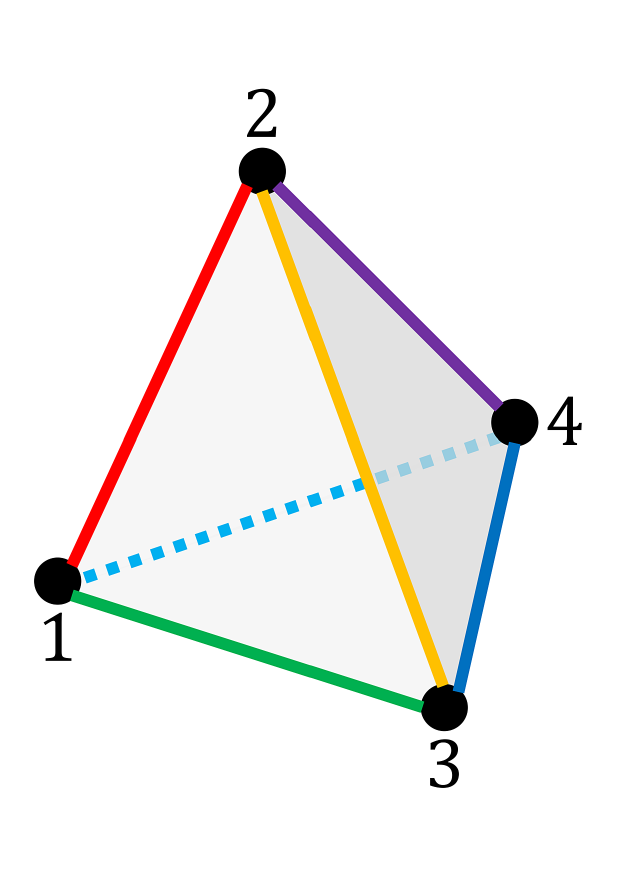}
    \end{minipage}%
    \begin{minipage}{0.45\columnwidth}
        \raisebox{53mm}{(b)}\includegraphics[width=1.0\textwidth]{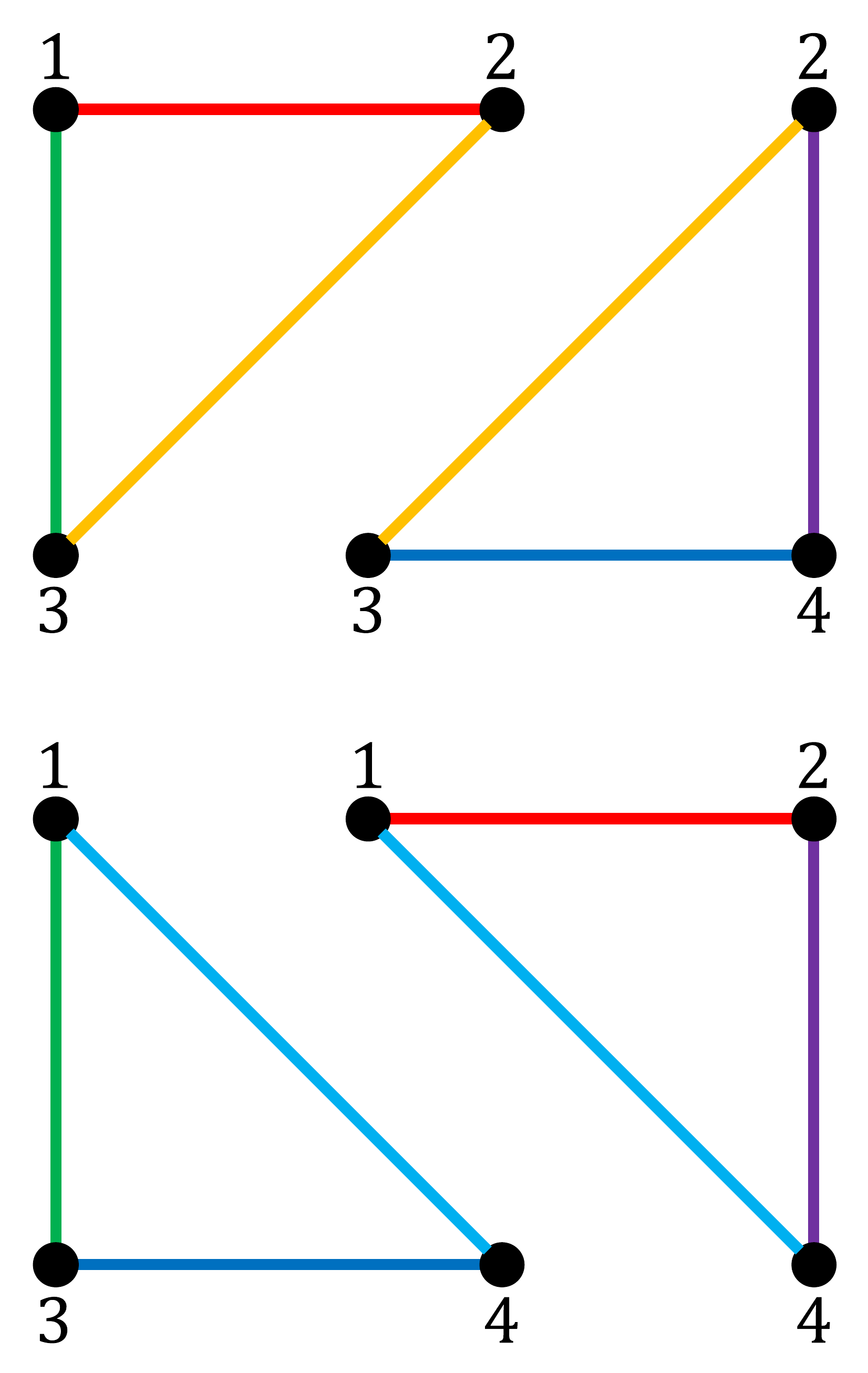}
    \end{minipage}
	\caption{Proof of Lemma~\ref{lem:wt4homologicalcode}. (a) The code $\mc C$ is placed on a graph with the qubits associated with the different colored edges. (b) The $Z$ checks are the qubits along the cycles $123$, $234$, $134$, and $124$. (c) By gluing the different cycles together in the correct orientation where they share the same edge, we obtain a closed surface (a sphere in this case) on which $\mc C$ is a generalized surface code.
	}
	\label{fig:manifold_ex}
\end{figure}

Next, using a bound by Fetaya~\cite{Fetaya12}, we show that generalized surface codes from cellulations with sufficiently low vertex and face degrees (the number of edges incident to each vertex or face) have restricted parameters $kd^2=O(n)$.

\begin{theorem}[Theorem 1.3 of Ref.~\cite{Fetaya12}]
    \label{thm:fetayabound}
    For any $w$, there exists a constant $A$ such that any generalized surface code encoding at least one logical qubit and defined by a cellulation with vertex and face degrees bounded by $w$ has parameters \code{n,k,d} satisfying
    \begin{equation}
        d^2\le An.
    \end{equation}
\end{theorem}

\begin{lemma}
    \label{lem:homologicalcodebound}
    Let $\mc C$ be an \code{n,k>0,d} generalized surface code on a closed surface cellulated so that the degree of every vertex and face is at most four. Then $kd^2\le An$, where $A$ is a universal constant.
\end{lemma}

\begin{proof}
    Let $M$ be the surface with cellulation $(V,E,F)$. Then the code $\mc C$ has $|E|$ qubits, $|V|$ $X$ checks, and $|F|$ $Z$ checks. By assumption of the vertex and face degrees, $2|E|\le 4|V|$ and $2|E|\le 4|F|$. First, suppose $M$ is connected. Then there are exactly two redundant checks. This is because the product of all $X$ checks associated with a nonempty subset of vertices $V'\subseteq V$ has support on the edges connecting $V'$ and $V\setminus V'$, which is empty if and only if $V'=V$. Similarly, the only redundancy in the $Z$ checks is that the product of all of them is the identity.
    Therefore, the number of encoded logical qubits is
    \begin{equation}
        k = |E| - |V| - |F| + 2 \le 2.
    \end{equation}
    By Theorem~\ref{thm:fetayabound}, there exists a constant $A'$ such that $d^2\le A'n$.

    If $M$ is the disjoint union of $m$ connected manifolds $\{M_i\}_{i=1}^m$, the code $\mc C$ is the direct sum of $m$ disjoint codes $\{C_i\}_{i=1}^m$. If each code has parameters \code{n_i,k_i,d_i}, then the previous argument gives $k_i\le 2$ and $d_i^2\le A'n_i$ for all $i$ where $k_i>0$. The parameters of the combined code are related by
    \begin{align}
        n &= \sum_{i=1}^m n_i,\\
        k &= \sum_{i=1}^m k_i,\\
        d &= \min \{d_1, \dots, d_m\},
    \end{align}
    where we take the convention that $d_i=\infty$ if $k_i=0$. Note that $d<\infty$ since at least one $k_i>0$.
    Thus,
    \begin{align}
        kd^2 &= \sum_{i=1}^m k_i\min\{d_1, \dots, d_m\}^2\\
        &\le \sum_{\substack{i=1\\k_i>0}}^m k_id_i^2\\
        &\le \sum_{\substack{i=1\\k_i>0}}^m 2A'n_i\\
        &\le 2A'n.
    \end{align}
    Therefore, the theorem holds with $A=2A'$.
\end{proof}

We remark that hyperbolic surface codes~\cite{FML02,Delfosse13} are not a counterexample to Lemma~\ref{lem:homologicalcodebound} because the negative curvature of the surface would necessitate a vertex or face degree to be greater than four.

Combining Lemmas~\ref{lem:wt4homologicalcode} and~\ref{lem:homologicalcodebound} gives the desired result.

\begin{proof}[Proof of Theorem~\ref{thm:wt4CSSdeg2}]
    If $\mc C$ does not contain any disentangled set of two qubits, Lemma~\ref{lem:wt4homologicalcode} implies $\mc C$ is a generalized surface code. The assumption on the check weights means that the degree of each vertex and face is at most four. Therefore, the result follows from Lemma~\ref{lem:homologicalcodebound}.

    If $\mc C$ does contains disentangled sets of two qubits, we may remove them to get an \code{n'<n,k,d} code. Applying the previous argument shows that $kd^2\le An'\le An$.
\end{proof}

As an application of Theorem~\ref{thm:wt4CSSdeg2}, we may consider a generalization of the bivariate bicycle (BB) codes in Ref.~\cite{Bravyi_2024}. Concretely, let $\mc C$ be a CSS code defined by parity-check matrices $H_X = [\begin{array}{c|c}A & B\end{array}]$ and $H_Z = [\begin{array}{c|c}B^\top & A^\top\end{array}]$, where $A$ and $B$ are polynomials in $x = S_\ell\otimes I_m$ and $y = I_\ell \otimes S_m$ and $S_r\in \mathbb F_2^{r\times r}$ denotes a cyclic shift matrix. Ref.~\cite{Bravyi_2024} showed that certain codes where $A$ and $B$ are the sum of three terms that are powers of $x$ or $y$ have good finite-length performance.
Certain codes where $A$ and $B$ consists of two terms that are powers of $x$, $y$, and $z=xy$ (called trivariate bicycle codes in Ref.~\cite{Voss25multivariatebicycle}) were shown to have bulks which are identical to the surface code.
This result is generalized by applying Lemma~\ref{lem:wt4homologicalcode} and Theorem~\ref{thm:wt4CSSdeg2}: if $A$ and $B$ are each the sum of any two monomials in $\mathbb F_2[x,y]$, then $\mc C\in \CSS^*(4)$, so it is a generalized surface code and must have parameters bounded by $kd^2\le An$.

We expect Theorem~\ref{thm:wt4CSSdeg2} to generalize to CSS codes $\mc C$ with check weight 4 such that each qubit is in \emph{at most} two $X$ checks and two $Z$ checks. In this case, the proof technique of Lemma~\ref{lem:wt4homologicalcode} would show that $\mc C$ is a generalized surface code on a manifold with boundary. Any qubit that is in only one $X$ check would be associated with an edge on a rough boundary, and any qubit in only one $Z$ check would be associated with an edge on a smooth boundary. The issue with completing such a proof is that we are not aware of an analogous result to Theorem~\ref{thm:fetayabound} showing that $kd^2=O(n)$, or even $d^2=O(n)$, for manifolds with boundary. Theorem~\ref{thm:fetayabound} is proven in Ref.~\cite{Fetaya12} by applying Gromov's systolic inequality for essential manifolds to an appropriate Riemannian metric on the cellulated surface. For the case of surface codes on a manifold with boundary, we would need a relative homology version of this result which constrains the length of a curve connecting two rough boundaries or two smooth boundaries of a surface. Note that Ref.~\cite{Fetaya12} does claim to work for surfaces with boundaries. However, it is assumed there that all boundaries are of the same type and that the code supports logical qubits if all of the ``holes'' defined by the boundaries are filled, which is not sufficient for our purposes. An alternative approach is to use techniques bounding the parameters of codes on lattice quotients~\cite{arnault2025variantBPT}.

\subsection{Subsystem codes of check weight 2}
\label{subsec:subsystemwt2}

By breaking down stabilizer checks as the product of gauge checks, subsystem codes can generally achieve code parameters with lower check weight than stabilizer codes. In particular, there exist subsystem codes with check weight 2 achieving $k, d=\Theta(\sqrt n)$~\cite{bravyi2011subsystem} and good subsystem codes ($k,d=\Theta(n)$) with check weight 3~\cite{baspin2024wire}.
It is natural to wonder whether we could obtain better subsystem codes, perhaps even asymptotically good, with check weight 2. We rule out this possibility for CSS subsystem codes and show that the codes constructed in \cite{bravyi2011subsystem} are optimal among CSS subsystem codes. Fig.~\ref{fig:subsystem-bpt} illustrates the achievable parameters for subsystem codes of check weight 2.

\begin{figure}
    \centering

\usetikzlibrary{arrows.meta}
\tikzset{
    myarrow/.style={-{Triangle[length=3mm,width=2mm]}}
}
\begin{tikzpicture}[scale=0.48]
\hypersetup{
  colorlinks=true,
  linkcolor=black,
  urlcolor=black,
  citecolor=black,
  filecolor=black
}
    \begin{axis}[
        axis lines = left,
        axis line style = myarrow,
        xlabel = {\huge $k$},
        ylabel = {\huge $d$},
        xmin = 0, xmax = 1.1,
        ymin = 0, ymax = 1.1,
        xtick = {0, 0.5, 1},
        xticklabels = {\LARGE 1\vphantom{$\sqrt{n}$},\LARGE $\sqrt{n}$, \LARGE $n$\vphantom{$\sqrt{n}$}},
    ytick = {0, 0.5, 1},
    yticklabels = {\LARGE $1$, \LARGE $\sqrt{n}$, \LARGE $n$},
        width=10cm, height=10cm,
        domain=-0.1:1,
        samples=100,
        axis on top
    ]

    \addplot[thick, black] coordinates {(0, 0.5) (0.5, 0.5) (1, 0)};

    \draw (1,0) -- (0,0.5);

    \

    \draw[pattern=crosshatch, pattern color=red!30] (0,0) rectangle (1,1);

    \draw[fill=MellowGreenRGB] (0,0.5) -- (0.5,0.5) -- (1,0) -- (0,0);

    \addplot[only marks, mark=*, mark size=3pt] coordinates {(0.5, 0.5)};

    \node at (0.65, 0.55) {\LARGE \cite{bravyi2011subsystem}};

    \node at (0.5,0.75) {\LARGE ruled out by Thm.~\ref{thm:subsystem}};
    
    \node at (0.25,0.43) {\LARGE $d\leq \sqrt{n}$};

    \node at (0.6,0.23) {\LARGE $kd\leq n$};
    
    \end{axis}

\end{tikzpicture}
    \caption{
    A log-log plot visualization of bounds on subsystem codes with $n$ physical qubits and check weight 2.
    The green region corresponds to the achievable parameters $(k,d)$, where $k$ is the number of logical qubits and $d$ is the distance.
    Codes in the green region can be realized by the construction in Ref.~\cite{bravyi2011subsystem}.
    The parameters in the red region are ruled out by Theorem~\ref{thm:subsystem}.%
    }
    \label{fig:subsystem-bpt}
\end{figure}
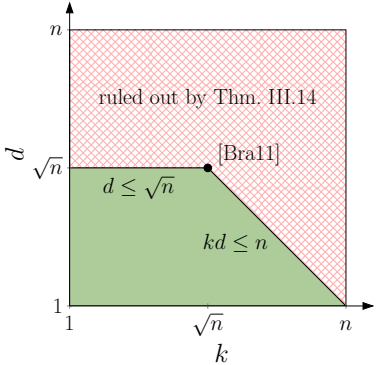

\begin{theorem}
    The parameters of a CSS subsystem code with check weight 2 satisfy $d\le \sqrt{n}$ and $kd\le n$.
    \label{thm:subsystem}
\end{theorem}

To prove this theorem, we relate the parameters of a CSS subsystem code with check weight 2 to the properties of an associated binary matrix $A$.
The optimal parameters of such a subsystem code can then be characterized in terms of the achievable properties of the matrix.
In the following discussion, we let $d_{col}$ and $d_{row}$ be the distance of the classical codes $\col(A)$ and $\row(A)$, respectively.

\begin{theorem}
There exists an \code{n,k,d} subsystem code with check weight 2 if and only if there exists a binary matrix $A$ with at most $n$ entries equal to 1, $k=\rank(A)$, and $d=\min(d_{col}(A),d_{row}(A))$.
    \label{thm:subsystem-char}
\end{theorem}

\begin{proof}
    The ``if'' direction is proven in Ref.~\cite{bravyi2011subsystem}, where it was shown that given a matrix $A$, it is possible to construct a subsystem code with $k= \rank(A)$, $d_X= d_{col}(A)$, and $d_Z= d_{row}(A)$.
    The code is defined by identifying the qubits with the nonzero entries of $A$ and placing $X$ or $Z$ gauge checks on pairs of qubits in the same row or column, respectively.
    The construction was presented for square matrices $A$, but it works just as well if $A$ is rectangular.

    We prove the converse direction by showing that any CSS subsystem code of check weight 2 has essentially the same structure if we group together qubits in the appropriate way.
    Let $\mc C$ be such a code with gauge group $\calG = \ab{G_1,\dots,G_m}$ and parameters \code{n,k,d}. We show that a desired matrix $A$ exists with at most $n$ ones, $\rank(A)= k$, $d_{col}(A)=d_X$, and $d_{row}(A)= d_Z$. 
    We assume without loss of generality that $\mc C$ has the minimal number of qubits for the given $k$ and $d$.
    We may also assume that $\mc G$ has minimal total weight among all generating sets of checks with weight at most 2.
    This means that $G_1,\dots,G_m$ are linearly independent.
    Additionally, each qubit is contained in at least one $X$ check and one $Z$ check by the same argument as in the proof of Lemma~\ref{lem:C3}.\ref{prop:1X1Z}, using minimality of $n$ and the total weight.
    
    Consider a graph $H$ on $2n$ vertices labeled with $X_1,Z_1,\dots,X_n,Z_n$. 
    Connect two Paulis $P_iP_j$ with an edge if $P_iP_j$ is a check of the subsystem code.
    Notice that if $P_i$ and $P_j$ are connected in this graph, then $P_iP_j\in \calG$ since it is the product of the operators along the path.
    Note that any $X_i$ is disconnected from any $Z_j$ since $\mc C$ is a CSS code.
    The graph $H$ has no cycles as $\{G_1,\dots,G_m\}$ is linearly independent.
    
    Partition the vertices of $H$ into connected components.
    A connected component cannot support a mix of weight-2 checks and weight-1 checks. Otherwise, there would exist checks $Z_i$ and $Z_iZ_j$ with $i$, $j$ in the connected component. Since we could replace $Z_iZ_j$ with $Z_j$, this contradicts minimality of total weight.
    Furthermore, each size-1 connected component must be a weight-1 check because each qubit is contained in at least one check of each type.
    We will accordingly call the connected components \emph{weight-1} or \emph{weight-2}.
    
    Let $C_1^X,\cdots,C_{r_X}^X, C_1^Z,\cdots,C_{r_Z}^Z$
    be the weight-2 connected components, where $r_X$ is the number of weight-2 $X$ components, and similarly for $r_Z$.
    Let $r_1$ denote the number of weight-1 connected components of either $X$ or $Z$ type.
    In a cycle-free graph, the number of edges (here, exactly $|\calG|-r_1$) plus the number of connected components ($r_X+r_Z+r_1$) equals the number of vertices, so $|\calG|+r_X+r_Z= 2n$.
    Hence, $2k=2n-|\calG|-|\calS| = r_X+r_Z-|\calS|$.

    Construct a binary matrix $A\in\mathbb{F}_2^{r_X\times r_Z}$ where $A_{i,j}$ is the parity of the number of qubits in both $C_i^X$ and $C_j^Z$; see Fig.~\ref{fig:subsystemexample} for an example.
    
    Each qubit participates in at most one intersection $C_i^X\cap C_j^Z$, so the number of ones in the matrix $A$ is at most $n$.
    
    \begin{figure}[htpb]
        \centering
        \hspace{-1em}\raisebox{17mm}{(a)}\includegraphics[width=0.25\columnwidth]{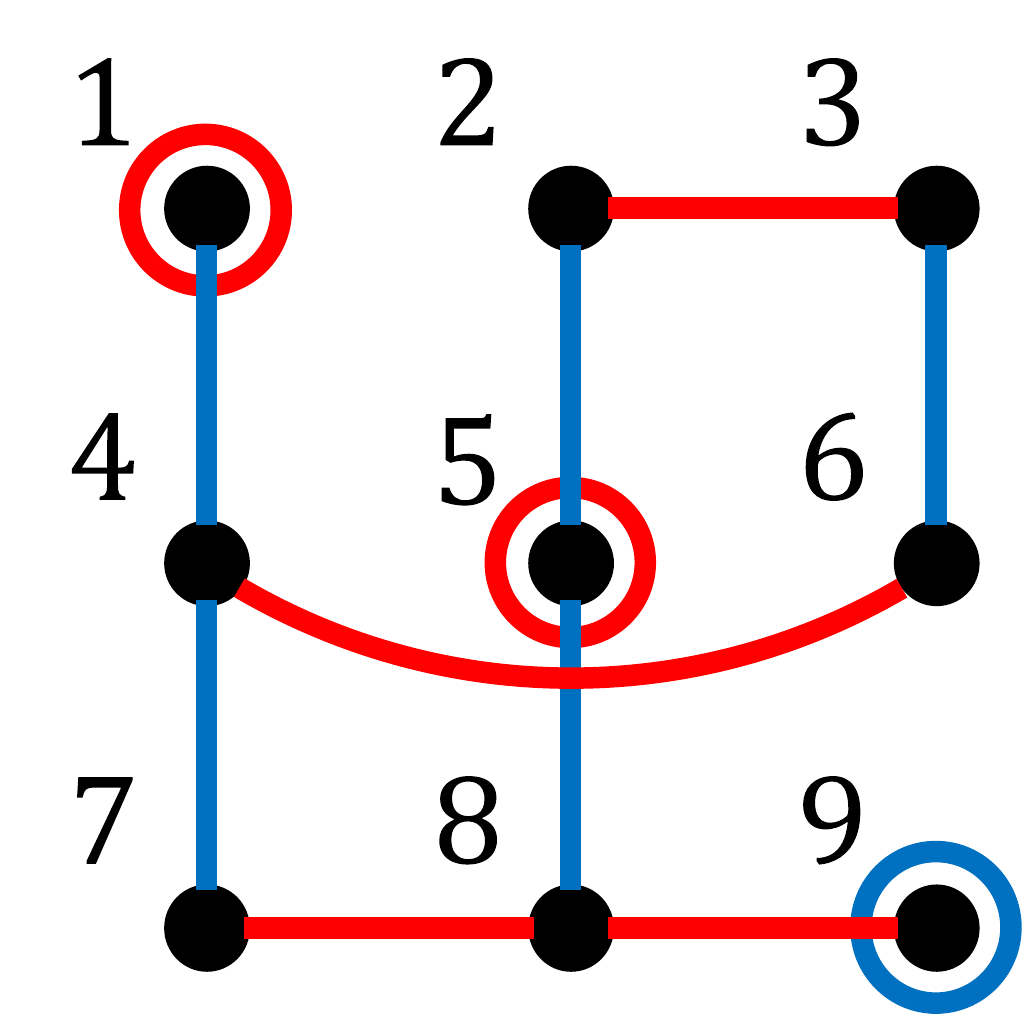}\qquad
        \raisebox{17mm}{(b) }\raisebox{8mm}{$\renewcommand{\arraystretch}{1.4}\begin{bmatrix}0&1&1\\1&0&1\\1&1&0\end{bmatrix}$}
        \caption{(a) An example of a subsystem code with check weight 2. Red (blue) lines and circles indicate $X$ ($Z$) checks of weight two and one, respectively. (b) The associated matrix $A$, where the rows correspond to the weight-2 $X$ components $\{X_2, X_3\}$, $\{X_4, X_6\}$, $\{X_7, X_8, X_9\}$ and the columns correspond to the weight-2 $Z$ components $\{Z_1, Z_4, Z_7\}$, $\{Z_2, Z_5, Z_8\}$, $\{Z_3, Z_6\}$. Entries correspond to the parities of the number of common qubits between the components.}
        \label{fig:subsystemexample}
    \end{figure}

    We first characterize the bare logical operators.
    A $Z$ operator $P$ is a bare logical operator if and only if, (i) it is not supported on any weight-1 $X$ components, and (ii) for all weight-2 $X$ components, $P$ is supported on either the entire component or none of the component.
    The ``if'' direction follows because $Z$ operators supported on a single weight-2 $X$ component are bare logical operators.
    The ``only if'' direction follows because there cannot be any weight-2 $X$ checks intersecting both the support of $P$ and its complement by commutativity and because $P$ does not commute with any weight-1 $X$ checks.
    Similarly, an $X$ operator is bare logical operator if and only if, (i) it is not supported on any weight-1 $Z$ components, and (ii) for all weight-2 $Z$ components, the $X$ operator is supported on either the entire $Z$ component or none of the component.

    We next show the rank of the stabilizer group is $|\calS| =  \dim \ker A + \dim\ker A^\top$. 
    Since our subsystem code is CSS, the stabilizer group can be generated by a set of $X$ stabilizers and $Z$ stabilizers. 
    We show the dimension of the $Z$ stabilizers is $\dim\ker A^\top$ and the dimension of the $X$ stabilizers is $\dim\ker A$.
    By the previous paragraph, there is a one-to-one correspondence between vectors $v\in\mathbb{F}_2^{r_X}$ and bare $Z$-type logical operators $Z_v$ supported on the $X$ components indicated by $v$ ($C_i^X$ where $v_i=1$).
    A bare logical operator $Z_v$ is a stabilizer if and only if it is in the gauge group, which happens if and only if its support has even intersection with each weight-2 $Z$ component. This is because $Z_iZ_j$ is in the gauge group if and only if $Z_i$ and $Z_j$ are in the same weight-2 $Z$ component, and any support on weight-1 $Z$ components can be cleaned out.
    The final condition is equivalent to $v A = 0$, as the $j$-th entry tracks the parity of the support of $Z_v$ on the $j$-th $Z$ component, i.e., $v\in \ker A^\top$.
    Hence, the rank of the $Z$ stabilizers is $\dim\ker A^\top$. Similarly, the rank of the $X$ stabilizers is $\dim\ker A$, so $|\calS|= \dim\ker A^\top + \dim\ker A$. 

    Combining with the previous discussion, we have
    \begin{align}
        2k &= r_X+r_Z-|\calS|\\
        &= (r_X-\dim \ker A^\top)+(r_Z - \dim \ker A)\\
        &= 2\rank(A).
    \end{align}
    Hence, $k= \rank (A)$.

    We now consider the code distance.
    A dressed logical operators is the product of a bare logical operator and a gauge group element.
    Let $Z_v$ be the bare logical operator corresponding to some vector $v\in\mathbb{F}_2^{r_X}$.
    Multiplying by $Z$ checks, $Z_v$ is gauge-equivalent to an operator with (i) weight 1 on any $C_i^Z$ where $(vA)_i=1$, (ii) weight 0 on any $C_i^Z$ where $(vA)_i=0$, and (iii) weight 0 on all weight-1 $Z$ components.
    Hence, the operator $Z_v$ is gauge-equivalent to an operator of weight $|vA|$.
    Furthermore, the parity of the support of $Z_v$ on each $C_i^Z$ cannot change when multiplying by gauge operators, so the minimum weight of a gauge-equivalent operator to $Z_v$ is exactly $|vA|$.
    The minimum nonzero value of $|vA|$ is exactly $d_{row}(A)$, so $d_Z = d_{row}(A)$.
    Similarly, $d_X = d_{col}(A)$.
\end{proof}
    
We now show how Theorem~\ref{thm:subsystem} readily follows.
\begin{proof}[Proof of Theorem~\ref{thm:subsystem}]
    For an \code{n,k,d} CSS subsystem code with check weight 2, let $A$ be the matrix given by Theorem~\ref{thm:subsystem-char}.
    The number of nonzero entries in $A$ is at least $d_{col}(A)\cdot d_{row}(A)$ because at least $d_{col}(A)$ rows have nonzero entries, each of which has at least $d_{row}(A)$ nonzero entries. However, the number of nonzero entries is at most $n$. Hence, $d^2\le d_X\cdot d_Z\le d_{col}(A)\cdot d_{row}(A) \le n$, so $d\le \sqrt{n}$.

    Similarly, $d_{col}(A)\cdot \rank(A)\le n$, because $k$ linearly independent columns of $A$ must have at least $d_{col}$ nonzero entries each. Hence, $d_X\le n/k$, and similarly $d_Z\le n/k$, so $d\le n/k$, giving the other bound.
\end{proof}

We remark that the optimal tradeoffs between $\rank(A)$, $d_{col}(A)$, and $d_{row}(A)$ for a matrix $A$ are only known asymptotically, but our characterization implies that finding better CSS subsystem codes with check weight 2 is equivalent to finding better matrices.
Bravyi showed the existence of matrices $A$ with $\rank(A),d_{col}(A), d_{row}(A) = \Theta(\sqrt{n})$ using an argument analogous to the Gilbert-Varshamov bound~\cite{bravyi2011subsystem}, giving associated subsystem codes with $k,d=\Theta(\sqrt n)$. These codes can also be made geometrically local in two spatial dimensions at a slight cost to the parameters.
Theorem~\ref{thm:subsystem} shows that, even without geometric locality, beating these parameters is not possible with weight-2 checks.

\section{Finite-size upper bounds}
\label{sec:finitesizebounds}
Having obtained asymptotic bounds on the parameters of quantum codes based on their check weight, we now consider codes of finite length.
We derive upper bounds on possible parameters of CSS and stabilizer codes using LP techniques.
Bounds for general quantum codes based on LP have been studied in Refs.~\cite{ShorLaflamme97,Rains99shadow,AshikhminLitsyu99}.
Most of our constraints are based on the ones from Ref.~\cite{calderbank1998quantum} for stabilizer codes and extended to CSS codes in Ref.~\cite{KrishnaTillich25LP}. We modify these systems of constraints by adding additional constraints based on requiring the code to have a certain check weight. 

\subsection{CSS codes}
The variables of the LP are the number of stabilizers and logical operators of the code of different weights, called the weight distribution. If a code has certain parameters and check weight, its weight distribution must satisfy certain linear inequalities, so the infeasibility of the LP guarantees that no code with the desired parameters and check weight exists.

\begin{definition}
    The weight distribution of a length-$n$ classical code $C$ is $A_0$, \dots, $A_n$, where $A_i$ is the number of codewords in $C$ of weight $i$. For an $n$-qubit quantum CSS code derived from classical codes $C_X$ and $C_Z$, its weight distribution is defined as the weight distributions of $C_X^\perp$ and $C_Z^\perp$, i.e., the quantities $A_0^X, \ldots, A_n^X$ and $A_0^Z, \ldots, A_n^Z$, where $A_i^X$ and $A_i^Z$ denote the number of $X$ and $Z$ stabilizers of weight $i$, respectively. Additionally, we define the dual weight distributions $B_0^X, B_1^X, \ldots, B_n^X$ and $B_0^Z, B_1^Z, \ldots, B_n^Z$ as the weight distributions of $C_X$ and $C_Z$, i.e., the respective number of (possibly trivial) $Z$ and $X$ logical operators.
\end{definition}

The MacWilliams identities~\cite{Macwilliams1963}, defined using Krawtchouk polynomial coefficients, relate weight distributions of codes to those of their duals.

\begin{definition}
    The \emph{Krawtchouk polynomial coefficients} over a field of size $q$ are 
    \begin{equation}
        K_{\ell}(j;n, q) = \sum_{s=0}^\ell (q-1)^{j-s}(-1)^s \binom{j}{s} \binom{n-j}{\ell - s}.
    \end{equation}
\end{definition}

These expressions are the coefficients of $x^{n-j}y^j$ in the expansion of $(x+(q-1)y)^{n-j}(x-y)^j$.

\begin{theorem}[\cite{Macwilliams1963}]\label{thm:macwilliams}
    The weight distributions $A_0, \ldots, A_n$ of a classical binary linear code $C$ and $B_0, \ldots, B_n$ of its dual code $C^\perp$ are related by
    \begin{equation}
        \sum_{j=0}^n A_j K_\ell(j; n, 2) = |C| B_\ell
    \end{equation}
    for all $\ell=0, \dots, n$.
\end{theorem}

\begin{theorem}[\cite{calderbank1998quantum,KrishnaTillich25LP}]\label{thm:css-lp}
Any distance-$d$ CSS code derived from classical codes $C_X$, $C_Z$ has weight distributions $\{A_i^X\}$, $\{A_i^Z\}$ and dual weight distributions $\{B_i^X\}$, $\{B_i^Z\}$ satisfying the following constraints.
\begin{enumerate}
    \item $A_0^X = 1$, $A_0^Z = 1$, and $A_i^X \geq 0$, $A_i^Z \geq 0$ for $i=1, \dots, n$.
    \item $\sum_{i=0}^n A_i^X = |C_X^\perp|$,  $\sum_{i=0}^n A_i^Z = |C_Z^\perp|$.
    \item $A_\ell^X \leq B_\ell^Z$ for $\ell = 0, \ldots, n$.
    \item $A_\ell^Z \leq B_\ell^X$ for $\ell = 0, \ldots, n$.
    \item $A_\ell^X = B_\ell^Z$ for $\ell = 1, \ldots, d - 1$.
    \item $A_\ell^Z = B_\ell^X$ for $\ell = 1, \ldots, d - 1$.
\end{enumerate}
\end{theorem}
\begin{proof}
    The first two constraints follow from properties of $C_X^\perp$ and $C_Z^\perp$---they each have the zero codeword and the sum of the weight distributions is the total number of codewords. Constraints 3 and 4 follow from the fact that $C_Z^\perp \subseteq C_X$ and $C_X^\perp \subseteq C_Z$. Constraints 5 and 6 use the property that there are no nontrivial $X$ or $Z$ logical operators of weight below the distance.
\end{proof}

We introduce new constraints that apply to codes with a given check weight $w$.

\begin{theorem}[Check weight constraints]\label{thm:css-checkweight}
Any CSS code derived from classical codes $C_X$ and $C_Z$ with respective dimensions $k_X$ and $k_Z$, weight distribution $\{A_i^X\}$, $\{A_i^Z\}$, and check weight $w$ additionally satisfies the following constraints.
\begin{enumerate}
    \item $\sum_{i=0}^{mw} A_i^X \geq \sum_{j=0}^m \binom{n-k_X}{j}$ for $m = 0, \ldots, \left\lfloor \frac{n}{w} \right\rfloor$.
    \item $\sum_{i=0}^{mw} A_i^Z \geq \sum_{j=0}^m \binom{n-k_Z}{j}$ for $m = 0, \ldots, \left\lfloor \frac{n}{w} \right\rfloor$.
\end{enumerate}
\end{theorem}
\begin{proof}
    Consider a linearly independent subset of the $X$ checks. Since each check has weight at most $w$, any product of at most $m$ checks gives a unique codeword of $C_X^\perp$ of weight at most $mw$, giving a lower bound for the sum of the weight distribution coefficients of index at most $mw$. A similar argument holds for $Z$ checks.
\end{proof}

We remark that Ref.~\cite{Ben-HaimLitsyn} also presents a constraint for classical codes based on their check weight that was adapted to quantum CSS codes in Ref.~\cite{KrishnaTillich25LP}. That formulation requires all checks of the code to have weight exactly $w$, whereas our bound holds for codes with checks of weight at most $w$.

\subsection{Stabilizer codes}

We can use a similar approach to bound the parameters of general stabilizer codes since any stabilizer code can be associated with a classical additive quaternary code~\cite{calderbank1998quantum}. The primal and dual weight distributions of the code, also related by the MacWilliams identities~\cite{Macwilliams1963}, are defined as follows.

\begin{definition}
The weight distribution of an $n$-qubit stabilizer code $\mc C$ is $A_0, \dots, A_n$, where $A_i$ is the number of stabilizers of weight $i$. The dual weight distribution is $B_0, \dots, B_n$, where $B_i$ is the number of (possibly trivial) logical Pauli operators of weight $i$.
\end{definition}

\begin{theorem}[\cite{calderbank1998quantum}]\label{thm:stab-lp}
    Any stabilizer code with parameters \code{n,k,d} has primal and dual weight distributions $\{A_i\}$, $\{B_i\}$ satisfying the following constraints.
    \begin{enumerate}
        \item $A_i\geq0, B_i\geq0$ for $i = 0,\dots,n$.
        \item $A_0=1, B_0 = 1$.
        \item $\sum_{i=0}^n A_i=2^{n-k}$, $\sum_{i=0}^n B_i = 2^{n
        +k}$.
        \item $B_\ell = \frac{1}{2^{n-k}} \sum_{j=1}^n K_\ell(j;n,4) A_j$.
        \item $A_\ell \leq B_\ell$ for $\ell = 0,\dots,n$.
        \item $A_\ell = B_\ell$ for $\ell = 1,\dots, d-1$.
    \end{enumerate}
\end{theorem}

\begin{proof}
    The first two constraints are due to the nonnegativity of the primal and dual weight distributions and the zero codeword. The third is a result of the weight distribution coefficients summing to the total number of stabilizers or logical operators. The fourth constraint is the MacWilliams identity for quaternary codes. The fact that all stabilizers are logical operators and that there are no nontrivial logical operators of weight less than the distance implies the fifth and sixth constraints.
\end{proof}

We add a similar constraint for a code of check weight $w$.

\begin{theorem}\label{thm:stab-checkweight}
    Any stabilizer code with parameters \code{n,k,d}, primal and dual weight distributions  $\{A_i\}$,$\{B_i\}$, and check weight $w$ must satisfy the following constraint
    \begin{equation}
        \sum_{i=1}^{mw} A_i\geq \sum_{j=0}^m \binom{n-k}{j} \text{ for } m = 0, \ldots, \left\lfloor \frac{n}{w} \right\rfloor.
    \end{equation}
\end{theorem}

\begin{proof}
    Just as in the proof of Theorem~\ref{thm:css-checkweight}, any product of at most $m$ linearly independent checks gives a unique stabilizer of weight at most $mw$, giving the required lower bound on $\sum_{i=1}^{mw} A_i$.
\end{proof}

\subsection{Implementation}\label{subsec:implementation}
We implemented the LP constraints for both stabilizer codes and CSS codes in Python using Gurobi~\cite{gurobi} to search for feasible solutions with the weight distribution as variables, obtaining the upper bounds shown in Figs.~\ref{fig:CSSall3D}(a),~\ref{fig:CSSfixednkd}, and~\ref{fig:StabCSS} as well as Figs.~\ref{fig:Stab3Dappendix},~\ref{fig:stabfixednkd}, and~\ref{fig:StabCSSappendix} in Appendix~\ref{app:numerics}.

\paragraph{Stabilizer codes.}
Our LP uses constraints from Theorems~\ref{thm:stab-lp} and~\ref{thm:stab-checkweight} to check whether there exist valid weight distributions for codes with parameters $n$, $k$, $d$ and check weight $w$.
We take the largest possible feasible $k$ while fixing the other quantities.

\paragraph{CSS codes.}
For fixed $n$, $d$ and check weight $w$, we find the largest possible $k$ such that there exist $k_X$ and $k_Z$ with $k = k_X + k_Z - n$ giving a feasible solution for the constraints in Theorems~\ref{thm:macwilliams},~\ref{thm:css-lp}, and~\ref{thm:css-checkweight}. To do this, we iterate over all valid combinations of $k_X$ and $k_Z$. We simplify the search by imposing $k_Z\le k_X$ since the LP is symmetric in $X$ and $Z$. Furthermore, we can also skip over very large values of $k_X$ when $w(n-k_X) < n$. This is because for such a code, $n-k_X$ independent $X$ checks of weight at most $w$ cannot cover all $n$ qubits, so the remaining qubits must support single-qubit $Z$ stabilizers to prevent weight-1 logical operators. But changing these $Z$ stabilizers to single-qubit $X$ stabilizers would give a code with the same parameters, which would be feasible for the LP with a smaller value of $k_X$ with $w(n-k_X) \ge n$.

\begin{figure*}[htpb]
    \centering
    \raisebox{39mm}{(a)}\includegraphics[width=0.27\linewidth,trim={0 0 40mm 0},clip]{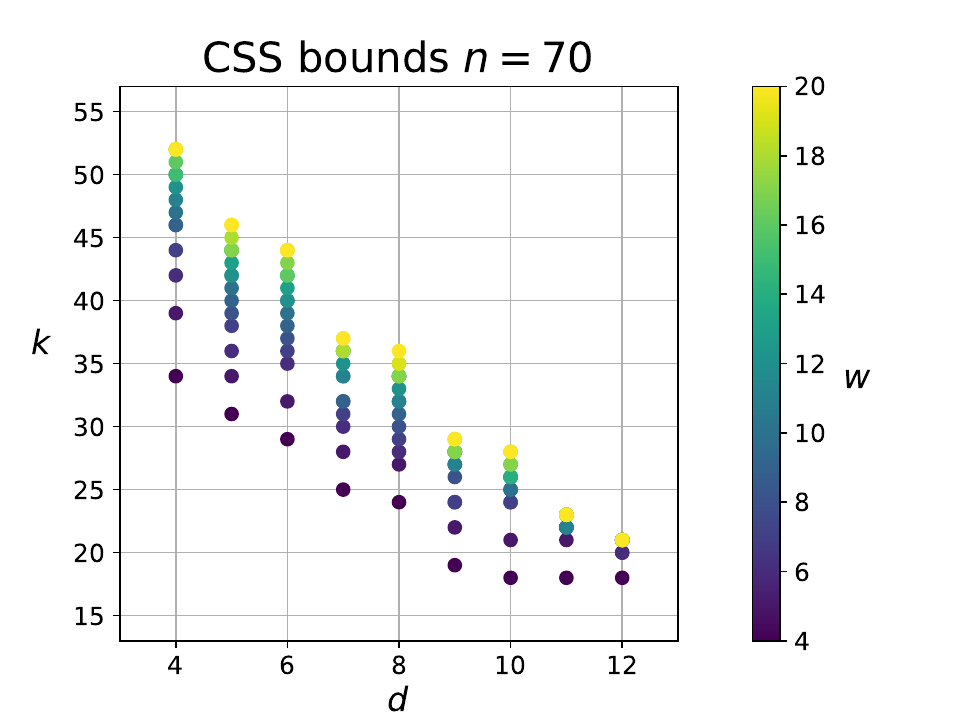}
    \raisebox{39mm}{(b)}\includegraphics[width=0.27\linewidth,trim={0 0 40mm 0},clip]{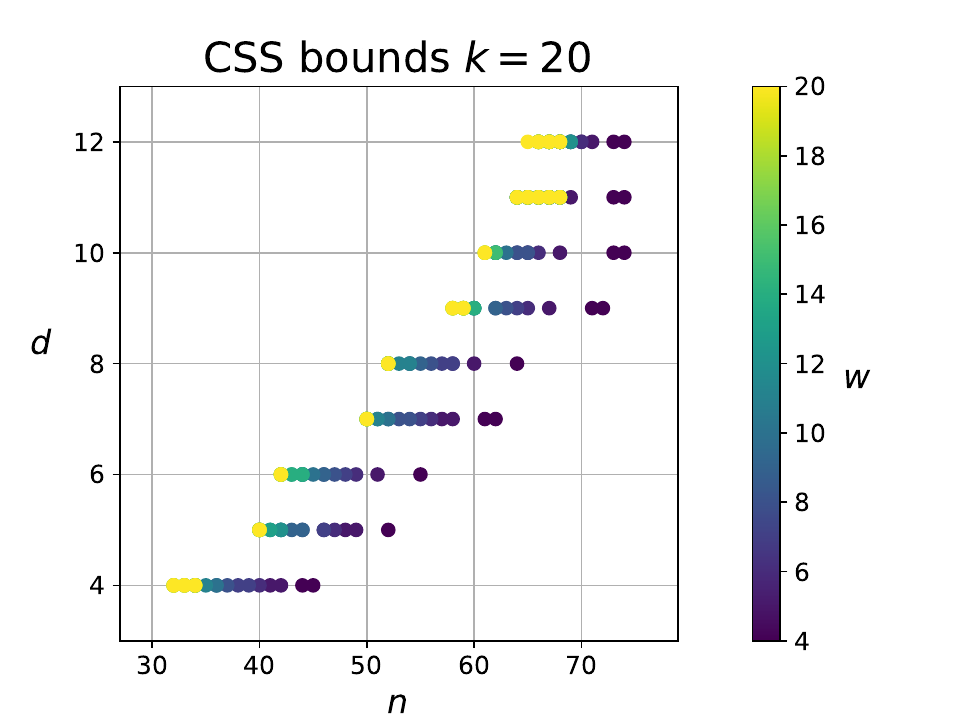}
    \raisebox{39mm}{(c)}\includegraphics[width=0.36\linewidth,trim={0 0 0 0},clip]{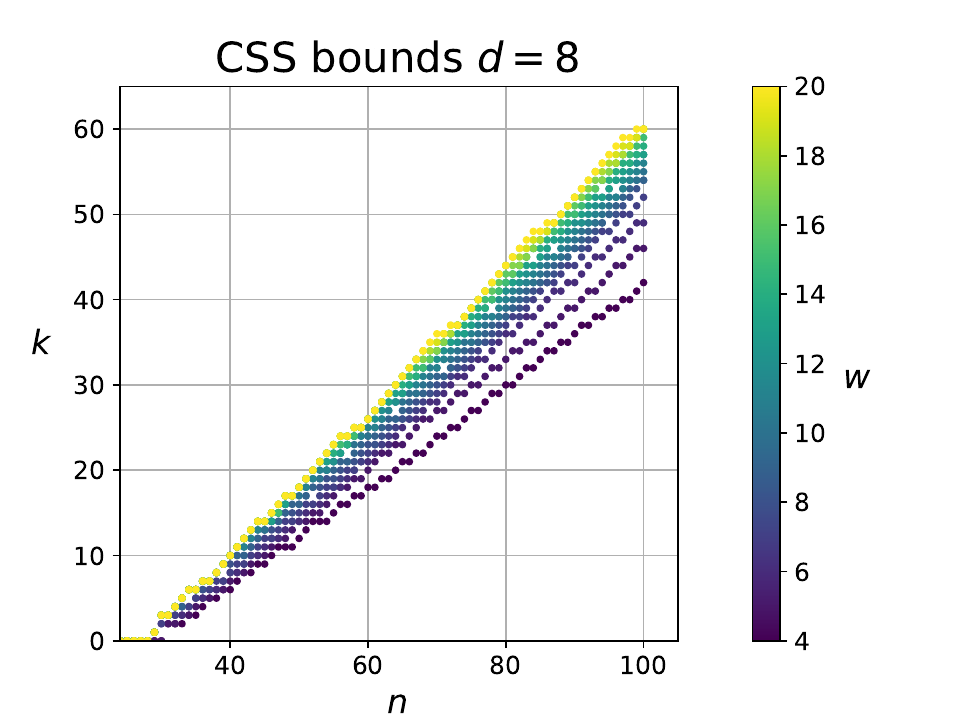}\qquad
    \caption{
    Cross-sections of upper bounds on code parameters of CSS codes from Fig.~\ref{fig:CSSall3D}(a), with the different colors representing bounds for different check weights. We pick a representative value of $n$, $k$, or $d$ to fix in (a), (b), and (c), respectively. Parameters above the points in (a) and (c), and those in (b) above or to the left of the furthest point of a given check weight, are ruled out by the LP. 
    } 
    \label{fig:CSSfixednkd}
\end{figure*}

\paragraph{Post-processing.}
For CSS codes or stabilizer codes, we denote the maximum feasible choice of $k$ obtained for any $n$, $d$, $w$ by $\bar k_1(n,d,w)$, and use this data for post-processing to improve our bounds.
Since LP solutions do not guarantee a corresponding quantum code while infeasible solutions do certify that no code can exist, the existence of $n'>n$ such that $\bar k_1(n',d,w) < \bar k_1(n,d,w)$ would also imply a maximum of $\bar k_1(n',d,w)$ logical qubits for an $n$-qubit code, given $d$ and $w$; we could always add disentangled qubits to reach an $n'$-qubit code. We therefore define $k_1(n,d,w) = \min_{n' \geq n} \bar k_1(n',d,w)$.

We run a second version of the LP with the additional constraint of $A_1 = 0$ in the stabilizer case, or $A_1^X = 0$ and $A_1^Z = 0$ in the CSS case, and we denote the maximum feasible $k$ obtained as $\bar k_2(n,d,w)$.
This constraint ensures that there are no weight-1 checks giving disentangled qubits which do not contribute to the dimension or distance. However, $n$-qubit codes should have access to feasible solutions for codes of fewer physical qubits, so we post-process by defining $k_2(n,d,w) = \max_{n'\leq n} \bar k_2(n',d,w)$.

Since $k_1(n,d,w)$ and $k_2(n,d,w)$ are both valid upper bounds, we may take the minimum of the two. Additionally, a given value of $k$ for $d$ and $w$ would not be possible if it is not possible for a lower value of $d$ or a higher value of $w$, as such conditions are less restrictive. Therefore, we report the value
\begin{equation}
    \label{eq:kpostprocessing}
    k(n,d,w) = \min_{\substack{d'\le d\\ w'\ge w}} \min(k_1(n,d',w'), k_2(n,d',w')).
\end{equation}

\paragraph{Numerical instability.} We discuss issues of numerical instability in solving the LP in Appendix~\ref{app:numerics}.

\subsection{Results}

Fig.~\ref{fig:CSSall3D}(a) presents the upper bounds on $k$ obtained for CSS codes across different distances, code lengths, and check weights, with cross sections shown in Fig.~\ref{fig:CSSfixednkd}. 
For given $n$, $d$, and $w$, we find the largest $k$ for which there is a feasible solution, as discussed in Sec.~\ref{subsec:implementation}. Therefore, in Fig.~\ref{fig:CSSfixednkd}(a)(c), we can rule out CSS codes with parameters above plotted points on the $k$ axis. In the cross section of fixed $k$, the plotted points indicate that no higher $k$ is possible; see Fig.~\ref{fig:CSSfixednkd}(b).
Because of the monotonicity of $k(n,d,w)$ with respect to $n$ and $d$ captured by Eq.~\eqref{eq:kpostprocessing}, CSS codes with parameters above or to the left of the furthest point of a given check weight are not possible.

As expected, the upper bounds on $k$ increase with lower distances, higher code lengths, and higher check weight. For a given check weight, there is a roughly linear relationship between the remaining two parameters when one of $n$, $k$, or $d$ is fixed. When $k$ or $d$ is fixed and $w\ge 5$, this is consistent with the existence of asymptotically good CSS codes where $k,d=\Theta(n)$. For fixed $n$, one might expect a Gilbert-Varshamov-type bound~\cite{Gilbert52,Varshamov57} for quantum codes of the form $k/n = 1 - 2H(d/n)$, where $H(\cdot)$ is the binary entropy function~\cite{CS96CSScode}. However, this is not distinguishable from a linear relationship due to the small range of $d$ plotted. As we allow $w$ to increase, the curves converge to the upper bound obtained for codes with no check weight constraint, i.e., when the constraints in Theorem~\ref{thm:css-checkweight} are removed; see also Fig.~\ref{fig:boundsandsmallcodes}. We also note that the check weight constraints give a more pronounced improvement in the upper bound for smaller values of $d$, as seen by the larger vertical spread between the curves of different colors in Fig.~\ref{fig:CSSall3D}(a).

Plots for the bounds on stabilizer codes, giving similar conclusions, are shown in Figs.~\ref{fig:Stab3Dappendix} and~\ref{fig:stabfixednkd} of Appendix~\ref{app:numerics}.

\begin{figure}[htpb]
    \centering
    \includegraphics[width=0.9\columnwidth,trim={0 0 0 0},clip]{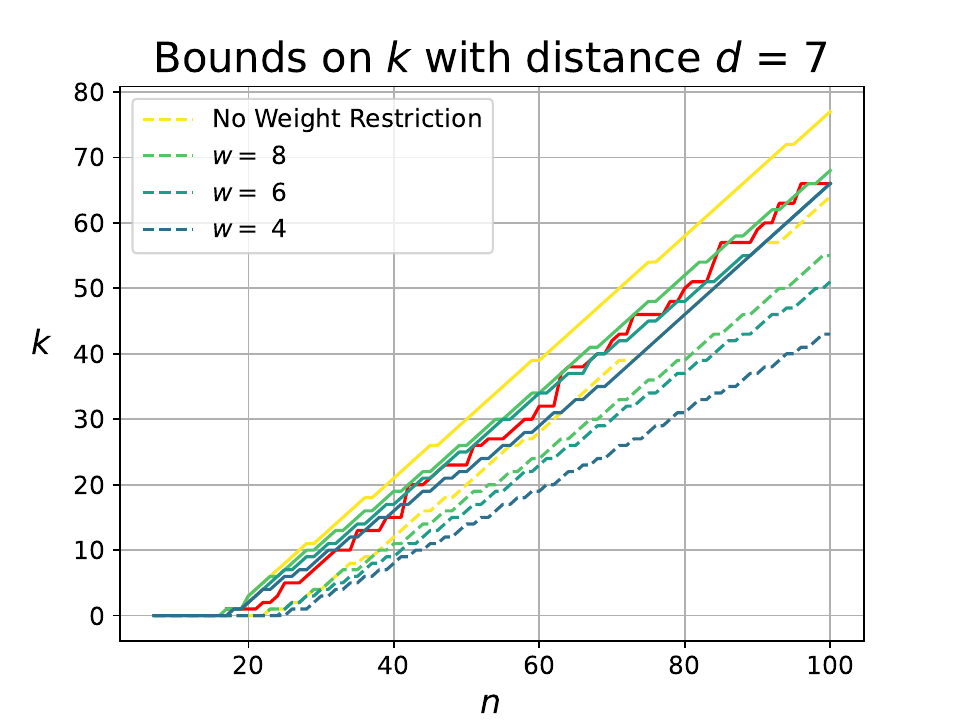}
    \caption{Comparison of upper bounds on code parameters for CSS and stabilizer codes in the $(n,k)$ space for fixed distance $d=7$ and various check weights. Solid (dashed) lines denote bounds for stabilizer (CSS) codes. Explicit constructions from Ref.~\cite{Grassl:codetables} (without constraints on check weight) are shown as the red line.}
    \label{fig:StabCSS}
\end{figure}

Additionally, we plot our bounds for both CSS and stabilizer codes across different check weights for distance $d=7$ in Fig.~\ref{fig:StabCSS} (with plots for different distances in Fig.~\ref{fig:StabCSSappendix} of Appendix~\ref{app:numerics}). We also show in the red curve achievable parameters from Ref.~\cite{Grassl:codetables}, obtained through explicit stabilizer codes without bounds on the check weight. As a subset of stabilizer codes, CSS codes have upper bounds which are lower than those for stabilizer codes, even below the parameters of known stabilizer codes. Furthermore, some of the upper bounds for stabilizer codes of fixed check weight are also below the red curve. Thus, our numerical results prove that despite the existence of asymptotically good codes that are CSS~\cite{PKAsymptoticGoodLDPC,leverrier2022quantumtannercodes,Dinur_goodqLDPCCodes} and have check weight 5~\cite{hastings2023quantumweightreduction}, there is a concrete advantage to allowing higher check weights or non CSS codes for finite sizes.

\section{Small code instances}
\label{sec:smallcodes}
In the previous section, we established finite-size bounds on the achievable parameters of quantum codes, particularly focusing on constraints imposed by low check weights.
To contextualize these theoretical limits, we search for explicit instances of qLDPC codes within the promising family of quantum Tanner codes~\cite{leverrier2022quantumtannercodes}. We present the parameters of the codes found in this section, situating them within the parameter space defined by our bounds.
This comparison visually highlights the existing gaps between known codes and the fundamental limits for a given check weight, offering insights into where new codes might be discovered.
The details of the search are described in Appendix~\ref{app:codesearch}.

\begin{figure}[htpb]
    \centering
    \raisebox{49mm}{(a)}\includegraphics[width=0.9\columnwidth,trim={1.3cm 0 3cm 0},clip]{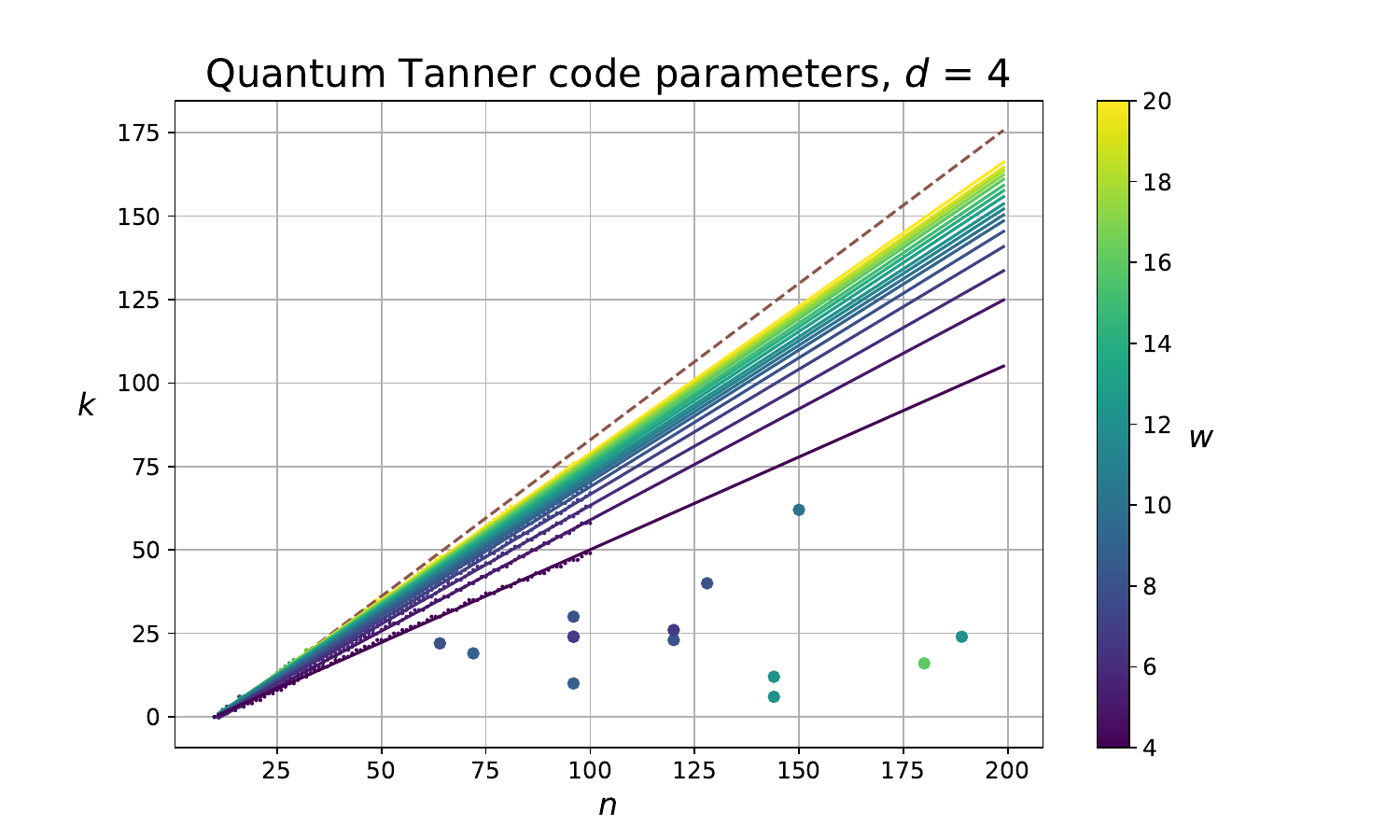}
    \raisebox{49mm}{(b)}\includegraphics[width=0.9\columnwidth,trim={1.3cm 0 3cm 0},clip]{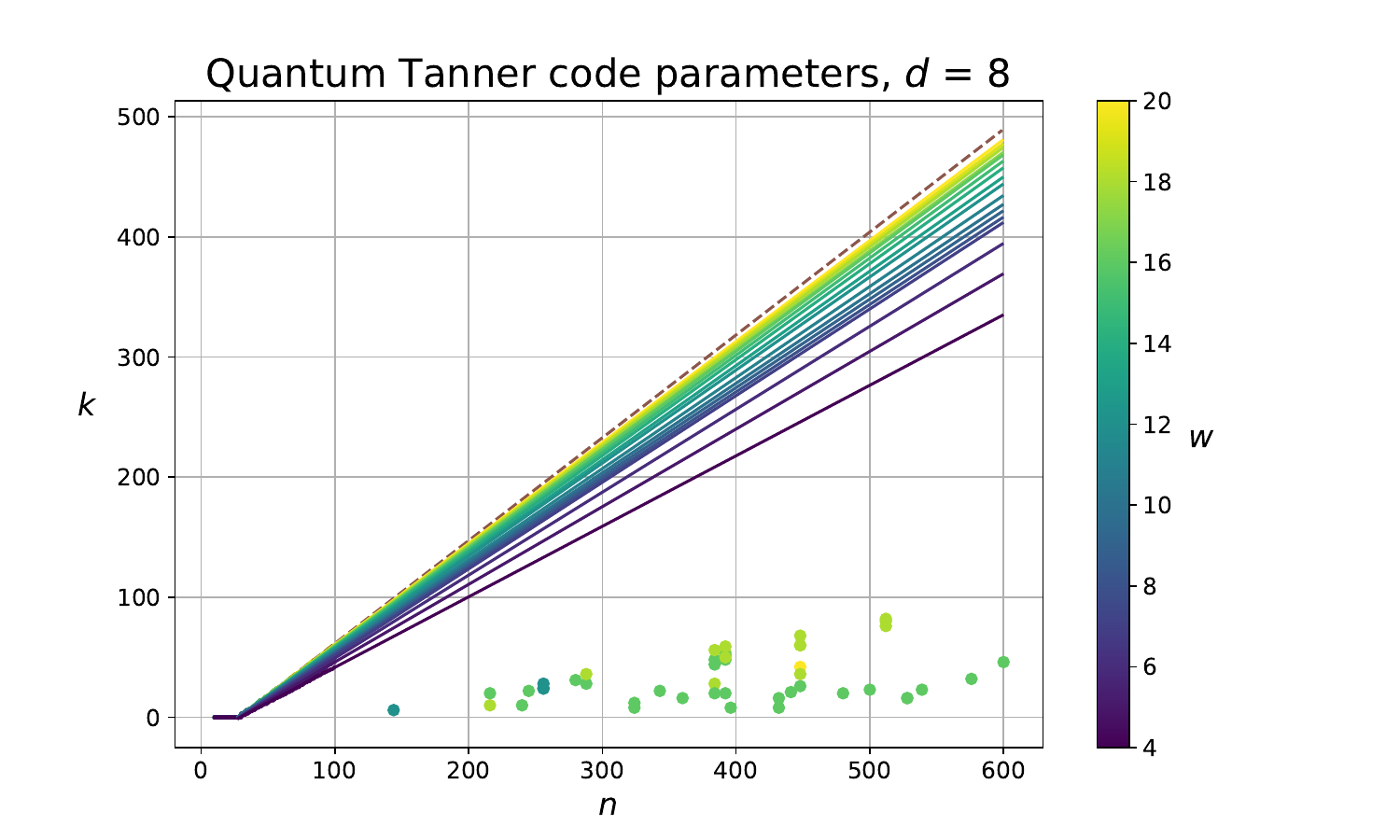}
    \caption{Code parameters of our quantum Tanner codes in Tables~\ref{tab:qtcode_instances} and~\ref{tab:qtcode_instances_8_5} with (a) $d\ge 4$ and (b) $d\ge 8$ in the $(n,k)$ space. We also plot the upper bounds on code parameters of CSS codes, extrapolated to values of $n\ge 100$ using a linear fit. Colors represent the check weights of the codes found as well as the check weights used for the upper bounds. The dashed line is the extrapolated upper bound for CSS codes with no weight restriction.}
    \label{fig:boundsandsmallcodes}
\end{figure}

Fig.~\ref{fig:boundsandsmallcodes} presents the code parameters of our selected quantum Tanner codes in the $(n,k)$ space for codes with $d\ge 4$ or $d\ge 8$, compared to the CSS code upper bounds from Sec.~\ref{sec:finitesizebounds}. Since the bounds are roughly linear for $n\le 100$, we extrapolate them to $n=200$ or $n=600$. For distance 4, we find small codes which are close to the upper bounds, meaning that they are nearly optimal. For distance 8, there is a gap between the best small codes and the LP upper bounds, indicating that either more efficient codes can be constructed or the upper bounds can be improved.

\begin{figure}[htpb]
    \centering
    \includegraphics[width=\columnwidth]{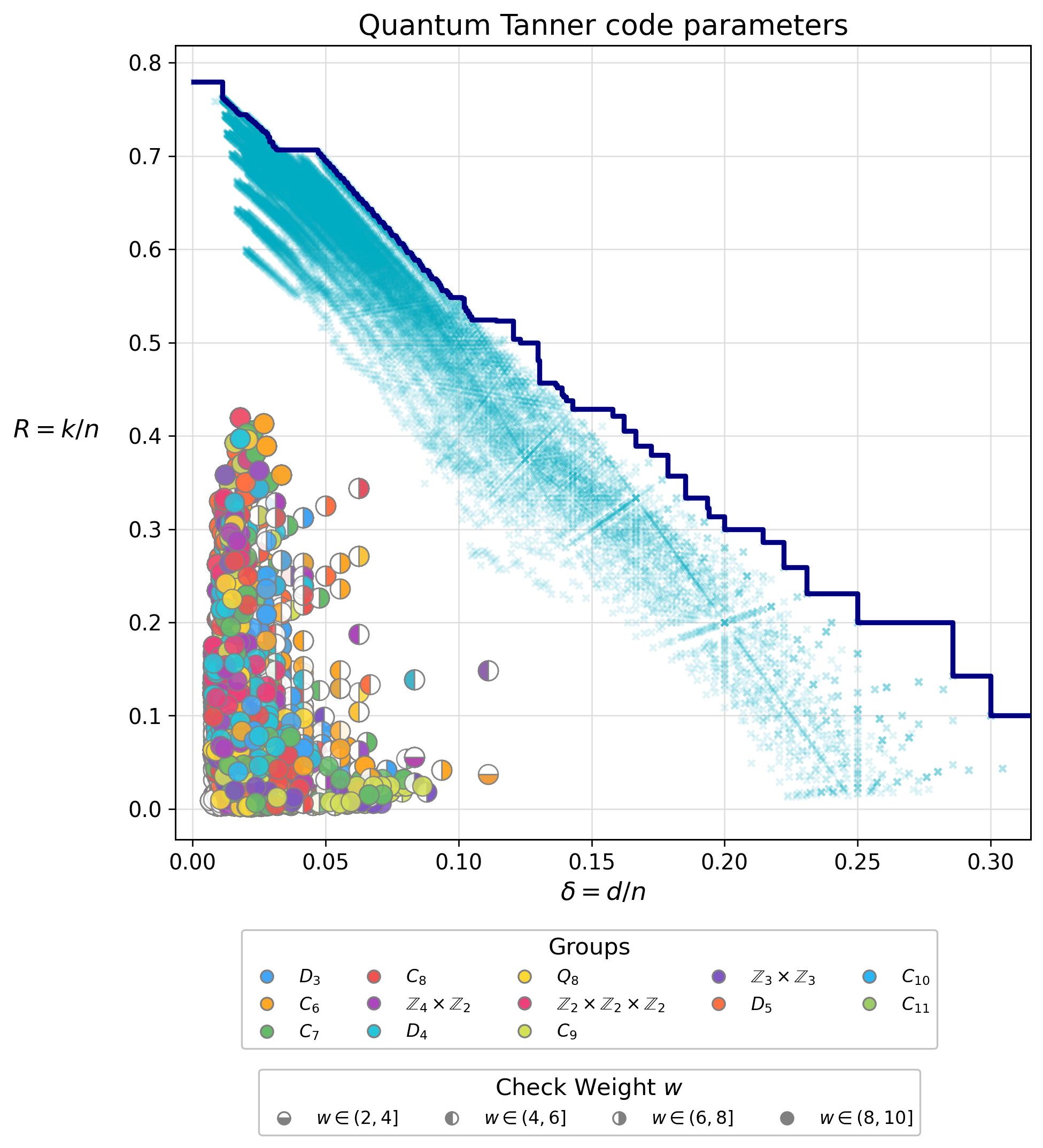}\qquad
    \caption{
    Explicit quantum Tanner code parameters compared with LP bounds. We restrict to codes with $w\le10$ and $d\geq3$. Colors label the underlying base groups, while fill styles represent the check weight $w$. Cyan crosses denote feasible points from the LP with $w\leq10$, and the solid step curve traces the corresponding achievable LP frontier.
}
    \label{fig:rdtradeoffwithbound}
\end{figure}
In Fig.~\ref{fig:rdtradeoffwithbound}, we compare all the generated quantum Tanner codes with $w\le 10$ and $d\geq3$, with distinct colors representing different base groups as input to the construction. The solid curve is the LP achievability frontier $R_{\rm LP}(\delta)$, defined as the envelope of all CSS code feasible parameter triples $(n,k,d)$ under the same weight constraint, i.e.,
\begin{equation*}
    R_{\rm LP}(\delta):=\max\{R_i:\delta_i\geq\delta\},
\end{equation*}
where $(R_i,\delta_i)$ are achievable parameters for the LP for some $n\le 300$. Due to numerical instabilities in the LP solver for larger $n$ as detailed in Appendix~\ref{app:numerics}, exact LP bound parameters are computed for $n\le 150$, while values for $150<n\leq300$ are extrapolated using a linear fit derived from the stable regime $0<n\leq150$. As expected, all explicit quantum Tanner code constructions lie strictly below the LP frontier. The best instances populate the high-rate and low-relative-distance regime, demonstrating that at finite blocklength, quantum Tanner codes can achieve comparatively high rates under stringent weight constraints. More complete results of generated codes up to $w\le 20$ can be found in Fig.~\ref{fig:rdtradeoff_w20} of Appendix~\ref{app:codesearch}.

To place these observations in a broader context, Fig.~\ref{fig:CSSall3D}(b) further incorporates BB codes and their variants, as well as representative quantum Tanner codes reported in the literature. We plot the code parameters against the same LP bounds for $w=6$ and $w=10$ as these are the respective lower and upper bounds of check weights of the finite instances presented. All the codes from the literature used in the simulation are presented in Tables~\ref{tab:qt_lit} and~\ref{tab:bb_codes} in Appendix~\ref{app:existingBBcodes}. In contrast with quantum Tanner codes, BB codes extend to noticeably larger relative distances, albeit typically at lower rates.

A representative example of a BB code with favorable parameters is the \code{30,4,6} code with $w=6$. On the quantum Tanner code side, one of the strongest finite-size examples is the \code{64,22,4} code with $w=8$, which approaches the bound in the moderate-rate regime. We present a few more selected codes from our search as well as the best BB-type codes in the literature in Table~\ref{tab:explicit_codes_comparison}.

\begin{table}[t]
\centering
\caption{Selected instances of qLDPC codes}
\label{tab:explicit_codes_comparison}

\setlength{\tabcolsep}{8pt}
\renewcommand{\arraystretch}{1.15}

\begin{tabular}{l c c c}
\toprule
\textbf{Construction} & \code{\boldsymbol{n,k,d}} & $\boldsymbol{w}$ & $\boldsymbol{\bar w}$ \\
\midrule
\addlinespace[0.25em]
Quantum Tanner code & \code{96,30,4}     & $8$  & $8$ \\
 & \code{288,16,16}   & $9$  & $9$ \\
 & \code{392,53,12}   & $16$ & $13.3$ \\
& \code{528,16,32}   & $16$  & $13.3$ \\

 & \code{576,32,24}   & $16$ & $13.3$ \\
\addlinespace[0.6em]
\addlinespace[0.25em]
BB code $\&$ variants & \code{126,12,10} & $6$ & $6$ \\
 & \code{144,14,14} & $8$ & $8$ \\
 & \code{294,10,20} & $6$ & $6$ \\
 & \code{340,16,18} & $6$ & $6$ \\
 & \code{730,18,10} & $6$ & $6$ \\
\bottomrule
\end{tabular}
\end{table}

\section{Discussion}
\label{sec:discussion}

In this work, we bounded achievable code parameters using constraints on the check weight, both in the asymptotic regime and for finite-size codes.
We showed that nontrivial stabilizer codes with check weight 3 cannot exist.
For certain CSS stabilizer and subsystem codes of check weight 4 and 2, respectively, we proved a tight upper bound on $k$ and $d$ that can be saturated by existing constructions.
In the finite-length regime, we found numerical constraints on possible code parameters and compared them with our explicit constructions of quantum Tanner codes as well as other existing codes.

There are a few natural extensions of our results.
First, generalizing our bounds to all stabilizer and subsystem codes with check weight 4 and 2, respectively, would complete the characterization of the asymptotically achievable code parameters given constrained check weight.
Second, introducing more constraints on the weight distribution could further improve the parameter bounds.
These constraints may apply to general codes or to specific subfamilies, such as self-dual~\cite{Gottesman1997} or triorthogonal CSS codes~\cite{BravyiHaahMSD,nezami22triorthogonal}, to rule out the existence of codes with given properties and code parameters.
Third, it would be valuable for practical applications to go beyond the setting of static codes and study dynamical protocols, such as Floquet codes, the syndrome extraction circuits, and logical operations.
Since increasing the qubit connectivity of a device may come at the cost of increased error rates, it will be important to understand the tradeoffs with the rate and protection of the protocol given constrained check weight.

Our results combine analytical arguments with numerical optimization to establish strong upper bounds on achievable parameters of qLDPC codes.
Combined with our search of explicit code constructions that approach these limits, our work delineates the landscape of practically relevant qLDPC codes with tens or hundreds of physical qubits.
We hope that the explicit code instances found and made available on Github~\cite{github} will prove useful for future experimental demonstrations of quantum error correction.

\section*{Acknowledgements}

We thank Anirudh Krishna for helpful discussions about unpublished work with Jean-Pierre Tillich~\cite{Krishnatillich} and their code~\cite{KrishnaTillich25LP}.
L.W. is supported by NSERC Fellowship PGSD3-587672-2024 and thanks the Yale Quantum Institute for hosting her during the completion of the project.
R.L. is supported  by NSF grant CCF-2347371.
A.K. and S.G. acknowledge support from the NSF (QLCI, Award No.~OMA-2120757), IARPA and the Army Research Office (ELQ Program, Cooperative Agreement No.~W911NF-23-2-0219).
This research was supported in part by grant NSF PHY-2309135 to the Kavli Institute for Theoretical Physics (KITP).

\appendix

\section{Details of LP numerics and additional figures}
\label{app:numerics}

\paragraph{Numerical instability.}
The range of coefficient obtained from our LP constraints exceeds precision tolerances of Gurobi for large values of $n$ and $d$. To handle this, we normalize coefficients for each constraint and remove any constraint where its largest coefficient is a factor of more than $10^{10}$ larger than its smallest. Our upper bounds thus exclude more constraints for greater values of $n$ and $d$ and may be looser. For very large values of $n$, weight distributions (summing to $2^{n-k}$ for stabilizer codes, or $2^{n-k_X}$ and $2^{n-k_Z}$ for CSS codes) have a wide range of values, causing numerical precision issues and sometimes leading the LP solver to report programs as infeasible even for low values of $k$. We therefore chose to only process bounds for $n \leq 100$ and extrapolated values beyond $100$ with a linear fit when needed; see Fig.~\ref{fig:boundsandsmallcodes}.

Numerical issues were more problematic with stabilizer codes than CSS codes: since stabilizer codes correspond to additive quaternary codes~\cite{calderbank1998quantum}, both the range of Krawtchouk polynomial coefficients and weight distribution solutions are significantly larger. In higher distances of stabilizer codes, more check weight constraints were removed and bounds were not significantly improved using the check weight constraint. Moreover, for higher distances, some parameters generated much lower feasible $k$ than their neighbors, leading to less smooth bounds in Fig.~\ref{fig:Stab3Dappendix}. This may be due to numerical instability.

\begin{figure}[htpb]
    \centering
    \includegraphics[width=\columnwidth,trim={5cm 0 4cm 0},clip]{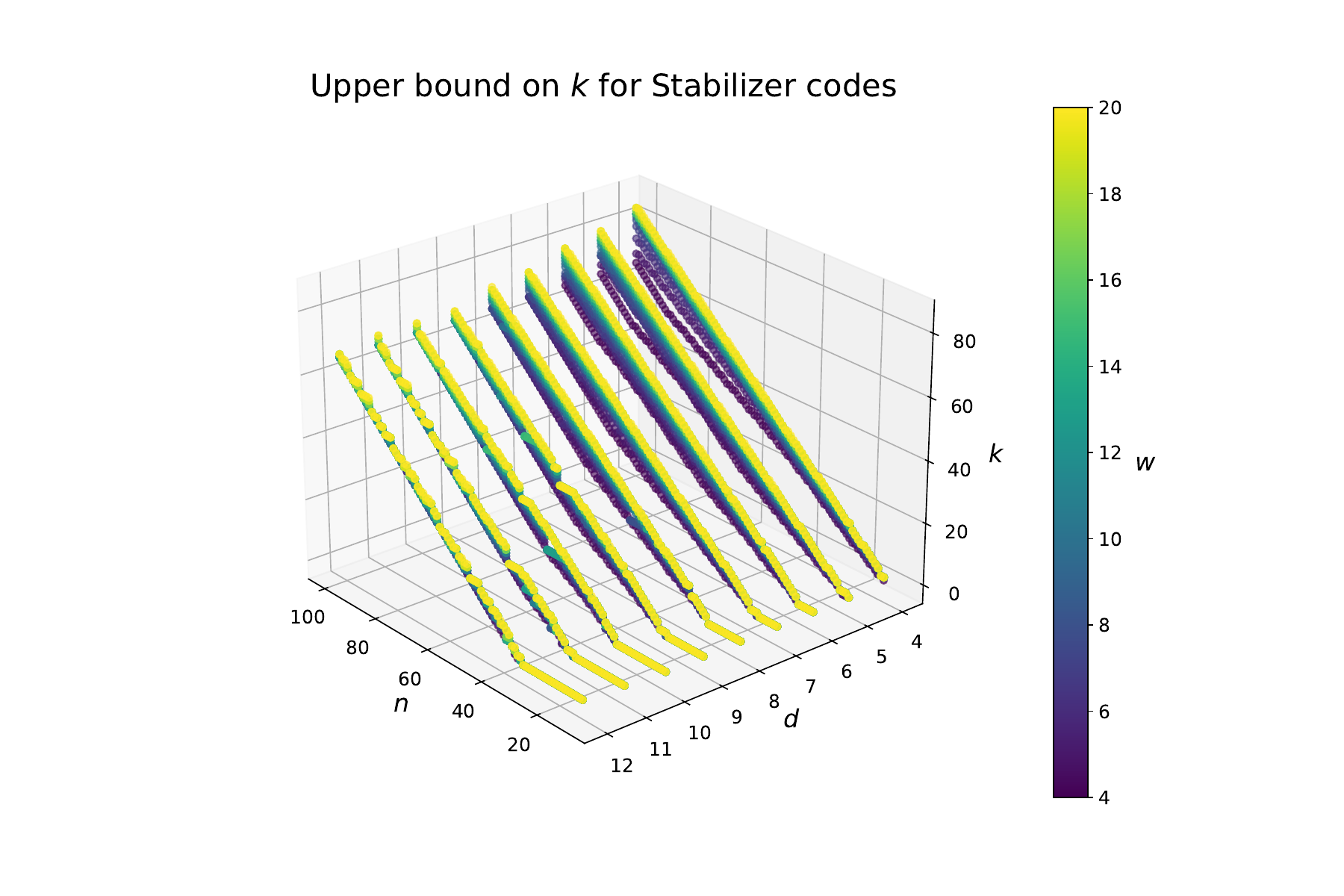}
    \caption{Visualization of the upper bounds on code parameters of stabilizer codes in the $(n,k,d)$ space as a function of check weight $w$ in different colors.}
    \label{fig:Stab3Dappendix}
\end{figure}

\paragraph{Supplementary figures.}
We provide some additional bounds obtained from the LP solver for general stabilizer codes, supplementing the bounds for CSS codes in the main text. In Fig.~\ref{fig:Stab3Dappendix}, we plot our upper bounds on $k$ for stabilizer codes over different values of $n$, $d$, and $w$, and take cross sections for constant $n$, $k$, and $d$ in Fig.~\ref{fig:stabfixednkd}. The linear relationships seen between pairs of parameters for CSS code bounds also occur here, and there is a similar trend in check weight constraints improving the upper bound more for smaller distances. Finally, we present additional plots comparing bounds for CSS and stabilizer codes in Fig.~\ref{fig:StabCSSappendix}.

\begin{figure*}[htpb]
    \centering
    \raisebox{39mm}{(a)}\includegraphics[width=0.27\linewidth,trim={0 0 40mm 0},clip]{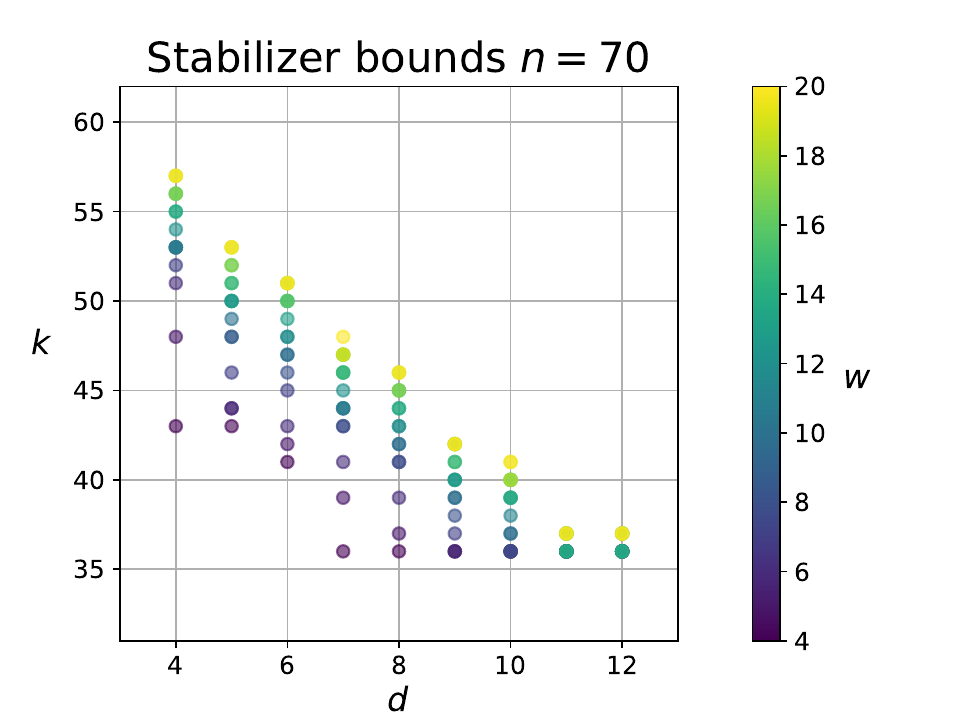}
    \raisebox{39mm}{(b)}\includegraphics[width=0.27\linewidth,trim={0 0 40mm 0},clip]{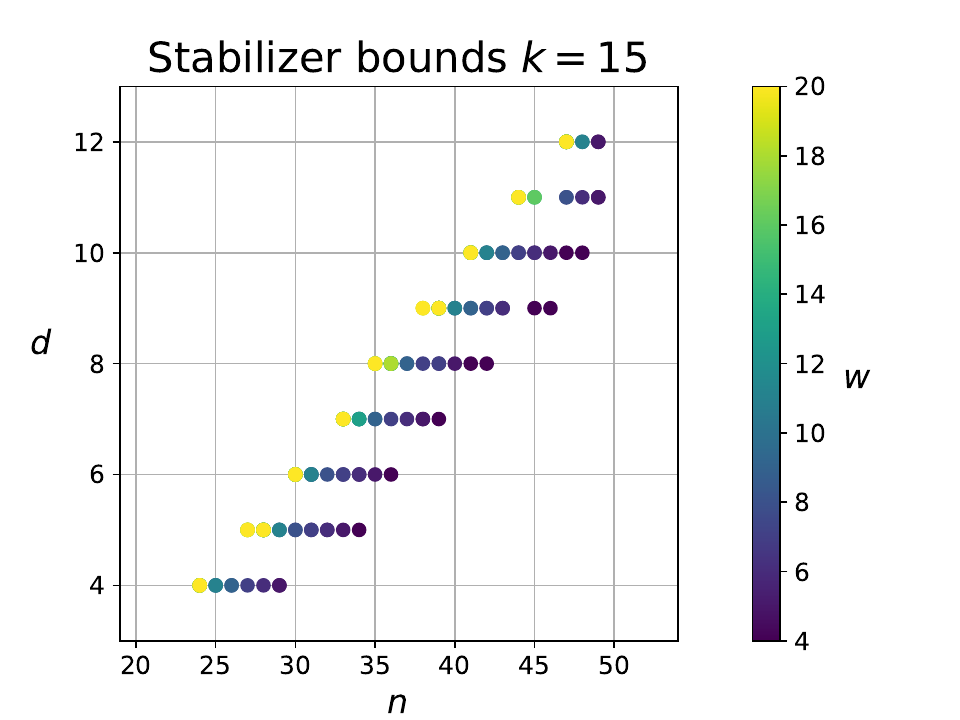}
    \raisebox{39mm}{(c)}\includegraphics[width=0.36\linewidth,trim={0 0 0 0},clip]{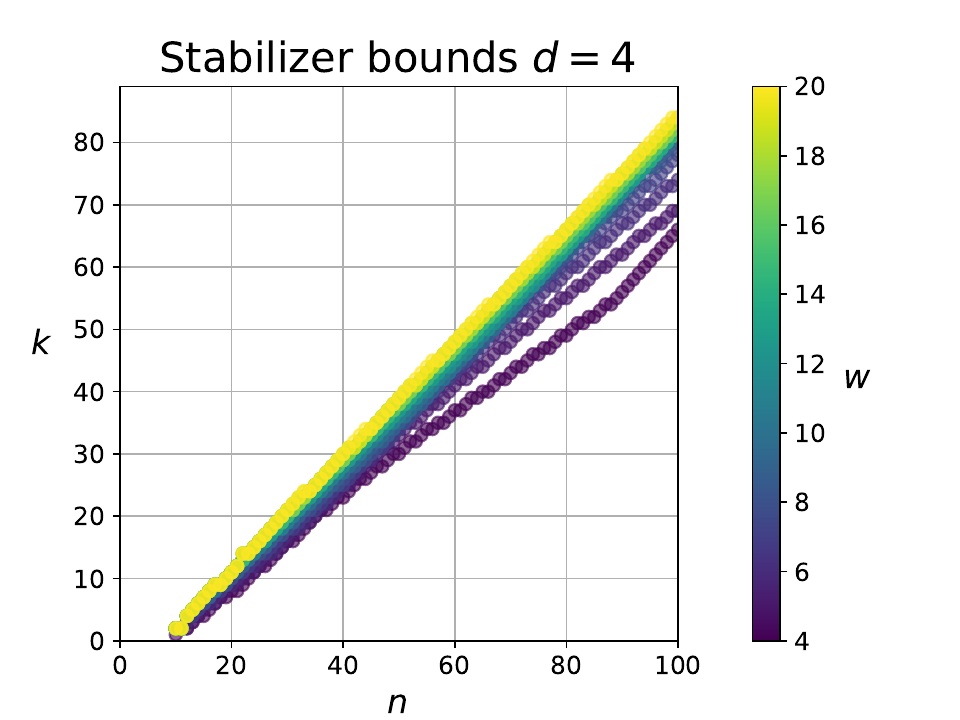}\qquad
    \caption{
    Cross-sections of upper bounds on code parameters of stabilizer codes from Fig.~\ref{fig:Stab3Dappendix}, with the different colored points representing bounds for different check weights. We pick a representative value of $n$, $k$, or $d$ to fix in (a), (b), and (c), respectively. Parameters above the points in (a) and (c), and those in (b) above or to the left of the furthest point of a given check weight, are ruled out by the LP. } 
    \label{fig:stabfixednkd}
\end{figure*}

\begin{figure}[h!]
    \centering
    \raisebox{48mm}{(a)}\includegraphics[width=0.9\columnwidth,trim={0 0 0 0},clip]{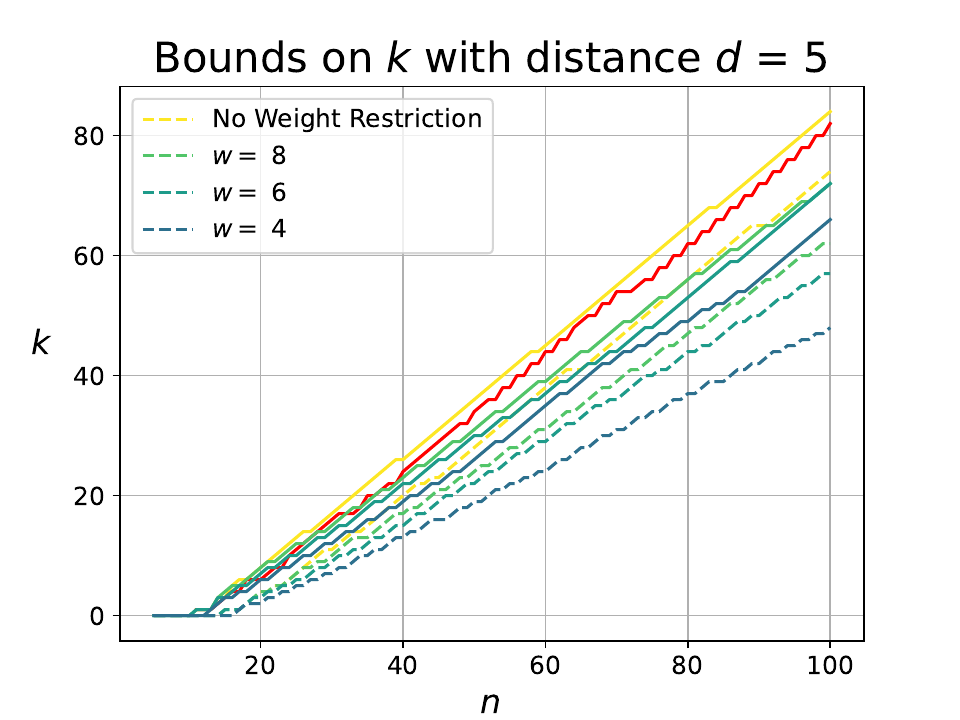}
    \raisebox{48mm}{(b)}\includegraphics[width=0.9\columnwidth,trim={0 0 0 0},clip]{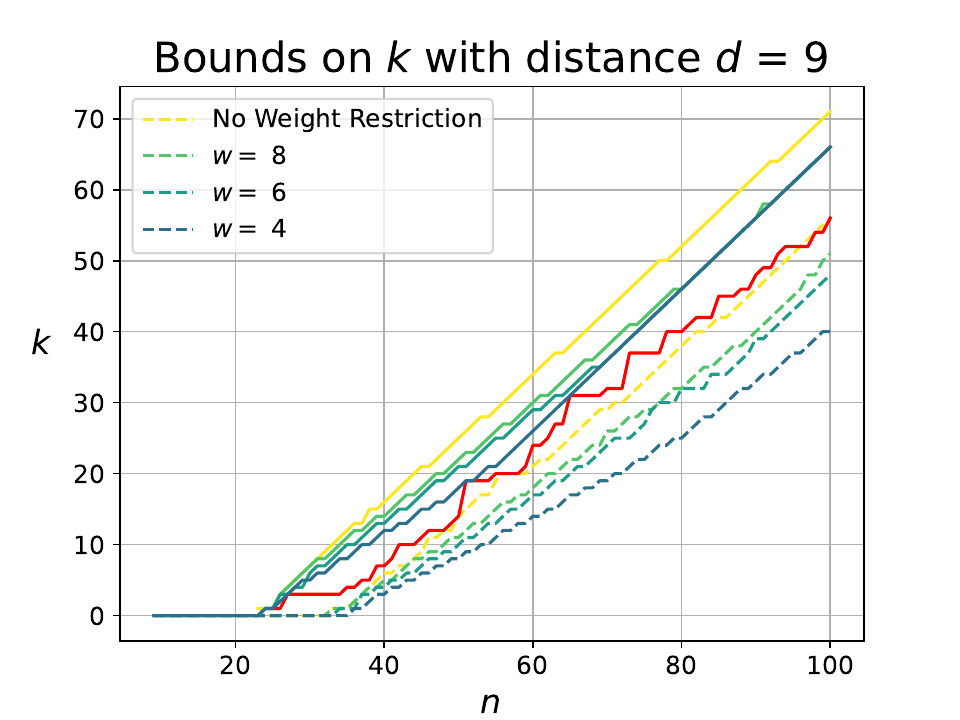}
    \caption{Comparison of upper bounds on code parameters for CSS and stabilizer codes in the $(n,k)$ space for fixed distances (a) $d=5$ and (b) $d=9$ and various check weights. Solid (dashed) lines denote bounds for stabilizer (CSS) codes. Explicit constructions from Ref.~\cite{Grassl:codetables} (without constraints on check weight) are shown as the red line.}
    \label{fig:StabCSSappendix}
\end{figure}

\section{The search for small codes}
\label{app:codesearch}
In this section, we summarize the quantum Tanner code construction and document our randomized search procedure used to produce explicit finite instances.

\subsection{Quantum Tanner codes}
We first provide a concise description of quantum Tanner codes and refer the reader to Refs.~\cite{leverrier2022quantumtannercodes,LZ2023decoding} for more details. Quantum Tanner codes are constructed using two ingredients: a left-right Cayley complex and two local codes. The local codes are classical linear codes $C_A\subseteq \mathbb F_2^{\Delta_A}$, $C_B\subseteq \mathbb F_2^{\Delta_B}$ for constants $\Delta_A$, $\Delta_B$.

To define a left-right Cayley complex, let $G$ be a finite group with symmetric generating sets $A=A^{-1}$, $B=B^{-1}$ of sizes $\Delta_A$, $\Delta_B$, respectively. The left-right Cayley complex is a geometric object with vertices $V=G$, edges $E$, and squares $Q$. The edges are partitioned into two sets $E_A$ and $E_B$ with
\begin{align}
    E_A &= \{\{g,ag\}: g\in G, a\in A\},\\
    E_B &= \{\{g,gb\}: g\in G, a\in B\},
\end{align}
and the squares are defined as
\begin{equation}
    Q = \{\{g, ag, gb, agb\}: g\in G, a\in A, b\in B\}.
\end{equation}

We require $(V,E)$ to be a bipartite graph with vertices $V=V_0\sqcup V_1$ and the squares $Q$ to be nondegenerate. To satisfy these conditions, we may start with an arbitrary left-right Cayley complex and take its double or quadruple cover. For an arbitrary group $\tilde G$ with symmetric generating sets $\tilde A$, $\tilde B$, the double cover is defined by taking $G=\tilde G\times \mathbb Z_2$ with generating sets $A = \tilde A\times \{1\}$ and $B = \tilde B\times \{1\}$. Then, the bipartition of the vertices is given by $V_i = G\times \{i\}$ for $i\in \{0,1\}$. To avoid degeneracy in the definition of the squares, the original generating sets must satisfy the total no-conjugacy condition (TNC): $ag\ne gb$ for any $g\in G$, $a\in A$, and $b\in B$. If this condition is not satisfied, the quadrapartite version of the left-right Cayley complex can instead be used, where 
\begin{align}
    G &= \tilde G\times \mathbb Z_2\times \mathbb Z_2,\\
    A &= \tilde A\times \{0\} \times \{1\},\\
    B &= \tilde B\times \{1\} \times \{0\},\\
    V_0 &= G\times \{(0,0), (1,1)\},\\
    V_1 &= G\times \{(0,1), (1,0)\}.
\end{align}

The quantum Tanner code defined by $G$, $A$, $B$, $C_A$, $C_B$ above is a CSS code constructed as follows. Place qubits on the squares of the left-right Cayley complex. For every vertex $v\in V$, its incident faces $Q(v)$ are identified with the set $\Delta_A \times \Delta_B$. If $v\in V_0$, we place $X$ checks with support in $Q(v)$ corresponding to basis elements of $C_A\otimes C_B$, viewing codewords of $C_A\otimes C_B$ as $\Delta_A\times \Delta_B$ matrices whose nonzero entries are the support of the $X$ check. If $v\in V_1$, several $Z$ checks are placed with support in $Q(v)$ corresponding to basis elements of $C_A^\perp\otimes C_B^\perp$.

Leverrier and Z\'emor~\cite{leverrier2022quantumtannercodes} showed that by taking a sufficiently expanding family of left-right Cayley complexes and local codes that satisfy a robustness condition, one obtains a family of quantum codes with parameters $k,d=\Theta(n)$.

\subsection{Search procedure}
Having specified the construction, we summarize the pipeline used to generate the explicit finite instances.
\begin{enumerate}
    \item \textbf{Group:} We enumerate candidate base groups using the \textsf{GAP} SmallGroup database~\cite{SmallGrp}. In the bipartite quantum Tanner code construction, the number of qubits scale as $n=|G||A||B|/2$, and in the quadripartite variant, it scales as $n=|G||A||B|$. As we are broadly interested in medium-size codes (codes with $n\lessapprox1000$), this constrains the search space. For nontrivial group structure and nontrivial classical codes $C_A$ and $C_B$, we therefore focus on groups with order $6\leq|G|\le 12$. Small groups in this range include cyclic groups ($C_n$), dihedral groups $(D_n)$, and other non-abelian groups characterized in \textsf{GAP} notation (e.g. SmallGroup$(8,3)$ for the quaternion group).

    \item \textbf{Local classical codes:} For the local codes $C_A$ and $C_B$, we use a curated library of short binary linear codes of respective lengths $\Delta_A, \Delta_B\leq 12$ taken from \textsf{codetables.de}~\cite{Grassl:codetables}. For each selected code we fix a parity-check matrix representation (optionally, after a lightweight basis optimization to reduce row weights). Following the strategy advocated by Leverrier et al.~\cite{leverrier2025smallquantumtannercodes}, we randomize coordinate identifications for $C_B$.  Concretely, we sample random column permutations of its parity-check matrix, which correspond to different bijections between the $\Delta_B$ code coordinates and the elements of the generating set $B$. Since a global permutation leaves code parameters and check weights unchanged, it is sufficient to fix $C_A$.
    Empirically, these permutations often reduce the resulting maximum/average stabilizer weights without compromising code parameters.

    \item \textbf{Sampling symmetric generating sets:} For each choice of $(G,C_A,C_B)$, we sample ten random pairs of symmetric subsets $(A,B)$, where $A,B\subseteq G$ with $|A|=\Delta_A$ and $|B|=\Delta_B$. If $(A,B)$ satisfies TNC, we use the bipartite construction; otherwise, we default to the quadripartite construction. Empirically, the bipartite construction produces few valid codes for most small groups and moderate $\Delta_A$, $\Delta_B$ due to the strong TNC restriction, including the constraint $|A|+|B|<|G|$. Holding $(A,B)$ fixed, we then generate 10 instances by applying random column permutations to $C_B$ while fixing $C_A$, and we retain the best instances after post-processing.

    \item \textbf{Code construction and parameter extraction:} Given $(G, A,B,C_A,C_B)$, we construct the corresponding CSS code, record $n$, $k$, and the maximum and average $X$ and $Z$ check weights and qubit degrees $w_X$, $w_Z$, $q_X$, $q_Z$.
    \item \textbf{Distance estimation and certification:}
    To estimate the $X$ and $Z$ distances $d_X$, $d_Z$, we first use the \textsf{GAP} package \textsf{QDistRnd}~\cite{Pryadko_2022}, which returns an upper bound for the distance. We ran 50000 trials for the distance upper bound estimation. As described in the documentation, the failure probability of the algorithm is upper bounded by $p<e^{-\langle n\rangle}$ where $\langle n\rangle$ is the average number of times a codeword of weight $d$ was found. In addition, we implement a deterministic verification step:
    for a proposed distance $t$, we check that for all pairs of Pauli operators with weight at most $t/2$ whose product has weight less than $t$, the two Paulis have the same syndrome only if they multiply to a stabilizer. For all codes that \textsf{QDistRnd} reports as having distances $d_X,d_Z\le 9$, this certification confirms that the distances quoted are exact.
    \item \textbf{Score-based post-processing:}
    We rank finite instances using the score
    \begin{equation}
    \mathrm{score}_\beta := \frac{k\,d^2}{n\,\bar w^\beta},
    \end{equation}
    where $\bar w$ is the average check weight. The $1/\bar w^\beta$ term penalizes high check weight. For each fixed $(G,C_A,C_B, A, B)$, we compute $\mathrm{score}_\beta$ for $\beta\in\{0.5,1.0,1.5,2.0,2.5\}$, and retain the top $K=200$ codes based on $\mathrm{score}_1$, which we find empirically yields a balanced trade-off between distance and check weight. We report selected instances in Tables~\ref{tab:qtcode_instances} and~\ref{tab:qtcode_instances_8_5}. 
\end{enumerate}

\subsection{Observations from generated quantum Tanner codes}
In this section we summarize several empirical trends observed from the full collection of generated finite-size quantum Tanner codes 
\paragraph{Rate-distance tradeoff and comparison to the LP frontier.}
In Figs.~\ref{fig:rdtradeoffwithbound} and~\ref{fig:rdtradeoff_w20}, we observe that the generated quantum Tanner codes concentrate at relatively small $\delta$ and moderate $R$. Increasing $\delta$ rapidly suppresses the achievable rate, revealing a steep finite-size tradeoff. When superposed with LP-feasible points and the corresponding LP achievability frontier under the same check-weight cutoff, the explicit instances lie well below the LP envelope. The gap is distance-dependent: quantum Tanner codes come closer to the bound for small distances (e.g., $d=4$) but fall significantly below it for larger distances, suggesting either room for improved search (e.g., altering the local code parity-check matrices) or the possibility of tighter LP constraints.
\paragraph{Across-group comparison: limited separation between abelian and non-abelian groups.} In Fig.~\ref{fig:3d_frontier}, we showcase the 3D frontier of all the generated codes for different base groups. Combining with Fig.~\ref{fig:rdtradeoff_w20}, we see that codes derived from abelian groups (e.g., cyclic groups and abelian products) substantially overlap with those derived from non-abelian groups (e.g., dihedral, quaternion and dicyclic groups). In particular, the order-8 groups exhibit very similar finite-length frontiers in the regime of a few hundred qubits, and high ${\rm score}_1$ arise across multiple groups. Thus, although intuition about asymptotic codes suggests that non-abelian choices lead to good codes via very good expansion, non-abelianness alone is not the best predictor of code performance in the finite-size regime. Furthermore, across all groups, the 3D frontiers have a consistent wedge-like shape: increasing $d$ at fixed or moderately growing $n$ forces a rapid decrease in $k$, providing a 3D restatement of the finite-size rate-distance tradeoff. 

\paragraph{Role of local codes and generating sets.} In the quantum Tanner code construction, the classical codes $C_A,C_B$ directly determine the supports of the $X$ and $Z$ checks via tensor products $C_A\otimes C_B$ and $C_A^{\perp}\otimes C_B^{\perp}$. Thus, the distance of the quantum code is strongly influenced by the distances of both the local codes and their duals. Empirically, we observe exactly this behavior: combinations of codes such as the $[6,3,3]$, $[7,3,4]$, and $[8,4,4]$ codes often yield comparatively decent distances. By contrast, local codes with very weak dual (e.g., repetition code) tend to produce quantum codes with poor distance. We observe that the choice of the sampled generating sets can also affect $(n,k,d)$. In this sense, the base group provides the geometric scaffold, while local code selection and generator sets are the primary levers for finite-blocklength parameter optimization.  

\begin{figure}[htpb]
    \centering
    \includegraphics[width=\columnwidth]{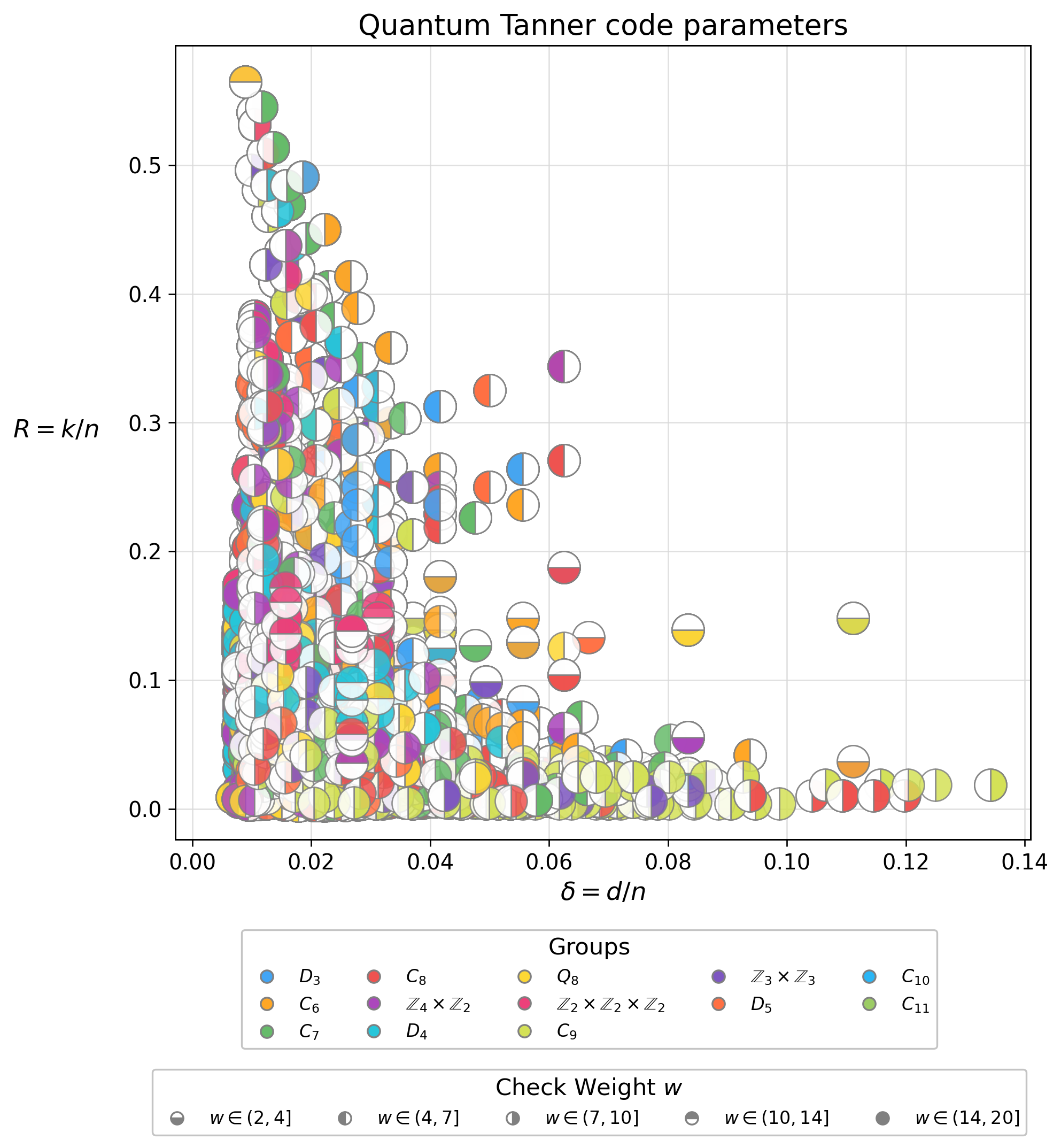}\qquad
    \caption{Rate-distance tradeoff for all generated instances of quantum Tanner codes with $w\le 20$ and $d\ge 3$. Colors distinguish the underlying base groups, while fill styles represent the check weight $w$.
}
    \label{fig:rdtradeoff_w20}
\end{figure}

\section{Existing quantum Tanner and bivariate bicycle codes}
\label{app:existingBBcodes}
Quantum Tanner codes from the literature are presented in Table~\ref{tab:qt_lit}.

BB codes and closely related ring-/lattice-based CSS constructions provide a complementary reference class to our generated quantum Tanner codes: they are explicitly specified in the literature, typically have very small stabilizer weights (often $w=6$), and have been optimized for practical decoding and threshold behavior. Table~\ref{tab:bb_codes} compiles representative explicit instances from prior works, including standard BB families and several generalizations (e.g., coprime factorizations, planar/open-boundary variants, twisted-torus constructions, and tile-code relatives). We report parameters in the form \code{n,k,d} together with stabilizer weight $w$ (when uniformly fixed in the construction), and include brief notes indicating the main structural feature or application emphasized in each reference.

\begin{figure*}[h!]
  \centering
  \includegraphics[width=\textwidth]{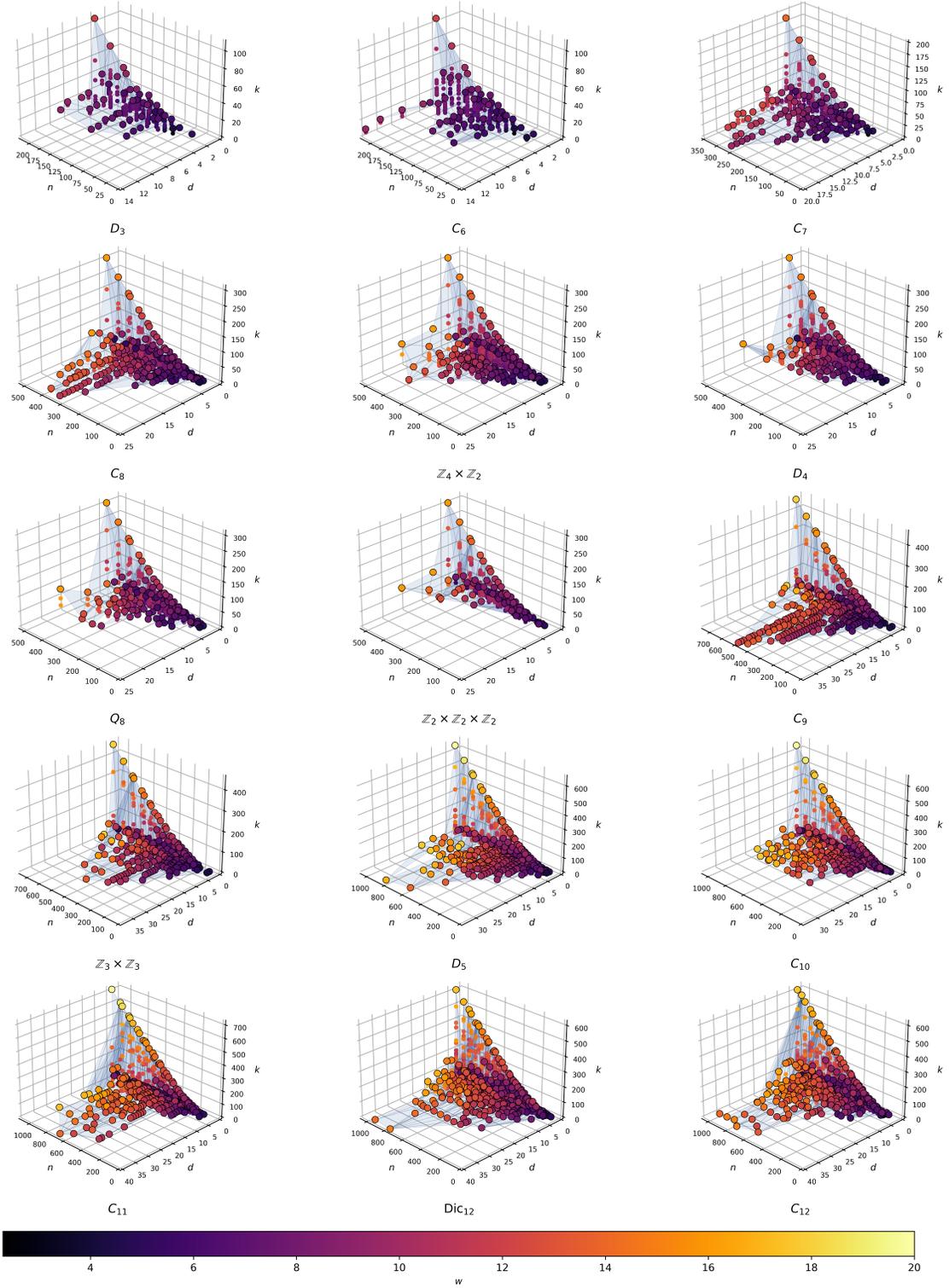}
  \caption{Achievable $(n,k,d)$ for explicit quantum Tanner codes with $w\le 20$ grouped by base group $G$. Each panel shows all generated instances constructed from a fixed group $G$, and the color of a point is the check weight of the code. The translucent surface is the achievable envelope: for each $(n,d)$, the maximum $k$ over all codes with the given $n$ and $d$ with $w\le 20$. Scales are the same for groups of the same order for comparison purpose. Note that we also filter out codes with $d>40$ as the \textsf{QDistRnd} distance estimation in this regime is less reliable.}
  \label{fig:3d_frontier}
\end{figure*}

\begin{table*}[h!]
\centering
\caption{Quantum Tanner code explicit instances}
\label{tab:qtcode_instances}

\setlength{\tabcolsep}{4pt}
\renewcommand{\arraystretch}{1.05}

\resizebox{\textwidth}{!}{%
\begin{tabular}{@{} c c c c c c c c @{}}
\toprule
\textbf{Group} &
\code{n,k,(d_X,d_Z)} &
$C_A$ &
$C_B$ &
$[\bar w, \bar q]$ &
$[w_X, q_X, w_Z, q_Z]$ &
$[\bar w_X, \bar q_X, \bar w_Z, \bar q_Z]$ &
$kd^2/(n\bar{w})$ \\
\midrule

$D_3$ & \code{72,19,(4,4)} & $[3,1,3]$ & $[4,3,2]$ & $[6.8, 2.8]$  & $[6,4,8,3]$ & $[6.0,\; 3.0,\; 8.0,\; 2.7]$ & $0.621$ \\
& \code{96,30,(4,4)}   & $[4,1,4]$ & $[4,3,2]$ & $[8.0,\; 3.0]$   & $[8,4,8,4]$    & $[8.0,\; 3.0,\; 8.0,\; 3.0]$   & $0.625$ \\
& \code{144,12,(7,7)}        & $[4,2,2]$ & $[6,3,3]$ & $[7.5,\; 3.8]$ & $[9,7,12,9]$ & $[7.5,\; 3.8,\; 7.5,\; 3.8]$ & $0.544$ \\
& \code{180,16,(9,7)}        & $[5,2,3]$ & $[6,3,3]$ & $[8.1,\; 4.0]$ & $[16,8,12,8]$ & $[9.6,\; 3.9,\; 7.1,\; 4.2]$ & $0.538$ \\
& \code{216,20,(8,8)}        & $[6,3,3]$ & $[6,3,3]$ & $[9.3,\; 4.7]$& $[16,9,12,9]$ & $[9.5,\; 4.8,\; 9.1,\; 4.6]$ & $0.637$ \\

\addlinespace[0.35em]
$C_6$ & \code{120,23,(4,6)}    & $[4,3,2]$ & $[5,2,3]$ & $[7.3,\; 3.3]$ & $[8,8,12,3]$     & $[6.3,\; 3.8,\; 9.3,\; 2.8]$   & $0.420$ \\
     & \code{144,6,(10,9)}     & $[4,2,2]$ & $[6,3,3]$ & $[7.6,\; 3.8]$ & $[12,7,9,8]$  & $[7.8,\; 3.9,\; 7.5,\; 3.8]$ & $0.444$ \\
     & \code{150,62,(4,4)}     & $[5,1,5]$ & $[5,4,2]$ & $[10.0,\; 3.2]$& $[10,5,10,5]$    & $[10.0,\; 3.2,\; 10.0,\; 3.2]$& $0.661$ \\
     & \code{216,10,(15,14)}   & $[6,3,3]$ & $[6,3,3]$ & $[9.5,\; 4.8]$& $[16,9,16,11]$  & $[9.5,\; 4.7,\; 9.5,\; 4.8]$& $0.955$ \\

\addlinespace[0.35em]
$C_7$ & \code{245,22,(11,8)}  & $[5,2,3]$ & $[7,4,3]$ & $[9.7,\; 4.7]$ & $[12,7,12,8]$   & $[10.1,\; 4.6,\; 9.3,\; 4.8]$ & $0.592$ \\
     & \code{343,22,(14,14)} & $[7,3,4]$ & $[7,4,3]$ & $[12.3,\; 6.0]$ & $[16,10,16,12]$ & $[12.2,\; 6.0,\; 12.3,\; 6.0]$ & $1.022$ \\

\addlinespace[0.35em]
$C_8$ & \code{64,22,(4,4)}     & $[4,1,4]$ & $[4,3,2]$ & $[8.0,\; 3.0]$  & $[8,4,8,4]$     & $[8.0,\; 3.0,\; 8.0,\; 3.0]$  & $0.688$ \\
    & \code{128,40,(4,4)}     & $[4,1,4]$ & $[4,3,2]$ & $[8.0,\; 3.0]$  & $[8,4,8,4]$     & $[8.0,\; 3.0,\; 8.0,\; 3.0]$  & $0.625$ \\
     & \code{240,10,(19,8)}    & $[5,2,3]$ & $[6,3,3]$ & $[9.3,\; 4.7]$  & $[16,7,12,7]$    & $[11.7,\; 4.7,\; 7.8,\; 4.7]$ & $0.287$ \\
     & \code{256,24,(10,10)}   & $[4,2,2]$ & $[8,4,4]$ & $[11.0,\; 5.5]$ & $[12,12,12,9]$    & $[12.0,\; 6.0,\; 10.0,\; 5.0]$ & $0.852$ \\
     & \code{256,28,(8,8)}    & $[4,2,2]$ & $[8,4,4]$ & $[11.0,\; 5.5]$ & $[12,12,12,9]$    & $[12.0,\; 6.0,\; 10.0,\; 5.0]$ & $0.636$ \\
     & \code{280,31,(10,8)}    & $[5,2,3]$ & $[7,4,3]$ & $[9.7,\; 4.7]$ & $[16,7,12,7]$    & $[10.2,\; 4.7,\; 9.3,\; 4.8]$ & $0.730$ \\
     & \code{384,28,(12,12)}   & $[6,3,3]$ & $[8,4,4]$ & $[13.3,\; 6.7]$ & $[16,18,16,18]$  & $[13.3,\; 6.7,\; 13.3,\; 6.7]$ & $0.789$ \\
     & \code{392,20,(17,17)}   & $[7,3,4]$ & $[7,4,3]$ & $[13.0,\; 6.4]$ & $[16,10,16,10]$  & $[13.0,\; 6.4,\; 13.0,\; 6.4]$ & $1.133$ \\
     & \code{392,48,(12,12)}   & $[7,3,4]$ & $[7,4,3]$ & $[13.0,\; 6.4]$ & $[16,10,16,10]$  & $[13.0,\; 6.4,\; 13.0,\; 6.4]$ & $1.356$ \\
     & \code{448,26,(16,12)}   & $[7,3,4]$ & $[8,4,4]$ & $[14.3,\; 7.1]$ & $[16,12,16,12]$  & $[16.0,\; 6.9,\; 13.0,\; 7.4]$ & $0.585$ \\
     & \code{448,42,(16,8)}    & $[7,2,4]$ & $[8,4,4]$ & $[12.6,\; 6.3]$ & $[20,9,12,18]$   & $[20.0,\; 5.7,\; 9.6,\; 6.9]$  & $0.477$ \\

\addlinespace[0.35em]
$Z_4\times Z_2$ & \code{96,10,(8,4)}    & $[3,1,3]$ & $[4,2,2]$ & $[5.8,\; 2.9]$  & $[9,4,6,6]$     & $[7.5,\; 2.5,\; 5.0,\; 3.3]$   & $0.287$ \\
     & \code{96,24,(4,4)}      & $[3,1,3]$ & $[4,3,2]$ & $[6.8,\; 2.8]$  & $[6,4,8,3]$     & $[6.0,\; 3.0,\; 8.0,\; 2.7]$   & $0.588$ \\
     & \code{192,16,(3,16)}    & $[3,2,2]$ & $[8,4,4]$ & $[6.7,\; 3.3]$  & $[8,6,8,6]$     & $[6.0,\; 4.0,\; 8.0,\; 2.7]$   & $0.112$ \\
     & \code{384,20,(16,16)}   & $[6,3,3]$ & $[8,4,4]$ & $[13.3,\; 6.7]$ & $[16,12,16,12]$ & $[13.3,\; 6.7,\; 13.3,\; 6.7]$ & $1.001$ \\
     & \code{384,48,(12,12)}   & $[6,3,3]$ & $[8,4,4]$ & $[13.3,\; 6.6]$ & $[16,12,16,13]$ & $[13.3,\; 6.7,\; 13.3,\; 6.6]$ & $1.353$ \\
     & \code{448,36,(12,16)}   & $[7,4,3]$ & $[8,4,4]$ & $[14.3,\; 7.1]$ & $[16,18,16,15]$ & $[13.0,\; 7.4,\; 16.0,\; 6.9]$ & $0.815$ \\
     & \code{512,82,(8,8)}     & $[8,4,4]$ & $[8,4,4]$ & $[16.0,\; 8.0]$ & $[16,18,16,18]$ & $[16.0,\; 8.0,\; 16.0,\; 8.0]$ & $0.641$ \\

\addlinespace[0.35em]
$D_4$ & \code{256,24,(10,10)}  & $[4,2,2]$ & $[8,4,4]$ & $[10.2,\; 5.1]$ & $[12,10,12,9]$   & $[10.4,\; 5.2,\; 10.0,\; 5.0]$ & $0.919$ \\
     & \code{392,54,(12,12)}  & $[7,3,4]$ & $[7,4,3]$ & $[12.9,\; 6.3]$ & $[16,18,16,15]$  & $[12.9,\; 6.3,\; 12.9,\; 6.3]$ & $1.559$ \\
     & \code{448,60,(12,16)}  & $[7,4,3]$ & $[8,4,4]$ & $[14.3,\; 7.1]$ & $[16,18,16,18]$  & $[13.0,\; 7.4,\; 16.0,\; 6.9]$ & $1.358$ \\
     & \code{512,76,(16,16)}  & $[8,4,4]$ & $[8,4,4]$ & $[16.0,\; 8.0]$ & $[16,18,16,18]$  & $[16.0,\; 8.0,\; 16.0,\; 8.0]$ & $2.375$ \\

\addlinespace[0.35em]
$Q_8$ & \code{96,24,(4,4)}    & $[3,1,3]$ & $[4,3,2]$ & $[6.8,\; 2.8]$  & $[6,4,8,3]$      & $[6.0,\; 3.0,\; 8.0,\; 2.7]$    & $0.588$ \\
     & \code{288,28,(12,12)}  & $[6,3,3]$ & $[6,3,3]$ & $[11.1,\; 5.6]$ & $[16,10,16,10]$  & $[11.1,\; 5.6,\; 11.1,\; 5.6]$  & $1.210$ \\
     & \code{384,44,(12,12)}  & $[6,3,3]$ & $[8,4,4]$ & $[13.3,\; 6.7]$ & $[16,12,16,12]$  & $[13.3,\; 6.7,\; 13.3,\; 6.7]$  & $1.241$ \\
     & \code{392,52,(12,12)}  & $[7,3,4]$ & $[7,4,3]$ & $[13.0,\; 6.4]$ & $[16,15,16,15]$  & $[13.0,\; 6.4,\; 13.0,\; 6.4]$  & $1.469$ \\
     & \code{448,60,(12,16)}  & $[7,4,3]$ & $[8,4,4]$ & $[14.3,\; 7.1]$ & $[16,18,16,15]$  & $[13.0,\; 7.4,\; 16.0,\; 6.9]$  & $1.350$ \\

\addlinespace[0.35em]
$\mathbb{Z}_2\times\mathbb{Z}_2\times\mathbb{Z}_2$ & \code{288,36,(8,8)}    & $[6,3,3]$ & $[6,3,3]$ & $[10.0,\; 5.0]$ & $[12,12,12,12]$ & $[9.9,\; 5.0,\; 10.0,\; 5.0]$ & $0.800$ \\
     & \code{384,56,(12,12)}  & $[6,3,3]$ & $[8,4,4]$ & $[13.3,\; 6.7]$ & $[16,18,16,18]$ & $[13.3,\; 6.7,\; 13.3,\; 6.7]$ & $1.579$ \\
     & \code{392,50,(16,8)}   & $[7,3,4]$ & $[7,3,4]$ & $[12.5,\; 6.4]$ & $[16,18,16,18]$ & $[16.0,\; 5.9,\; 10.6,\; 6.9]$ & $0.653$ \\
     & \code{392,59,(12,12)}  & $[7,3,4]$ & $[7,4,3]$ & $[13.0,\; 6.4]$ & $[16,18,16,18]$ & $[13.0,\; 6.4,\; 13.0,\; 6.4]$ & $1.667$ \\
     & \code{448,68,(12,16)}  & $[7,4,3]$ & $[8,4,4]$ & $[14.3,\; 7.1]$ & $[16,18,16,18]$ & $[13.0,\; 7.4,\; 16.0,\; 6.9]$ & $1.528$ \\
     & \code{512,80,(16,16)}    & $[8,4,4]$ & $[8,4,4]$ & $[16.0,\; 8.0]$ & $[16,18,16,18]$ & $[16.0,\; 8.0,\; 16.0,\; 8.0]$ & $2.500$ \\

\addlinespace[0.35em]
$C_9$ & \code{189,24,(8,4)}   & $[3,1,3]$ & $[7,4,3]$ & $[8.4,\; 4.0]$  & $[9,5,8,8]$    & $[9.0,\; 3.4,\; 8.0,\; 4.6]$   & $0.242$ \\
& \code{324,8,(26,28)}   & $[6,3,3]$ & $[6,3,3]$ & $[11.1,\; 5.5]$ & $[12,7,16,10]$ & $[10.0,\; 5.0,\; 11.1,\; 5.5]$ & $1.504$ \\
     & \code{432,16,(28,26)}  & $[6,3,3]$ & $[8,4,4]$ & $[13.3,\; 6.7]$ & $[16,10,16,10]$ & $[13.3,\; 6.7,\; 13.3,\; 6.7]$ & $1.882$ \\
     & \code{441,21,(40,28)}  & $[7,3,4]$ & $[7,4,3]$ & $[13.3,\; 6.7]$ & $[16,10,16,10]$ & $[13.3,\; 6.7,\; 13.3,\; 6.7]$ & $2.807$ \\
\addlinespace[0.35em]
$\mathbb{Z}_3\times\mathbb{Z}_3$ & \code{324,12,(18,18)}  & $[6,3,3]$ & $[6,3,3]$ & $[11.1,\; 5.6]$ & $[16,13,16,13]$ & $[11.1,\; 5.6,\; 11.1,\; 5.6]$ & $1.081$ \\
     & \code{432,8,(24,34)}   & $[6,3,3]$ & $[8,4,4]$ & $[13.3,\; 6.7]$ & $[16,15,16,15]$ & $[13.3,\; 6.7,\; 13.3,\; 6.7]$ & $0.802$ \\

\bottomrule
\end{tabular}%
}

\end{table*}

\begin{table*}[h!]
\centering
\caption{Quantum Tanner code explicit instances (continued)}
\label{tab:qtcode_instances_8_5}

\setlength{\tabcolsep}{4pt}
\renewcommand{\arraystretch}{1.05}

\resizebox{\textwidth}{!}{%
\begin{tabular}{@{} c c c c c c c c @{}}
\toprule
\textbf{Group} &
\code{n,k,(d_X,d_Z)} &
$C_A$ &
$C_B$ &
$[\bar w, \bar q]$ &
$[w_X, q_X, w_Z, q_Z]$ &
$[\bar w_X, \bar q_X, \bar w_Z, \bar q_Z]$ &
$kd^2/(n\bar{w})$ \\
\midrule
$D_5$ & \code{120,26,(4,4)}   & $[3,1,3]$ & $[4,3,2]$ & $[6.8,\; 2.8]$  & $[6,4,8,3]$     & $[6.0,\; 3.0,\; 8.0,\; 2.7]$   & $0.510$ \\
& \code{500,23,(15,17)} & $[5,2,3]$  & $[10,6,4]$ & $[11.1,\; 5.3]$ & $[16,12,15,9]$    & $[12.2,\; 5.9,\; 9.9,\; 4.8]$   & $0.934$ \\
      & \code{600,46,(15,12)} & $[6,2,4]$  & $[10,6,4]$ & $[12.1,\; 5.6]$ & $[16,16,15,12]$   & $[14.0,\; 5.6,\; 10.6,\; 5.7]$ & $0.915$ \\
      & \code{640,30,(16,16)} & $[8,4,4]$  & $[8,4,4]$  & $[16.0,\; 8.0]$ & $[16,18,16,18]$   & $[16.0,\; 8.0,\; 16.0,\; 8.0]$ & $0.750$ \\
      & \code{700,33,(15,19)} & $[7,3,4]$  & $[10,6,4]$ & $[13.9,\; 6.8]$ & $[16,24,20,18]$   & $[14.0,\; 7.2,\; 13.8,\; 6.3]$ & $0.762$ \\
\addlinespace[0.35em]
$C_{10}$& \code{360,16,(16,16)} & $[6,3,3]$ & $[6,3,3]$ & $[11.1,\; 5.6]$ & $[16,11,16,11]$ & $[11.1,\; 5.6,\; 11.1,\; 5.6]$ & $1.024$ \\
& \code{480,20,(20,20)} & $[6,3,3]$ & $[8,4,4]$ & $[13.3,\; 6.7]$ & $[16,12,16,12]$ & $[13.3,\; 6.7,\; 13.3,\; 6.7]$ & $1.500$ \\
& \code{630,22,(17,29)} & $[7,4,3]$ & $[9,4,5]$ & $[13.3,\; 6.5]$ & $[16,18,16,15]$ & $[13.0,\; 6.6,\; 13.6,\; 6.5]$ & $0.759$ \\
& \code{640,48,(16,16)} & $[8,4,4]$ & $[8,4,4]$ & $[16.0,\; 8.0]$ & $[16,18,16,18]$ & $[16.0,\; 8.0,\; 16.0,\; 8.0]$ & $1.200$ \\
& \code{720,30,(16,16)} & $[8,4,4]$ & $[9,4,5]$ & $[14.7,\; 7.3]$ & $[16,18,16,18]$ & $[16.0,\; 7.1,\; 13.6,\; 7.6]$ & $0.727$ \\
\addlinespace[0.35em]

$C_{11}$ & \code{396,8,(25,24)}   & $[6,3,3]$  & $[6,3,3]$  & $[11.1,\; 5.6]$ & $[16,10,16,10]$ & $[11.1,\; 5.6,\; 11.1,\; 5.6]$ & $1.047$ \\
& \code{528,16,(26,24)}  & $[6,3,3]$  & $[8,4,4]$  & $[13.3,\; 6.7]$ & $[16,10,16,10]$ & $[13.3,\; 6.7,\; 13.3,\; 6.7]$ & $1.309$ \\
& \code{528,16,(32,36)}  & $[6,3,3]$  & $[8,4,4]$  & $[13.3,\; 6.7]$ & $[16,10,16,10]$ & $[13.3,\; 6.7,\; 13.3,\; 6.7]$ & $2.327$ \\
& \code{539,23,(19,19)}  & $[7,3,4]$  & $[7,4,3]$  & $[13.0,\; 6.4]$ & $[16,10,16,10]$ & $[13.0,\; 6.4,\; 13.0,\; 6.4]$ & $1.185$ \\
& \code{770,41,(20,19)}  & $[7,3,4]$  & $[10,6,4]$ & $[13.9,\; 6.8]$ & $[16,13,20,10]$ & $[14.0,\; 7.2,\; 13.8,\; 6.3]$ & $1.382$ \\

\addlinespace[0.35em]
${\rm Dic}_{12}$ & \code{576,32,(16,16)}   & $[6,3,3]$ & $[8,4,4]$    & $[13.3,\; 6.7]$ & $[16,15,16,15]$ & $[13.3,\; 6.7,\; 13.3,\; 6.7]$  & $1.067$ \\
         & \code{672,38,(16,16)}   & $[7,3,4]$ & $[8,4,4]$    & $[14.3,\; 7.1]$ & $[16,12,16,12]$ & $[16.0,\; 6.9,\; 13.0,\; 7.4]$  & $1.013$ \\
         & \code{792,36,(25,18)}   & $[6,3,3]$ & $[11,7,4]$   & $[13.6,\; 6.8]$ & $[16,17,24,13]$ & $[11.9,\; 7.6,\; 16.7,\; 6.1]$  & $1.080$ \\
         & \code{924,48,(20,22)}   & $[7,4,3]$ & $[11,4,7]$   & $[16.2,\; 7.8]$ & $[28,14,20,16]$ & $[19.5,\; 8.1,\; 13.7,\; 7.5]$  & $1.281$ \\
         & \code{1008,73,(20,16)}  & $[7,4,3]$ & $[12,4,8]$   & $[16.9,\; 8.0]$ & $[32,14,20,16]$ & $[21.1,\; 8.0,\; 14.0,\; 8.0]$  & $1.100$ \\

\addlinespace[0.35em]
$C_{12}$ & \code{576,32,(26,24)}   & $[6,3,3]$  & $[8,4,4]$   & $[13.3,\; 6.7]$ & $[16,15,16,15]$ & $[13.3,\; 6.7,\; 13.3,\; 6.7]$ & $2.406$ \\
                       & \code{792,100,(12,10)}  & $[6,2,4]$  & $[11,7,4]$  & $[13.4,\; 6.1]$ & $[16,16,18,11]$ & $[14.3,\; 6.1,\; 12.5,\; 6.1]$ & $0.942$ \\
                       & \code{864,44,(26,66)}   & $[6,3,3]$  & $[12,4,8]$  & $[16.7,\; 7.5]$ & $[32,15,20,25]$ & $[21.7,\; 7.2,\; 11.7,\; 7.8]$ & $2.061$ \\
                       & \code{924,38,(20,20)}   & $[7,4,3]$  & $[11,4,7]$  & $[16.6,\; 7.8]$ & $[28,14,20,17]$ & $[19.5,\; 8.1,\; 13.7,\; 7.5]$ & $0.991$ \\
                       & \code{1008,113,(16,12)} & $[7,3,4]$  & $[12,7,4]$  & $[16.5,\; 8.0]$ & $[16,21,24,18]$ & $[16.0,\; 8.0,\; 16.9,\; 8.0]$ & $0.981$ \\

\bottomrule
\end{tabular}%
}
\end{table*}

\begin{table*}[ht]
\centering
\caption{Other explicit quantum Tanner codes from the literature. Dashes indicate quantities which are not presented.}
\label{tab:qt_lit}

\setlength{\tabcolsep}{5pt}
\renewcommand{\arraystretch}{1.15}

\begin{tabular}{@{} c c c c c c c c c @{}}
\toprule
\textbf{Reference} & \textbf{Group} &
\code{n,k,(d_X,d_Z)} &
$C_X$ &
$C_Z$ &
$\bar{w}$ &
$[\bar{w}_X, \bar{q}_X]$ &
$[\bar{w}_Z, \bar{q}_Z]$  &
$k d^2/(n w)$ \\
\midrule

Leverrier et al.~\cite{leverrier2025smallquantumtannercodes}& $C_{2}\!\times\! C_{2}$ & \code{144,8,(12,12)}   & $[6,3,3]$ & $[6,3,3]$ & $9$ & $[9,9]$ & $[9,9]$ & $0.889$ \\

 & $C_{2}\!\times\! C_{2}$ & \code{144,12,(11,11)}  & $[6,3,3]$ & $[6,3,3]$ & $9$ & $[9,9]$ & $[9,9]$ & $1.12$ \\

&$C_{6}$                & \code{216,8,(18,18)}   & $[6,3,3]$ & $[6,3,3]$ & $9$ & $[9,9]$ & $[9,9]$ &
$1.33$ \\

&$Q_{8}$                & \code{288,16,(16,16)}  & $[6,3,3]$ & $[6,3,3]$ & $9$ & $[9,9]$ & $[9,9]$  & $1.58$ \\

&$C_{8}$                & \code{288,8,(19,19)}   & $[6,3,3]$ & $[6,3,3]$ & $9$ & $[9,9]$ & $[9,9]$ &
 $1.11$ \\

&$C_{9}$                & \code{324,4,(26,26)}   & $[6,3,3]$ & $[6,3,3]$ & $9$ & $[9,9]$ & $[9,9]$ &
 $0.927$ \\

&$C_{11}$               & \code{396,2,(29,29)}   & $[6,3,3]$ & $[6,3,3]$ & $9$ & $[9,9]$ & $[9,9]$ &
 $0.472$ \\

&$C_{6}\!\times\! C_{2}$& \code{432,20,(22,22)}  & $[6,3,3]$ & $[6,3,3]$ & $9$ & $[9,9]$ & $[9,9]$ &
 $2.49$ \\

&$C_{6}\!\times\! C_{2}$& \code{432,24,(18,18)}  & $[6,3,3]$ & $[6,3,3]$ & $9$ & $[9,9]$ & $[9,9]$ &
 $2$ \\

&$C_{4}\rtimes C_{4}$   & \code{576,28,(24,24)}  & $[6,3,3]$ & $[6,3,3]$ & $9$ & $[9,9]$ & $[9,9]$ &
 $3.11$ \\

\midrule

Radebold et al.~\cite{radebold2025explicitinstancesquantumtanner} & $D_{4}$  & \code{36,8,(3,3)}     & -- & -- & $6$ &
$[6,\;2.7]$ & $[6,\;2.7]$ &
 $0.333$ \\

& $D_{6}$  & \code{54,11,(4,4)}   & -- & -- & $6$ &
$[6,\;2.7]$ & $[6,\;2.7]$ &
 $0.543$ \\

& $D_{8}$  & \code{72,14,(4,4)}    & -- & -- & $6$ &
$[6,\;2.7]$ & $[6,\;2.7]$ &
 $0.519$ \\

& $D_{8}$  & \code{200,10,(10,10)} & -- & -- & 8.1 &
$[8.2,\;3.9]$ & $[8.0,\;3.8]$ &
 0.619 \\

& $D_{10}$ & \code{250,10,(15,15)} & -- & -- & 8.1 & $[8.2,\;3.8]$ & $[8.0,\;3.9]$ &
 1.11 \\

\bottomrule
\end{tabular}
\end{table*}

\begin{table*}[ht]
\centering
\caption{Explicit instances of BB-code-related constructions}
\label{tab:bb_codes}

\renewcommand{\arraystretch}{1.3}
\small

\begin{tabularx}{\textwidth}{@{} >{\bfseries\raggedright\arraybackslash}p{3.8cm} >{\raggedright\arraybackslash}X c >{\raggedright\arraybackslash}p{5cm} @{}}
\toprule
Paper/Construction & 
\textbf{Codes \code{n,k,d}} & 
\textbf{$w$} & 
\textbf{Notes} \\
\midrule
Vanilla BB codes~\cite{Bravyi_2024} & \code{72,12,6}, \code{90,8,10}, \code{108,8,10}, \code{144,12,12}, \code{288,12,18} & $6$ & Original high-threshold BB code family with weight-6 checks.

\\

Coprime BB codes~\cite{wang2024coprime}
& \code{54,8,6}, \code{98,6,12}, \code{126,8,10}, \code{150,16,8}, \code{162,8,14}, \code{180,8,16}, \code{30,4,6}, \code{42,6,6}, \code{70,6,8}, \code{108,12,6}, \code{126,12,10}, \code{154,6,16}
& $6$
& A subclass of BB codes constructed using coprime polynomial factors, enabling efficient search and predictable parameter selection. \\
\addlinespace

Logical operators and fold-transversal gates~\cite{eberhardt2024logical}  
& \code{98,6,12}, \code{162,8,12}
& $6$
& Provide explicit logical operator bases for certain high-rate BB codes and construct fold-transversal Clifford gates. \\
\addlinespace

Existence and characterization~\cite{postema2025existence} 
& \code{150,4,10}, \code{198,4,12}, \code{270,4,16}, \code{378,6,12}, \code{186,10,6}, \code{310,10,14}, \code{438,18,8}, \code{730,18,10}, \code{18,4,2}, \code{18,8,2}, \code{98,6,8}, \code{196,6,12}, \code{450,6,12}
& $6$
& Characterize when BB codes exist and predict dimensions and distances via ring structure; supports listed weight-6 BB code instances. \\
\addlinespace

Tile codes on a lattice with boundary~\cite{steffan2025tilecodeshighefficiencyquantum} 
& \code{288,8,12}
& $6$
& Planar 2D-lattice tile code generalizing surface codes, maintaining weight-6 local checks. \\
& \code{288,8,14}
& $8$
& Variants with weight-8 checks on a planar 2D lattice; trade slightly higher check weight for improved distance.\\
& \code{288,8,13}, \code{512,18,19} 
& $\geq 8$
& Relax strict locality, using checks of weight $w\geq8$ to achieve better distance and rate. \\
\addlinespace

Twisted-torus (generalized toric) BB codes~\cite{Liang_2025}  
& \code{120,8,12}, \code{186,14,10}, \code{210,14,12}, \code{248,10,18}, \code{254,14,16},  \code{294,10,20}, \code{310,10,22}, \code{340,16,18}, \code{360,8,24}
& $6$
& Introduce a ring-theoretic framework for topological CSS codes on twisted tori. \\
\addlinespace

Planar open-boundary BB codes~\cite{liang2025planarquantumlowdensityparitycheck}  
& \code{78,6,6}, \code{107,7,7}, \code{268,8,12}, \code{405,9,15}, \code{348,10,13},  \code{450,11,15}, \code{386,12,12}, \code{362,13,11}
& $6$
& Introduce a lattice-grafting method to adapt BB codes to open boundaries with local weight-6 checks, yielding planar codes with comparable $kd^2/n$ and improved implementability. \\
& \code{282, 12, 14}
& $8$
& Single open-boundary code using weight-8 stabilizers. \\
\addlinespace

Self-dual BB codes~\cite{liang2025selfdual} 
& \code{16,4,4}, \code{40,6,6}, \code{56,6,8}, \code{64,8,8}, \code{120,8,12}, \code{152,6,16}, \code{160,8,16}
& $8$
& Self-dual constructions admitting transversal Clifford gates. \\
\addlinespace

Covering-graph ($h$-cover) BB codes~\cite{symons2025BBcovering} 
& \code{64,14,8}, \code{144,14,14}, \code{128,14,12}, \code{192,14,20}, \code{168,10,18}, \code{120,10,14}, \code{144,10,16}
& $8$
& Weight-8 BB codes generated via covering graph constructions. \\
\addlinespace

\bottomrule
\end{tabularx}
\end{table*}

\clearpage
\bibliographystyle{alphaurl}
\bibliography{ref_truncate10}

@article{SalesRodriguez2025,
  title = {Experimental demonstration of logical magic state distillation},
  volume = {645},
  ISSN = {1476-4687},
  url = {http://dx.doi.org/10.1038/s41586-025-09367-3},
  DOI = {10.1038/s41586-025-09367-3},
  number = {8081},
  journal = {Nature},
  publisher = {Springer Science and Business Media LLC},
  author = {Sales Rodriguez,  Pedro and Robinson,  John M. and Jepsen,  Paul Niklas and He,  Zhiyang and Duckering,  Casey and Zhao,  Chen and Wu,  Kai-Hsin and Campo,  Joseph and Bagnall,  Kevin and Kwon,  Minho and others},
  year = {2025},
  month = jul,
  pages = {620–625}
}

@article{Bluvstein2023,
  title = {Logical quantum processor based on reconfigurable atom arrays},
  volume = {626},
  ISSN = {1476-4687},
  url = {http://dx.doi.org/10.1038/s41586-023-06927-3},
  DOI = {10.1038/s41586-023-06927-3},
  number = {7997},
  journal = {Nature},
  publisher = {Springer Science and Business Media LLC},
  author = {Bluvstein,  Dolev and Evered,  Simon J. and Geim,  Alexandra A. and Li,  Sophie H. and Zhou,  Hengyun and Manovitz,  Tom and Ebadi,  Sepehr and Cain,  Madelyn and Kalinowski,  Marcin and Hangleiter,  Dominik and others},
  year = {2023},
  month = dec,
  pages = {58–65}
}

@article{RyanAnderson2021,
  title = {Realization of Real-Time Fault-Tolerant Quantum Error Correction},
  volume = {11},
  ISSN = {2160-3308},
  url = {http://dx.doi.org/10.1103/PhysRevX.11.041058},
  DOI = {10.1103/physrevx.11.041058},
  number = {4},
  journal = {Physical Review X},
  publisher = {American Physical Society (APS)},
  author = {Ryan-Anderson,  C. and Bohnet,  J. G. and Lee,  K. and Gresh,  D. and Hankin,  A. and Gaebler,  J. P. and Francois,  D. and Chernoguzov,  A. and Lucchetti,  D. and Brown,  N. C. and others},
  year = {2021},
  month = dec 
}

@article{Putterman2025,
  title = {Hardware-efficient quantum error correction via concatenated bosonic qubits},
  volume = {638},
  ISSN = {1476-4687},
  url = {http://dx.doi.org/10.1038/s41586-025-08642-7},
  DOI = {10.1038/s41586-025-08642-7},
  number = {8052},
  journal = {Nature},
  publisher = {Springer Science and Business Media LLC},
  author = {Putterman,  Harald and Noh,  Kyungjoo and Hann,  Connor T. and MacCabe,  Gregory S. and Aghaeimeibodi,  Shahriar and Patel,  Rishi N. and Lee,  Menyoung and Jones,  William M. and Moradinejad,  Hesam and Rodriguez,  Roberto and others},
  year = {2025},
  month = feb,
  pages = {927–934}
}

@misc{bravyi1998quantumcodeslatticeboundary,
      title={Quantum codes on a lattice with boundary}, 
      author={S. B. Bravyi and A. Yu. Kitaev},
      year={1998},
      howpublished={arXiv:quant-ph/9811052},
    doi = {10.48550/arXiv.quant-ph/9811052}
}

@article{Shor1995decoherence,
  title = {Scheme for reducing decoherence in quantum computer memory},
  author = {Shor, Peter W.},
  journal = {Phys. Rev. A},
  volume = {52},
  issue = {4},
  pages = {R2493--R2496},
  numpages = {0},
  year = {1995},
  month = {Oct},
  publisher = {American Physical Society},
  doi = {10.1103/PhysRevA.52.R2493},
  url = {https://link.aps.org/doi/10.1103/PhysRevA.52.R2493}
}

@article{Steane1996,
  title = {Error Correcting Codes in Quantum Theory},
  volume = {77},
  ISSN = {1079-7114},
  url = {http://dx.doi.org/10.1103/PhysRevLett.77.793},
  DOI = {10.1103/physrevlett.77.793},
  number = {5},
  journal = {Physical Review Letters},
  publisher = {American Physical Society (APS)},
  author = {Steane,  A. M.},
  year = {1996},
  month = jul,
  pages = {793–797}
}

@article{Raussendorf2012,
  title = {Key ideas in quantum error correction},
  volume = {370},
  ISSN = {1471-2962},
  url = {http://dx.doi.org/10.1098/rsta.2011.0494},
  DOI = {10.1098/rsta.2011.0494},
  number = {1975},
  journal = {Philosophical Transactions of the Royal Society A: Mathematical,  Physical and Engineering Sciences},
  publisher = {The Royal Society},
  author = {Raussendorf,  Robert},
  year = {2012},
  month = sep,
  pages = {4541–4565}
}

@article{baspin2024wire,
  title={Wire codes},
  author={Baspin, Nou{\'e}dyn and Williamson, Dominic},
  journal={arXiv preprint arXiv:2410.10194},
  url = {https://arxiv.org/abs/2410.10194},
  year={2024}
}

@article{bravyi2011subsystem,
    title = {Subsystem codes with spatially local generators},
    volume = {83},
    issn = {1094-1622},
    url = {http://dx.doi.org/10.1103/PhysRevA.83.012320},
    doi = {10.1103/physreva.83.012320},
    number = {1},
    journal = {Physical Review A},
    author = {Bravyi, Sergey},
    month = jan,
    year = {2011},
    note = {Publisher: American Physical Society (APS)},
}

@article{Fetaya12,
    author = {Fetaya, Ethan},
    title = {Bounding the distance of quantum surface codes},
    journal = {Journal of Mathematical Physics},
    volume = {53},
    number = {6},
    pages = {062202},
    year = {2012},
    month = {06},
    issn = {0022-2488},
    doi = {10.1063/1.4726034},
    url = {https://pubs.aip.org/aip/jmp/article-pdf/doi/10.1063/1.4726034/15812523/062202_1_online.pdf},
}

@INPROCEEDINGS {leverrier2022quantumtannercodes,
    author = { Leverrier, Anthony and Zemor, Gilles },
    booktitle = { 2022 IEEE 63rd Annual Symposium on Foundations of Computer Science (FOCS) },
    title = {{Quantum Tanner codes}},
    year = {2022},
    volume = {},
    ISSN = {},
    pages = {872-883},
    doi = {10.1109/FOCS54457.2022.00117},
    url = {https://doi.ieeecomputersociety.org/10.1109/FOCS54457.2022.00117},
    publisher = {IEEE Computer Society},
    address = {Los Alamitos, CA, USA},
    month =Nov}

@ARTICLE{LZ2023decoding,
  author={Leverrier, Anthony and Zémor, Gilles},
  journal={IEEE Transactions on Information Theory}, 
  title={Decoding Quantum Tanner Codes}, 
  year={2023},
  volume={69},
  number={8},
  pages={5100-5115},
  doi={10.1109/TIT.2023.3267945}}

@article{Bravyi_2024,
   title={High-threshold and low-overhead fault-tolerant quantum memory},
   volume={627},
   ISSN={1476-4687},
   url={http://dx.doi.org/10.1038/s41586-024-07107-7},
   DOI={10.1038/s41586-024-07107-7},
   number={8005},
   journal={Nature},
   publisher={Springer Science and Business Media LLC},
   author={Bravyi, Sergey and Cross, Andrew W. and Gambetta, Jay M. and Maslov, Dmitri and Rall, Patrick and Yoder, Theodore J.},
   year={2024},
   month=mar, pages={778–782} }

@article{Panteleev_2022,
   title={Quantum {LDPC} Codes With Almost Linear Minimum Distance},
   volume={68},
   ISSN={1557-9654},
   url={http://dx.doi.org/10.1109/TIT.2021.3119384},
   DOI={10.1109/tit.2021.3119384},
   number={1},
   journal={IEEE Transactions on Information Theory},
   publisher={Institute of Electrical and Electronics Engineers (IEEE)},
   author={Panteleev, Pavel and Kalachev, Gleb},
   year={2022},
   month=jan, pages={213–229} }

@inproceedings{PKAsymptoticGoodLDPC,
author = {Panteleev, Pavel and Kalachev, Gleb},
title = {Asymptotically good Quantum and locally testable classical {LDPC} codes},
year = {2022},
isbn = {9781450392648},
publisher = {Association for Computing Machinery},
address = {New York, NY, USA},
url = {https://doi.org/10.1145/3519935.3520017},
doi = {10.1145/3519935.3520017},
booktitle = {Proceedings of the 54th Annual ACM SIGACT Symposium on Theory of Computing},
pages = {375–388},
numpages = {14},
location = {Rome, Italy},
series = {STOC 2022}
}

@inproceedings{Dinur_goodqLDPCCodes,
author = {Dinur, Irit and Hsieh, Min-Hsiu and Lin, Ting-Chun and Vidick, Thomas},
title = {Good Quantum {LDPC} Codes with Linear Time Decoders},
year = {2023},
isbn = {9781450399135},
publisher = {Association for Computing Machinery},
address = {New York, NY, USA},
url = {https://doi.org/10.1145/3564246.3585101},
doi = {10.1145/3564246.3585101},
booktitle = {Proceedings of the 55th Annual ACM Symposium on Theory of Computing},
pages = {905–918},
numpages = {14},
location = {Orlando, FL, USA},
series = {STOC 2023}
}

@misc{hastings2023quantumweightreduction,
      title={On Quantum Weight Reduction}, 
      author={M. B. Hastings},
      year={2023},
      eprint={2102.10030},
      archivePrefix={arXiv},
      primaryClass={quant-ph},
      url={https://arxiv.org/abs/2102.10030}, 
}

@article{Macwilliams1963,
  title={A theorem on the distribution of weights in a systematic code},
  author={Jessie Macwilliams},
  journal={Bell System Technical Journal},
  year={1963},
  volume={42},
  pages={79-94},
  url={https://api.semanticscholar.org/CorpusID:121752366},
doi = {DOI:10.1002/J.1538-7305.1963.TB04003.X}
}

@article{Gottesman96stabilizercodes,
  title = {Class of quantum error-correcting codes saturating the quantum Hamming bound},
  author = {Gottesman, Daniel},
  journal = {Phys. Rev. A},
  volume = {54},
  issue = {3},
  pages = {1862--1868},
  numpages = {0},
  year = {1996},
  month = {Sep},
  publisher = {American Physical Society},
  doi = {10.1103/PhysRevA.54.1862},
  url = {https://link.aps.org/doi/10.1103/PhysRevA.54.1862}
}

@article{CS96CSScode,
  title = {Good quantum error-correcting codes exist},
  author = {Calderbank, A. R. and Shor, Peter W.},
  journal = {Phys. Rev. A},
  volume = {54},
  issue = {2},
  pages = {1098--1105},
  numpages = {0},
  year = {1996},
  month = {Aug},
  publisher = {American Physical Society},
  doi = {10.1103/PhysRevA.54.1098},
  url = {https://link.aps.org/doi/10.1103/PhysRevA.54.1098}
}

@article{Steane96CSScode,
    author = {Steane, Andrew},
    title = {Multiple-particle interference and quantum error correction},
    journal = {Proceedings of the Royal Society A: Mathematical, Physical and Engineering Sciences},
    volume = {452},
    number = {1954},
    pages = {2551-2577},
    year = {1996},
    month = {11},
    issn = {1364-5021},
    doi = {10.1098/rspa.1996.0136},
    url = {https://royalsocietypublishing.org/rspa/article-pdf/452/1954/2551/998878/rspa.1996.0136.pdf},
}

@article{Poulin05subsystem,
  title = {Stabilizer Formalism for Operator Quantum Error Correction},
  author = {Poulin, David},
  journal = {Phys. Rev. Lett.},
  volume = {95},
  issue = {23},
  pages = {230504},
  numpages = {4},
  year = {2005},
  month = {Dec},
  publisher = {American Physical Society},
  doi = {10.1103/PhysRevLett.95.230504},
  url = {https://link.aps.org/doi/10.1103/PhysRevLett.95.230504}
}

@article{calderbank1998quantum,
  title={Quantum error correction via codes over GF (4)},
  author={Calderbank, A Robert and Rains, Eric M and Shor, Peter M and Sloane, Neil JA},
  journal={IEEE Transactions on Information Theory},
  volume={44},
  number={4},
  pages={1369--1387},
  year={1998},
  publisher={IEEE},
  doi = {10.1109/18.681315}
}

@misc{gurobi,
	author = {{Gurobi Optimization, LLC}},
	title = {{Gurobi Optimizer Reference Manual}},
	year = 2024,
	url = "https://www.gurobi.com"
}

@article{Kitaev03anyons,
title = {Fault-tolerant quantum computation by anyons},
journal = {Annals of Physics},
volume = {303},
number = {1},
pages = {2-30},
year = {2003},
issn = {0003-4916},
doi = {https://doi.org/10.1016/S0003-4916(02)00018-0},
url = {https://www.sciencedirect.com/science/article/pii/S0003491602000180},
author = {A.Yu. Kitaev}
}

@article{Bombin07homological,
    author = {Bombin, H. and Martin-Delgado, M. A.},
    title = {Homological error correction: Classical and quantum codes},
    journal = {Journal of Mathematical Physics},
    volume = {48},
    number = {5},
    pages = {052105},
    year = {2007},
    month = {05},
    issn = {0022-2488},
    doi = {10.1063/1.2731356},
    url = {https://pubs.aip.org/aip/jmp/article-pdf/doi/10.1063/1.2731356/14848728/052105_1_online.pdf},
}

@manual{SmallGrp,
  author       = "Hans Ulrich Besche and Bettina Eick and Eamonn O'Brien and Max Horn",
  title        = "{SmallGrp: The GAP Small Groups Library, Version 1.5.4}",
  year         = "2024",
  note         = "(GAP package)",
  url          = "https://gap-packages.github.io/smallgrp/"
}

@Misc{Grassl:codetables,
  author       = "Markus Grassl",
  title        = "{Bounds on the minimum distance of linear codes and quantum codes}",
  howpublished = "Online at \url{http://www.codetables.de}",
  year         = "2007",
  note         = "Accessed: 2025-12-25"
}

@article{Pryadko_2022,
   title={Q{D}ist{R}nd: A {GAP} package for computing the distance of
quantum error-correcting codes},
   volume={7},
   ISSN={2475-9066},
   url={http://dx.doi.org/10.21105/joss.04120},
   DOI={10.21105/joss.04120},
   number={71},
   journal={Journal of Open Source Software},
   publisher={The Open Journal},
   author={Pryadko, Leonid P. and Shabashov, Vadim A. and Kozin, Valerii K.},
   year={2022},
   month=mar, pages={4120} }

@inbook{Stillwell93,
    address={New York, NY},
    title={Complex Analysis and Surface Topology},
    ISBN={9781461243724},
    url={https://doi.org/10.1007/978-1-4612-4372-4_2},
    DOI={10.1007/978-1-4612-4372-4_2},
    abstractNote={Topology may have had its tentative beginnings in isolated thoughts of Descartes, Leibniz, and Euler, but it was Riemann who brought the subject into the mainstream of mathematics with his inaugural dissertation in Göttingen in 1851. His introduction of the Riemann surface in that year showed the indispensable rôle of topology in questions of analysis, and thus ensured the future cultivation of the subject by the mathematical community, if only for the service of analysis. In fact, of course, Riemann surfaces were quickly seen to be of interest in themselves, and were the source of two ideas of profound significance in later topology—connectivity and covering spaces.},
    booktitle={Classical Topology and Combinatorial Group Theory},
    publisher={Springer},
    author={Stillwell, John},
    year={1993},
    pages={53–88},
    language={en} }

@INPROCEEDINGS{Delfosse13,
  author={Delfosse, Nicolas},
  booktitle={2013 IEEE International Symposium on Information Theory}, 
  title={Tradeoffs for reliable quantum information storage in surface codes and color codes}, 
  year={2013},
  volume={},
  number={},
  pages={917-921},
  doi={10.1109/ISIT.2013.6620360}}

@inbook{FML02,
    author={Freedman, Michael H. and Meyer, David A. and Luo, Feng},
    title={Z2-systolic freedom and quantum codes},
    url={https://www.taylorfrancis.com/chapters/edit/10.1201/9781420035377-13/z2-systolic-freedom-quantum-codes-michael-freedman-david-meyer-feng-luo},
    DOI={10.1201/9781420035377-13},
    abstractNote={Abstract A closely coupled pair of conjectures/questions-one in differential geometry (by M. Gromov), the other in quantum information theory-are both answered in the negative. The answer derives from a certain metrical ﬂexibility of manifolds and a corresponding improvement to the theoretical eﬃciency of existing local quantum codes. We exhibit this eﬀect by constructing a family of metrics on S2 × S1, and other three and four dimensional manifolds. Quantitatively, the explicit “freedom” exhibited is too weak (a log1/2 factor in the natural scaling) to yield practical codes but we cannot rule out the possibility of other families of geometries with more dramatic freedom.},
    booktitle={Mathematics of Quantum Computation},
    publisher={Chapman and Hall/CRC},
    year={2002},
    month=feb,
    pages={303–338},
    language={en} }

@article{Voss25multivariatebicycle,
  title = {Multivariate bicycle codes},
  author = {Voss, Lukas and Xian, Sim Jian and Haug, Tobias and Bharti, Kishor},
  journal = {Phys. Rev. A},
  volume = {111},
  issue = {6},
  pages = {L060401},
  numpages = {6},
  year = {2025},
  month = {Jun},
  publisher = {American Physical Society},
  doi = {10.1103/ll5p-z88p},
  url = {https://link.aps.org/doi/10.1103/ll5p-z88p}
}

@inproceedings{AharanovEldar,
    author = {Aharonov, Dorit and Eldar, Lior},
    title = {On the Complexity of Commuting Local Hamiltonians, and Tight Conditions for Topological Order in Such Systems},
    year = {2011},
    isbn = {9780769545714},
    publisher = {IEEE Computer Society},
    address = {USA},
    url = {https://doi.org/10.1109/FOCS.2011.58},
    doi = {10.1109/FOCS.2011.58},
    booktitle = {Proceedings of the 2011 IEEE 52nd Annual Symposium on Foundations of Computer Science},
    pages = {334–343},
    numpages = {10},
    series = {FOCS '11}
}

@misc{Krishnatillich,
title={Private Communication},
author={Anirudh Krishna and Jean-Pierre Tillich}}

@misc{KrishnaTillich25LP, type={MATLAB},
title={anirudh-krishna/LP-bounds},
url={https://github.com/anirudh-krishna/LP-bounds},
author={Anirudh Krishna and Jean-Pierre Tillich},
year={2025},
month=may }

@INPROCEEDINGS{Ben-HaimLitsyn,
  author={Ben-Haim, Y. and Litsyn, S.},
  booktitle={Proceedings. International Symposium on Information Theory, 2005. ISIT 2005.}, 
  title={Upper bounds on the rate of LDPC codes as a function of minimum distance}, 
  year={2005},
  volume={},
  number={},
  pages={47-51},
  doi={10.1109/ISIT.2005.1523290}}

@misc{symons2025BBcovering,
      title={Sequences of Bivariate Bicycle Codes from Covering Graphs}, 
      author={Benjamin C. B. Symons and Abhishek Rajput and Dan E. Browne},
      year={2025},
      eprint={2511.13560},
      archivePrefix={arXiv},
      primaryClass={quant-ph},
      url={https://arxiv.org/abs/2511.13560}, 
}

@article{wang2024coprime,
  title        = {Coprime Bivariate Bicycle Codes},
  author       = {Wang, Ming and Mueller, Frank},
  journal      = {arXiv preprint},
  eprint       = {2408.10001},
  eprinttype   = {arXiv},
  eprintclass  = {quant-ph},
  year         = {2024},
  url          = {https://arxiv.org/abs/2408.10001},
  note         = {Coprime subclass of BB codes with new short/medium-length instances}
}

@article{postema2025existence,
  title        = {Existence and Characterisation of Bivariate Bicycle Codes},
  author       = {Postema, Jasper Johannes and Kokkelmans, Servaas J.J.M.F.},
  journal      = {arXiv preprint},
  eprint       = {2502.17052},
  eprinttype   = {arXiv},
  eprintclass  = {quant-ph},
  year         = {2025},
  doi          = {10.48550/arXiv.2502.17052},
  url          = {https://arxiv.org/abs/2502.17052},
  note         = {Analysis of existence conditions and parameters for BB codes}
}

@misc{steffan2025tilecodeshighefficiencyquantum,
      title={Tile Codes: High-Efficiency Quantum Codes on a Lattice with Boundary}, 
      author={Vincent Steffan and Shin Ho Choe and Nikolas P. Breuckmann and Francisco Revson Fernandes Pereira and Jens Niklas Eberhardt},
      year={2025},
      eprint={2504.09171},
      archivePrefix={arXiv},
      primaryClass={quant-ph},
      url={https://arxiv.org/abs/2504.09171}, 
}

@article{eberhardt2024logical,
  title        = {Logical Operators and Fold-Transversal Gates of Bivariate Bicycle Codes},
  author       = {Eberhardt, Jens Niklas and Steffan, Vincent},
  journal      = {arXiv preprint},
  eprint       = {2407.03973},
  eprinttype   = {arXiv},
  eprintclass  = {quant-ph},
  year         = {2024},
  url          = {https://arxiv.org/abs/2407.03973},
  note         = {Fold-transversal Clifford gates for explicit bivariate bicycle codes}
}

@article{liang2025selfdual,
  title        = {Self-Dual Bivariate Bicycle Codes with Transversal Clifford Gates},
  author       = {Liang, Z. and Chen, Y.-A.},
  journal      = {arXiv preprint},
  eprint       = {2510.05211},
  eprinttype   = {arXiv},
  eprintclass  = {quant-ph},
  year         = {2025},
  url          = {https://arxiv.org/abs/2510.05211},
  note         = {Self-dual BB codes supportive of transversal Clifford operations}
}

@article{Liang_2025,
   title={Generalized Toric Codes on Twisted Tori for Quantum Error Correction},
   volume={6},
   ISSN={2691-3399},
   url={http://dx.doi.org/10.1103/rmy6-9n89},
   DOI={10.1103/rmy6-9n89},
   number={2},
   journal={PRX Quantum},
   publisher={American Physical Society (APS)},
   author={Liang, Zijian and Liu, Ke and Song, Hao and Chen, Yu-An},
   year={2025},
   month=jun }

@misc{liang2025planarquantumlowdensityparitycheck,
      title={Planar quantum low-density parity-check codes with open boundaries}, 
      author={Zijian Liang and Jens Niklas Eberhardt and Yu-An Chen},
      year={2025},
      eprint={2504.08887},
      archivePrefix={arXiv},
      primaryClass={quant-ph},
      url={https://arxiv.org/abs/2504.08887}, 
}

@ARTICLE{TillichZemorHGP,
  author={Tillich, Jean-Pierre and Zémor, Gilles},
  journal={IEEE Transactions on Information Theory}, 
  title={Quantum LDPC Codes With Positive Rate and Minimum Distance Proportional to the Square Root of the Blocklength}, 
  year={2014},
  volume={60},
  number={2},
  pages={1193-1202},
  doi={10.1109/TIT.2013.2292061}}

@inproceedings{hastings2020fiber,
  title={Fiber bundle codes: breaking the $n^{1/2} \text{polylog}(n)$ barrier for quantum {LDPC} codes},
  author={Hastings, Matthew B and Haah, Jeongwan and O'Donnell, Ryan},
  booktitle={Proceedings of the 53rd Annual ACM SIGACT Symposium on Theory of Computing},
  pages={1276--1288},
  year={2021},
  doi={10.1145/3406325.3451005}
}

@article{evra2020decodable,
	title={Decodable Quantum {LDPC} Codes beyond the $\sqrt{n}$ Distance Barrier Using High-Dimensional Expanders},
	ISSN={0097-5397, 1095-7111},
	DOI={10.1137/20M1383689},
	journal={SIAM Journal on Computing},
	author={Evra, Shai and Kaufman, Tali and Zémor, Gilles},
	year={2022},
	month=jun,
	pages={FOCS20-276-FOCS20-316}
}

@article{Breuckmann2020,
	doi = {10.1109/tit.2021.3097347},
  
	year = 2021,
	month = {oct},
  
	publisher = {Institute of Electrical and Electronics Engineers ({IEEE})},
  
	volume = {67},
  
	number = {10},
  
	pages = {6653--6674},
  
	author = {Nikolas P. Breuckmann and Jens N. Eberhardt},
  
	title = {Balanced Product Quantum Codes},
  
	journal = {{IEEE} Transactions on Information Theory}
}

@misc{lin2022goodquantumldpccodes,
      title={Good quantum LDPC codes with linear time decoder from lossless expanders}, 
      author={Ting-Chun Lin and Min-Hsiu Hsieh},
      year={2022},
    howpublished={arXiv:quant-ph/2203.03581},
    doi={10.48550/arXiv.2203.03581},
}

@misc{hsieh2025explicitlosslessvertexexpanders,
      title={Explicit Lossless Vertex Expanders}, 
      author={Jun-Ting Hsieh and Alexander Lubotzky and Sidhanth Mohanty and Assaf Reiner and Rachel Yun Zhang},
      year={2025},
    howpublished={arXiv:quant-ph/2504.15087},
    doi={10.48550/arXiv.2504.15087},
}

@article{nezami22triorthogonal,
  title = {Classification of small triorthogonal codes},
  author = {Nezami, Sepehr and Haah, Jeongwan},
  journal = {Phys. Rev. A},
  volume = {106},
  issue = {1},
  pages = {012437},
  numpages = {13},
  year = {2022},
  month = {Jul},
  publisher = {American Physical Society},
  doi = {10.1103/PhysRevA.106.012437},
  url = {https://link.aps.org/doi/10.1103/PhysRevA.106.012437}
}

@article{BravyiHaahMSD,
  title = {Magic-state distillation with low overhead},
  author = {Bravyi, Sergey and Haah, Jeongwan},
  journal = {Phys. Rev. A},
  volume = {86},
  issue = {5},
  pages = {052329},
  numpages = {10},
  year = {2012},
  month = {Nov},
  publisher = {American Physical Society},
  doi = {10.1103/PhysRevA.86.052329},
  url = {https://link.aps.org/doi/10.1103/PhysRevA.86.052329}
}

@misc{leverrier2025smallquantumtannercodes,
      title={Small quantum Tanner codes from left--right Cayley complexes}, 
      author={Anthony Leverrier and Wouter Rozendaal and Gilles Zémor},
      year={2025},
      eprint={2512.20532},
      archivePrefix={arXiv},
      primaryClass={quant-ph},
      url={https://arxiv.org/abs/2512.20532}, 
}

@phdthesis{Gottesman1997,
  doi = {10.7907/RZR7-DT72},
  url = {https://resolver.caltech.edu/CaltechETD:etd-07162004-113028},
  author = {Gottesman,  Daniel Eric},
  title = {Stabilizer Codes and Quantum Error Correction},
  publisher = {California Institute of Technology},
  year = {1997}
}

@article{Baspin2022connectivity,
  doi = {10.22331/q-2022-05-13-711},
  url = {https://doi.org/10.22331/q-2022-05-13-711},
  title = {Connectivity constrains quantum codes},
  author = {Baspin, Nou{\'{e}}dyn and Krishna, Anirudh},
  journal = {{Quantum}},
  issn = {2521-327X},
  publisher = {{Verein zur F{\"{o}}rderung des Open Access Publizierens in den Quantenwissenschaften}},
  volume = {6},
  pages = {711},
  month = may,
  year = {2022}
}

@inproceedings{Li25locality,
author = {Dai, Samuel and Li, Ray},
title = {Locality vs Quantum Codes},
year = {2025},
isbn = {9798400715105},
publisher = {Association for Computing Machinery},
address = {New York, NY, USA},
url = {https://doi.org/10.1145/3717823.3718113},
doi = {10.1145/3717823.3718113},
booktitle = {Proceedings of the 57th Annual ACM Symposium on Theory of Computing},
pages = {677–688},
numpages = {12},
location = {Prague, Czechia},
series = {STOC '25}
}

@misc{Li25subsystem,
      title={Optimal Locality and Parameter Tradeoffs for Subsystem Codes}, 
      author={Samuel Dai and Ray Li and Eugene Tang},
      year={2025},
      eprint={2503.22651},
      archivePrefix={arXiv},
      primaryClass={quant-ph},
      url={https://arxiv.org/abs/2503.22651}, 
}

@article{Baspin25improved,
doi = {10.1088/2058-9565/ad8370},
url = {https://doi.org/10.1088/2058-9565/ad8370},
year = {2024},
month = {oct},
publisher = {IOP Publishing},
volume = {10},
number = {1},
pages = {015021},
author = {Baspin, Nouédyn and Guruswami, Venkatesan and Krishna, Anirudh and Li, Ray},
title = {Improved rate-distance trade-offs for quantum codes with restricted connectivity},
journal = {Quantum Science and Technology}
}

@article{BPTbound,
  title = {Tradeoffs for Reliable Quantum Information Storage in {2D} Systems},
  author = {Bravyi, Sergey and Poulin, David and Terhal, Barbara},
  journal = {Phys. Rev. Lett.},
  volume = {104},
  issue = {5},
  pages = {050503},
  numpages = {4},
  year = {2010},
  month = {Feb},
  publisher = {American Physical Society},
  doi = {10.1103/PhysRevLett.104.050503},
}

@article{BTbound,
doi = {10.1088/1367-2630/11/4/043029},
year = {2009},
month = {apr},
publisher = {},
volume = {11},
number = {4},
pages = {043029},
author = {Bravyi, Sergey and Terhal, Barbara},
title = {A no-go theorem for a two-dimensional self-correcting quantum memory based on stabilizer codes},
journal = {New Journal of Physics}
}

@article{HaahBound,
   title={A degeneracy bound for homogeneous topological order},
   volume={10},
   ISSN={2542-4653},
   DOI={10.21468/scipostphys.10.1.011},
   number={1},
   journal={SciPost Physics},
   publisher={Stichting SciPost},
   author={Haah, Jeongwan},
   year={2021},
   month=jan }

@book{NielsenChuang10,
	address={Cambridge},
	edition={10th anniversary edition},
	title={Quantum computation and quantum information},
	ISBN={9781107002173},
	callNumber={530.12},
	publisher={Cambridge university press},
	author={Nielsen, Michael A. and Chuang, Isaac L.},
	year={2010},
	language={eng}
}

@article{BE21LDPC,
  title = {Quantum Low-Density Parity-Check Codes},
  author = {Breuckmann, Nikolas P. and Eberhardt, Jens Niklas},
  journal = {PRX Quantum},
  volume = {2},
  issue = {4},
  pages = {040101},
  numpages = {19},
  year = {2021},
  month = {Oct},
  publisher = {American Physical Society},
  doi = {10.1103/PRXQuantum.2.040101},
  url = {https://link.aps.org/doi/10.1103/PRXQuantum.2.040101}
}

@article{Saffman10Rydberg,
  title = {Quantum information with Rydberg atoms},
  author = {Saffman, M. and Walker, T. G. and M\o{}lmer, K.},
  journal = {Rev. Mod. Phys.},
  volume = {82},
  issue = {3},
  pages = {2313--2363},
  numpages = {0},
  year = {2010},
  month = {Aug},
  publisher = {American Physical Society},
  doi = {10.1103/RevModPhys.82.2313},
  url = {https://link.aps.org/doi/10.1103/RevModPhys.82.2313}
}

@article{Browaeys20Rydberg,
title={Many-body physics with individually controlled Rydberg atoms},
volume={16},
rights={2020 Springer Nature Limited},
ISSN={1745-2481},
url={https://www.nature.com/articles/s41567-019-0733-z},
DOI={10.1038/s41567-019-0733-z},
abstractNote={Recent decades have witnessed great developments in the field of quantum simulation—where synthetic systems are built and studied to gain insight into complicated, many-body real-world problems. Systems of individually controlled neutral atoms, interacting with each other when excited to Rydberg states, have emerged as a promising platform for this task, particularly for the simulation of spin systems. Here, we review the techniques necessary for the manipulation of neutral atoms for the purpose of quantum simulation—such as quantum gas microscopes and arrays of optical tweezers—and explain how the different types of interactions between Rydberg atoms allow a natural mapping onto various quantum spin models. We discuss recent achievements in the study of quantum many-body physics in this platform, and some current research directions beyond that.},
number={2},
journal={Nature Physics},
author={Browaeys, Antoine and Lahaye, Thierry},
year={2020},
month=feb,
pages={132–142},
language={en} }

@article{Cirac95trappedion,
  title = {Quantum Computations with Cold Trapped Ions},
  author = {Cirac, J. I. and Zoller, P.},
  journal = {Phys. Rev. Lett.},
  volume = {74},
  issue = {20},
  pages = {4091--4094},
  numpages = {0},
  year = {1995},
  month = {May},
  publisher = {American Physical Society},
  doi = {10.1103/PhysRevLett.74.4091},
  url = {https://link.aps.org/doi/10.1103/PhysRevLett.74.4091}
}

@article{Leibfried03trappedion,
  title = {Quantum dynamics of single trapped ions},
  author = {Leibfried, D. and Blatt, R. and Monroe, C. and Wineland, D.},
  journal = {Rev. Mod. Phys.},
  volume = {75},
  issue = {1},
  pages = {281--324},
  numpages = {0},
  year = {2003},
  month = {Mar},
  publisher = {American Physical Society},
  doi = {10.1103/RevModPhys.75.281},
  url = {https://link.aps.org/doi/10.1103/RevModPhys.75.281}
}

@article{Devoret13superconducting,
author = {M. H. Devoret  and R. J. Schoelkopf },
title = {Superconducting Circuits for Quantum Information: An Outlook},
journal = {Science},
volume = {339},
number = {6124},
pages = {1169-1174},
year = {2013},
doi = {10.1126/science.1231930},
URL = {https://www.science.org/doi/abs/10.1126/science.1231930}}

@article{Blais21cQED,
  title = {Circuit quantum electrodynamics},
  author = {Blais, Alexandre and Grimsmo, Arne L. and Girvin, S. M. and Wallraff, Andreas},
  journal = {Rev. Mod. Phys.},
  volume = {93},
  issue = {2},
  pages = {025005},
  numpages = {72},
  year = {2021},
  month = {May},
  publisher = {American Physical Society},
  doi = {10.1103/RevModPhys.93.025005},
  url = {https://link.aps.org/doi/10.1103/RevModPhys.93.025005}
}

@article{google25QECbelow,
title={Quantum error correction below the surface code threshold},
volume={638},
rights={2024 The Author(s)},
ISSN={1476-4687},
url={https://www.nature.com/articles/s41586-024-08449-y},
DOI={10.1038/s41586-024-08449-y},
abstractNote={Quantum error correction1–4 provides a path to reach practical quantum computing by combining multiple physical qubits into a logical qubit, in which the logical error rate is suppressed exponentially as more qubits are added. However, this exponential suppression only occurs if the physical error rate is below a critical threshold. Here we present two below-threshold surface code memories on our newest generation of superconducting processors, Willow: a distance-7 code and a distance-5 code integrated with a real-time decoder. The logical error rate of our larger quantum memory is suppressed by a factor of Λ = 2.14 ± 0.02 when increasing the code distance by 2, culminating in a 101-qubit distance-7 code with 0.143% ± 0.003 per cent error per cycle of error correction. This logical memory is also beyond breakeven, exceeding the lifetime of its best physical qubit by a factor of 2.4 ± 0.3. Our system maintains below-threshold performance when decoding in real time, achieving an average decoder latency of 63 microseconds at distance 5 up to a million cycles, with a cycle time of 1.1 microseconds. We also run repetition codes up to distance 29 and find that logical performance is limited by rare correlated error events, occurring approximately once every hour or 3 × 109 cycles. Our results indicate device performance that, if scaled, could realize the operational requirements of large-scale fault-tolerant quantum algorithms.},
number={8052},
journal={Nature},
author={Google Quantum AI and Collaborators},
year={2025},
month=feb,
pages={920–926},
language={en} }

@ARTICLE{Gilbert52,
  author={Gilbert, E. N.},
  journal={The Bell System Technical Journal}, 
  title={A comparison of signalling alphabets}, 
  year={1952},
  volume={31},
  number={3},
  pages={504-522},
  doi={10.1002/j.1538-7305.1952.tb01393.x}}

@article{Varshamov57,
  author = {Varshamov, R. R.},
  title = {Estimate of the number of signals in error correcting codes},
  journal = {Doklady Akademii Nauk SSSR},
  volume = {117},
  pages = {739--741},
  year = {1957}
}

@misc{radebold2025explicitinstancesquantumtanner,
      title={Explicit Instances of Quantum Tanner Codes}, 
      author={Rebecca Katharina Radebold and Stephen D. Bartlett and Andrew C. Doherty},
      year={2025},
      eprint={2508.05095},
      archivePrefix={arXiv},
      primaryClass={quant-ph},
      url={https://arxiv.org/abs/2508.05095}, 
}

@article{ShorLaflamme97,
  title = {Quantum Analog of the MacWilliams Identities for Classical Coding Theory},
  author = {Shor, Peter and Laflamme, Raymond},
  journal = {Phys. Rev. Lett.},
  volume = {78},
  issue = {8},
  pages = {1600--1602},
  numpages = {0},
  year = {1997},
  month = {Feb},
  publisher = {American Physical Society},
  doi = {10.1103/PhysRevLett.78.1600},
  url = {https://link.aps.org/doi/10.1103/PhysRevLett.78.1600}
}

@ARTICLE{Rains99shadow,
  author={Rains, E.M.},
  journal={IEEE Transactions on Information Theory}, 
  title={Quantum shadow enumerators}, 
  year={1999},
  volume={45},
  number={7},
  pages={2361-2366},
  doi={10.1109/18.796376}}

@ARTICLE{AshikhminLitsyu99,
  author={Ashikhmin, A. and Litsyu, S.},
  journal={IEEE Transactions on Information Theory}, 
  title={Upper bounds on the size of quantum codes}, 
  year={1999},
  volume={45},
  number={4},
  pages={1206-1215},
  doi={10.1109/18.761270}}

@misc{github, type={PYTHON},
title={Explicit QLDPC Database},
url={https://github.com/lilyxy/Check-weight-constrained-quantum-codes}, 
author={Andy Zeyi Liu},
year={2026},
month=Jan }

@misc{arnault2025variantBPT,
      title={A Variant of the Bravyi-Terhal Bound for Arbitrary Boundary Conditions}, 
      author={François Arnault and Philippe Gaborit and Wouter Rozendaal and Nicolas Saussay and Gilles Zémor},
      year={2025},
      eprint={2502.04995},
      archivePrefix={arXiv},
      primaryClass={quant-ph},
      url={https://arxiv.org/abs/2502.04995}, 
}

\end{document}